\documentclass{lmcs} %%% last changed 2014-08-20

\usepackage{hyperref}

\usepackage{proof}
\usepackage{cmll}
\usepackage{tikz}
\usetikzlibrary{cd}
\usepackage{multicol}
\usepackage{cite}
\usepackage[T1]{fontenc}
\usepackage{nicefrac}
\usepackage{stmaryrd}
\usepackage{amssymb}

\newcommand\irule[3]{\infer[\mbox{\footnotesize $#3$}]{#2}{#1}}
\newcommand\ket[1]{\ensuremath{|{#1}\rangle}}
\newcommand\abstr[1]{[#1]}
\newcommand\inl{\mbox{\it inl}}
\newcommand\inr{\mbox{\it inr}}
\newcommand\inlr{\mbox{\it inlr}}
\newcommand\inlrpair[2]{\inlr(#1, #2)}
\newcommand\elimtop{\delta_{\top}}
\newcommand\elimbot[1]{\delta^{(#1)}_{\bot}}
\newcommand\elimand{\delta_{\wedge}}
\newcommand\elimor{\delta_{\vee}}
\newcommand\elimorint{{\color{red} \delta_{\vee}}}
\newcommand\elimsup{\delta_{\odot}}
\DeclareRobustCommand{\plus}{%
  \mathbin{%
    \tikz[baseline={([yshift=-0.55ex]current bounding box.center)}]{
      \def\Size{1.4ex}%
      \def\Thick{0.36ex}%
      \fill (-\Thick/2, -\Size/2) rectangle (\Thick/2, \Size/2);
      \fill (-\Size/2, -\Thick/2) rectangle (\Size/2, \Thick/2);
    }%
  }%
}
\newcommand\plusr{\mathrel{\text{\normalfont\scalebox{0.7}{\color{red}$\plus$}}}}
\newcommand\pair[2]{\langle #1, #2 \rangle}
\newcommand\lra{\longrightarrow}
\newcommand\lla{\longleftarrow}
\newcommand\lras{\lra^*}
\newcommand\llas{\mathrel{{}^*{\longleftarrow}}}

\newcommand\Q{{\mathcal Q}}
\newcommand\B{{\mathcal B}}
\newcommand\boolzero{{\mathtt{false}}}
\newcommand\boolone{{\mathtt{true}}}

\newcommand\one{\ensuremath{\mathtt 1}}
\newcommand\elimone{\delta_{\one}}

\newcommand\elimplus{\delta_{\oplus}}

\begin{document}

\title{A new introduction rule for disjunction}
\titlecomment{{\lsuper\dagger}Gilles Dowek passed away in July~2025.  This paper represents one of his last contributions.}
\thanks{The authors want to thank Fr\'ed\'eric Blanqui, Nachum Dershowitz, and Jean-Pierre Jouannaud for their remarks on a previous version of this paper.}

\author[A.~Díaz-Caro]{Alejandro Díaz-Caro\lmcsorcid{0000-0002-5175-6882}}[a,b]
\author[G.~Dowek]{Gilles Dowek\lmcsorcid{0000-0001-6253-935X}}[c,\dagger]

\address{Université de Lorraine, CNRS, Inria, LORIA, France}	
\email{alejandro.diaz-caro@inria.fr} 
\address{Universidad Nacional de Quilmes, Argentina}
\address{Inria and ENS Paris-Saclay, France}

\begin{abstract}
We extend Natural Deduction for intuitionistic logic with a third introduction rule for the disjunction,
$\vee$-i3, with conclusion $\Gamma \vdash A \vee B$, but both premises $\Gamma
\vdash A$ and $\Gamma \vdash B$.  This rule is admissible in Natural Deduction.
This extension is interesting in several respects.  First, it permits to solve
a well-known problem in logics with interstitial rules (that is rules whose all
premises are equal to the conclusion) that have a weak introduction property:
closed cut-free proofs end with an introduction rule, except in the case of
disjunctions.  With this new introduction rule, we recover the strong
introduction property: closed cut-free proofs always end with an introduction.
Second, the termination proof of this proof system is simpler than that of the
usual propositional Natural Deduction with interstitial rules, as it does not
require the use of the so-called \emph{ultra-reduction} rules.  Third, this
proof system, in its linear version, has applications to quantum computing: the
$\vee$-i3 rule enables the expression of quantum measurement, without the cost
of introducing a new connective.  Finally, even in logics without interstitial
rules, the rule $\vee$-i3 is useful to reduce commuting cuts, although, in this
paper, we leave the termination of such reduction as an open problem.
\end{abstract}

\maketitle

\section{Introduction}

In Natural Deduction~\cite{Prawitz}, when we have two proofs $\pi_1$ of $\Gamma
\vdash A$ and $\pi_2$ of $\Gamma \vdash B$, and we want to prove
$\Gamma \vdash A \vee B$, we need to chose one introduction rule of
the disjunction,
\[
  \vcenter{
    \irule{\Gamma \vdash A}
    {\Gamma \vdash A \vee B}
  {\mbox{$\vee$-i1}}}
  \qquad\text{or}\qquad
  \vcenter{
    \irule{\Gamma \vdash B}
    {\Gamma \vdash A \vee B}
    {\mbox{$\vee$-i2}}        
  },
\]
use the corresponding proof, $\pi_1$ or $\pi_2$, and drop the
other. This implies an erasure of information and a break in the
symmetry between the proofs $\pi_1$ and $\pi_2$, just like Buridan's
donkey~\cite{wiki} needs to break the symmetry and chose one pile of
hay. 
To avoid this erasure of information and break in symmetry, we
extend Natural Deduction with a third introduction rule for the
disjunction
\[
  \irule{\Gamma \vdash A & \Gamma \vdash B}
  {\Gamma \vdash A \vee B}
  {\mbox{$\vee$-i3}}
\]
that is admissible in the usual calculus.
A similar rule appears in logics formalising the notion of \emph{proof-why} \cite{Poggiolesi2016,Genco2021},
where three rules permit to deduce that $A \vee B$
is true: one being our $\vee$-i3 rule, and the two others being related to,
but more complex than, $\vee$-i1 and $\vee$-i2.

This new introduction rule yields the new form of cut
\[
    \irule{\irule{\irule{\pi_1}
	{\Gamma \vdash A}
	{}
	&
	\irule{\pi_2}
	{\Gamma \vdash B}
	{}
      }
      {\Gamma \vdash A \vee B}
      {\mbox{$\vee$-i3}}
      &
      \irule{\pi_3}
      {\Gamma, A\vdash C}
      {}
      &
      \irule{\pi_4}
      {\Gamma, B\vdash C}
      {}
    }
    {\Gamma \vdash C}
    {\mbox{$\vee$-e}}
\]
that can then be reduced in various ways: either always to
$(\pi_1/A)\pi_3$
(that is the proof $\pi_3$ where the axiom rule with the proposition
$A$ is replaced with the proof $\pi_1$),
or always to $(\pi_2/B)\pi_4$, implying again an
erasure of information and a break in symmetry, or to one of these
proofs in a non-deterministic way, symmetry being preserved at the
cost of the loss of determinism. As we shall see, other options are
possible, where both proofs are kept.

This extension of Natural Deduction is interesting in several
respects.  First, it permits to solve a well-known problem in calculi
with an interstitial rule (that is a rule
whose all premises are equal to the conclusion, as the one below)
where the proof-term language is enriched with a $\plus$
constructor (sometimes written $\parallel$): 
\[
  \irule{\Gamma \vdash A &\Gamma \vdash A}
  {\Gamma \vdash A}
  {\mbox{sum.}}
\]

This constructor is sometimes seen as putting two
processes in parallel. The two proofs of the premises are then seen
as two processes of the same type, and the proof of the conclusion as
a process formed by putting these two processes in parallel
\cite{DanosKrivine}. The contexts of the two premises and the
conclusion are sometimes slightly different to model communication between
the processes 
\cite{Aschieri16,AschieriCG17,AschieriCG20}. Here, as in
\cite{ArrighiDiazcaroLMCS12,DiazcaroDowekMSCS24,LairdManzonettoMcCuskerPaganiLICS13,Vaux2009},
it is simply seen as an algebraic sum.

Such a rule can create an interstice between the introduction rules
and the elimination rule of a connective, introducing new forms of
commuting cuts, for example
\[
    \irule{\irule{\irule{\irule{\pi_1}
	  {\Gamma, A \vdash B}
	  {}
	}
	{\Gamma \vdash A \Rightarrow B}
	{\mbox{$\Rightarrow$-i}}
	&\irule{\irule{\pi_2}
	  {\Gamma, A \vdash B}
	  {}
	}
	{\Gamma \vdash A \Rightarrow B}
	{\mbox{$\Rightarrow$-i}}
      }
      {\Gamma \vdash A \Rightarrow B}
      {\mbox{sum}}
      & \irule{\pi_3}
      {\Gamma \vdash A}
      {}
    }
    {\Gamma \vdash B}
    {\mbox{$\Rightarrow$-e.}}
\]

One way to reduce such a cut is to drop one of the proofs of $\Gamma
\vdash A \Rightarrow B$, implying again an erasure of information and
a break in symmetry, yielding, for example, the proof
\[
  \irule{\irule{\irule{\pi_1}
      {\Gamma, A \vdash B}
      {}
    }
    {\Gamma \vdash A \Rightarrow B}
    {\mbox{$\Rightarrow$-i}}
    & \irule{\pi_3}
    {\Gamma \vdash A}
    {}
  }
  {\Gamma \vdash B}
  {\mbox{$\Rightarrow$-e}}
\]
that reduces with the usual reduction rules. To preserve information
and symmetry, we can, instead, commute the sum rule, either with
the introduction rules above,
\[
  \irule{\irule{\irule{\irule{\pi_1}
	{\Gamma, A \vdash B}
	{}
	&\irule{\pi_2}
	{\Gamma, A \vdash B}
	{}
      }
      {\Gamma, A \vdash B}
      {\mbox{sum}}
    }
    {\Gamma \vdash A \Rightarrow B}
    {\mbox{$\Rightarrow$-i}}
    & \irule{\pi_3}
    {\Gamma \vdash A}
    {}
  }
  {\Gamma \vdash B}
  {\mbox{$\Rightarrow$-e,}}
\]
or with the elimination rule below
\[
  \irule{\irule{\irule{\irule{\pi_1}
	{\Gamma, A \vdash B}
	{}
      }
      {\Gamma \vdash A \Rightarrow B}
      {\mbox{$\Rightarrow$-i}}
      &\irule{\pi_3}
      {\Gamma \vdash A}
      {}
    }
    {\Gamma \vdash B}
    {\mbox{$\Rightarrow$-e}}
    &\irule{\irule{\irule{\pi_2}
	{\Gamma, A \vdash B}
	{}
      }
      {\Gamma \vdash A \Rightarrow B}
      {\mbox{$\Rightarrow$-i}}
      &\irule{\pi_3}
      {\Gamma \vdash A}
      {}
    }
    {\Gamma \vdash B}
    {\mbox{$\Rightarrow$-e.}}
  }
  {\Gamma \vdash B}
  {\mbox{sum}}
\]
Both derivations reduce, under the usual reduction rules, to the proof
\[
  \irule{\irule{(\pi_3/A)\pi_1}
    {\Gamma \vdash B}
    {}
    &\irule{(\pi_3/A)\pi_2}
    {\Gamma \vdash B}
    {}
  }
  {\Gamma \vdash B}
  {\mbox{sum.}}
\]
          
Commuting the sum
rule with the introduction rules yields better properties of the
irreducible proofs, in particular a stronger introduction property: a
closed irreducible proof is always an introduction. It should
therefore be preferred.

But, in the usual Natural Deduction, such a commutation is not
possible for the disjunction, as there seems to be no way to commute
the sum rule with the introduction rules in the proof
\[
  \irule{\irule{\irule{\pi_1}
      {\Gamma \vdash A}
      {}
    }
    {\Gamma \vdash A \vee B}
    {\mbox{$\vee$-i1}}
    &\irule{\irule{\pi_2}
      {\Gamma \vdash B}
      {}
    }
    {\Gamma \vdash A \vee B}
    {\mbox{$\vee$-i2}}
  }
  {\Gamma \vdash A \vee B}
  {\mbox{sum.}}
\]

Thus, there seems to be no way to reduce the proof 
\[
  \irule{\irule{\irule{\irule{\pi_1}
	{\Gamma \vdash A}
	{}
      }
      {\Gamma \vdash A \vee B}
      {\mbox{$\vee$-i1}}
      &\irule{\irule{\pi_2}
	{\Gamma \vdash B}
	{}
      }
      {\Gamma \vdash A \vee B}
      {\mbox{$\vee$-i2}}
    }
    {\Gamma \vdash A \vee B}
    {\mbox{sum}}
    &\irule{\pi_3}{\Gamma, A \vdash C}{}
    &\irule{\pi_4}{\Gamma, B \vdash C}{}
  }
  {\Gamma \vdash C}
  {\mbox{$\vee$-e}}
\]
by commuting the sum rule with the introduction rules.  This has lead,
for example in \cite{DiazcaroDowekTCS23}, to a mixed system, where the
sum rule commutes with the introduction rules when possible, but only with
the elimination rule of the disjunction.
For instance, the proof above reduces to
\[
  \irule{\irule{\irule{\irule{\pi_1}
	{\Gamma \vdash A}
	{}
      }
      {\Gamma \vdash A \vee B}
      {\mbox{$\vee$-i1}}
      &\irule{\pi_3}{\Gamma, A \vdash C}{}
      &\irule{\pi_4}{\Gamma, B \vdash C}{}
    }
    {\Gamma \vdash C}
    {\mbox{$\vee$-e}}
    &
    \irule{\irule{\irule{\pi_2}
	{\Gamma \vdash B}
	{}
      }
      {\Gamma \vdash A \vee B}
      {\mbox{$\vee$-i2}}
      &\irule{\pi_3}{\Gamma, A \vdash C}{}
      &\irule{\pi_4}{\Gamma, B \vdash C}{}
    }
    {\Gamma \vdash C}
    {\mbox{$\vee$-e}}
  }
  {\Gamma \vdash C}
  {\mbox{sum}}
\]
and then, with the usual reduction rules, to
\[
  \irule{(\pi_1/A)\pi_3 & (\pi_2/B)\pi_4}
  {\Gamma \vdash C}
  {\mbox{sum.}}    
\]

The $\vee$-i3 rule makes it possible to obtain this same proof
by first commuting the sum rule with the introduction rules,
reducing this proof to
the proof 
\[
  \irule{\irule{\irule{\pi_1}
      {\Gamma \vdash A}
      {}
      &\irule{\pi_2}
      {\Gamma \vdash B}
      {}
    }
    {\Gamma \vdash A \vee B}
    {\mbox{$\vee$-i3}}
    &\irule{\pi_3}{\Gamma, A \vdash C}{}
    &\irule{\pi_4}{\Gamma, B \vdash C}{}
  }
  {\Gamma \vdash C}
  {\mbox{$\vee$-e}}
\]
where the sum rule has temporarily disappeared, and that can be
reduced.  If we decide to reduce this proof neither to
$(\pi_1/A)\pi_3$ nor to $(\pi_2/B)\pi_4$ but to
\[
  \irule{(\pi_1/A)\pi_3 & (\pi_2/B)\pi_4}
  {\Gamma \vdash C}
  {\mbox{sum,}}
\]
the sum rule reappears in the reduct. This way, we get a calculus with a
sum rule and a strong introduction property.

Second, to our surprise, the termination proof of this calculus is
simpler than that of the usual propositional Natural Deduction with
interstitial rules: since we never commute interstitial rules with
elimination rules, the termination proof does not need
to extend the reduction relation with rules reducing $t_1 \plus t_2$
to $t_1$ and $t_2$ (the 
so-called \emph{ultra-reduction} relation) in order to obtain the termination
of the original relation.

Third, several papers have introduced new connectives: $\Delta$
\cite{SatoshiIPSJ96}, $|$ \cite{TzouvarasIGPL17, TzouvarasIGPL19},
$\odot$ \cite{DiazcaroDowekTCS23}, to model non-determinism,
information-erasure, quantum-measurement, etc. Here, we model
non-determinism, information-erasure, and
quantum-measurement, without the cost of introducing a new connective.

Finally, even in calculi without interstitial rules, the rule
$\vee$-i3 shows a new way to reduce commuting cuts: instead of commuting the
blocking rules---which lie between the introduction and the elimination rule of
the commuting cut---with the elimination rules, they are commuted with the
introduction rules.

In this paper, we first present the in-left-right-+-calculus, the
proof-term language of the extension of Natural Deduction with the
rules $\vee$-i3 and sum, and in particular its termination proof
(\autoref{sec:inlrcalculus}).  We then discuss applications of
this calculus to quantum computing, by introducing the quantum in-left-right-calculus, building on
\cite{DiazcaroDowekTCS23,DiazcaroDowekMSCS24} (\autoref{sec:quantum}).  Finally, we introduce, the
in-left-right-calculus, a calculus with a $\vee$-i3 rule but no
interstitial rules, and discuss this idea of commuting the blocking
rules with the introduction rules rather than with the elimination
rules, in usual commuting cuts (\autoref{sec:commutingcuts}).

\subsection*{Related work}
In~\cite{SatoshiIPSJ96}, a new connective $\Delta$ is introduced to
express non-deter\-min\-ism in the context of multiplicative additive
linear logic and proof nets. This connective is self-dual which
suggests that it is a kind of additive disjunction and conjunction at
the same time. This produces non-determinism in the cut elimination.

An attempt to explore quantum superposition in logic is found in
Propositional Superposition Logic~\cite{TzouvarasIGPL17}. This logic
introduces the binary connective $|$ to capture the logical interplay
between conjunction and disjunction through the notion of
superposition. This framework was subsequently extended to First-Order
Superposition Logic~\cite{TzouvarasIGPL19}, which adapts these
ideas to a quantified logic setting.

In \cite{DiazcaroDowekTCS23} a new connective $\odot$ for
superposition is also introduced. It has the introduction rule of the
conjunction (similar to $\vee$-i3) and the elimination rule of the
disjunction ($\vee$-e), but no other introduction rules.  This $\odot$
connective is introduced in the context of quantum computing, and its
intuitionistic linear logic version discussed
in~\cite{DiazcaroDowekMSCS24}. The introduction of $\odot$ is the same
as that of $\with$ (or $\wedge$), the same as $\Delta$. However,
$\odot$ has an elimination rule which is the same as that of $\oplus$ (or $\vee$),
which renders it non-deterministic.

All these works introduce a new connective, whereas the logic in this paper
does not.

Among these papers \cite{SatoshiIPSJ96} and \cite{DiazcaroDowekTCS23}
have an explicit proof language  but not \cite{TzouvarasIGPL17} and
\cite{TzouvarasIGPL19}.
Among those that have an explicit proof language, 
\cite{DiazcaroDowekTCS23} introduces a sum rule, but not
\cite{SatoshiIPSJ96}.
In addition, these works do not address the problem of commuting cuts.

Independently, in~\cite{MogbilFOPARA09} a sum rule is added to multiplicative linear logic to model non-deterministic choices and introduce parallel reduction strategies for proof nets.  However, since its objectives are different, it does not includes additives and thus the relation of the rule sum with the additive disjunction is not treated.

The idea of having a disjunction with an introduction rule of the
conjunction seems natural in models of linear logic using monoidal
categories with biproducts~\cite{BarrMSCS91}, that is, where the
product (usually used to interpret the additive conjunction) and the
co-product (usually used to interpret the additive disjunction),
coincide. Indeed, we can stick together the diagrams of the universal
properties for product and co-product, and we have the mediating arrow
of the co-product with a ``product'' as domain.
\[
  \begin{tikzcd}[column sep=2cm,row sep=8mm]
    &  C\ar[d,"{\pair {f_1}{f_2}}"]\ar[dl,"{f_1}",sloped]\ar[dr,"{f_2}",sloped] \\
     A\ar[dr,"{g_1}"',sloped]\ar[r,"{i_1}"',bend right=10] & {A\oplus B}\ar[d,"{[g_1,g_2]}"]\ar[r,"{\pi_2}",bend left=10]\ar[l,"{\pi_1}"',bend right=10] &  B\ar[l,"{i_2}",bend left=10]\ar[dl,"{g_2}"',sloped]\\
    &  D
  \end{tikzcd}
\]
This has been recently further developed
in~\cite{DiazcaroMalherbeArxiv23}.

However, in the in-left-right-+-calculus (the framework presented in
\autoref{sec:inlrcalculus}), the
interpretation of disjunction with this third introduction rule is not as
straightforward, since we are not in a linear context.
A recent work~\cite{DiazcaroMalherbeFSTTCS25} builds on an early development
of our proposal to introduce a novel approach to interpreting
disjunction in the category Set. There, the disjunction $A \vee B$ is
interpreted not as $A\uplus B$ but as $A
\uplus B \uplus (A \times B)$. This idea extends to logics with a sum
rule, using magmas instead of sets.  Interestingly, this construction
does not yield a coproduct, unless arrows are homomorphisms---a
condition that is overly restrictive in this context, as not all
proofs can be interpreted. Thus, this framework offers a new
perspective on the interpretation of the disjunction.

\section{The in-left-right-+-calculus}
\label{sec:inlrcalculus}

\subsection{Proof-terms}

We consider propositional logic with the usual connectives
$\top$, $\bot$, $\Rightarrow$, $\wedge$, and $\vee$.
We introduce a term language, the in-left-right-+-calculus, for the proofs of
propositional Natural Deduction extended with the rules sum and
$\vee$-i3. Its syntax is
\begin{align*}
  t =~ x & \mid t \plus t
   \mid \star \mid \elimtop(t,t) \mid \elimbot{A}(t)
  \mid \lambda \abstr{x}t\mid t~t
  \\
  &\mid \pair{t}{t} \mid \elimand^1(t,\abstr{x}t)
  \mid \elimand^2(t,\abstr{x}t)
  \mid \inl(t) \mid \inr(t) \mid \inlr(t,t)
  \mid \elimor(t,\abstr{x}t,\abstr{x}t)
\end{align*}

The proofs of the form $\star$, $\lambda \abstr{x}t$, $\pair{t}{u}$,
$\inl(t)$, $\inr(t)$, and $\inlr(t,u)$ are called \emph{introductions},
and those of the form $\elimtop(t,u)$, $\elimbot{A}(t)$, $t~u$,
$\elimand^1(t,\abstr{x}u)$, $\elimand^2(t,\abstr{x}u)$, and
$\elimor(t,\abstr{x}u,\abstr{y}v)$ \emph{eliminations}.  The variables
and the proofs of the form $t \plus u$ are neither introductions nor
eliminations. As usual, we make the conclusion $A$ explicit in the
proof $\elimbot A(t)$.

The $\alpha$-equivalence relation and the free and bound variables of
a proof-term are defined as usual. Proof-terms are defined modulo
$\alpha$-equivalence.  A proof-term is closed if it contains no free
variables.  We write $(u/x)t$ for the substitution of $u$ for $x$ in
$t$.
If $FV(t) \subseteq \{x\}$, we may use the notation $t\{u\}$ for $(u/x)t$.

\begin{figure}[t]
  $$
    \irule{}
    {\Gamma \vdash x:A}
    {\mbox{axiom~$x:A \in \Gamma$}}
    \qquad
    \irule{\Gamma \vdash t:A & \Gamma \vdash u:A}
    {\Gamma \vdash t \plus u:A}
    {\mbox{sum}}
  $$
  $$
    \irule{}
    {\Gamma \vdash \star:\top}
    {\mbox{$\top$-i}}
    \qquad
    \irule{\Gamma \vdash t:\top & \Gamma \vdash u:C}
    {\Gamma \vdash \elimtop(t,u):C}
    {\mbox{$\top$-e}}
    \qquad
    \irule{\Gamma \vdash t:\bot}
    {\Gamma \vdash \elimbot{C}(t):C}
    {\mbox{$\bot$-e}}
  $$
  $$
    \irule{\Gamma, x:A \vdash t:B}
    {\Gamma \vdash \lambda \abstr{x}t:A \Rightarrow B}
    {\mbox{$\Rightarrow$-i}}
    \qquad
    \irule{\Gamma \vdash t:A\Rightarrow B & \Gamma \vdash u:A}
    {\Gamma \vdash t u:B}
    {\mbox{$\Rightarrow$-e}}
    \qquad
    \irule{\Gamma \vdash t:A & \Gamma \vdash u:B}
    {\Gamma \vdash \pair{t}{u}:A \wedge B}
    {\mbox{$\wedge$-i}}
  $$
  $$
    \irule{\Gamma \vdash t:A \wedge B & \Gamma, x:A \vdash u:C}
    {\Gamma \vdash \elimand^1(t,\abstr{x}u):C}
    {\mbox{$\wedge$-e1}}
    \qquad
    \irule{\Gamma \vdash t:A \wedge B & \Gamma, x:B \vdash u:C}
    {\Gamma \vdash \elimand^2(t,\abstr{x}u):C}
    {\mbox{$\wedge$-e2}}
  $$
  $$
    \irule{\Gamma \vdash t:A}
    {\Gamma \vdash \inl(t):A \vee B}
    {\mbox{$\vee$-i1}}
    \qquad
    \irule{\Gamma \vdash u:B}
    {\Gamma \vdash \inr(u):A \vee B}
    {\mbox{$\vee$-i2}}
    \qquad
    \irule{\Gamma \vdash t:A & \Gamma \vdash u:B}
    {\Gamma \vdash \inlr(t,u):A \vee B}
    {\mbox{$\vee$-i3}}
  $$
  $$
    \irule{\Gamma \vdash t:A \vee B & \Gamma, x:A \vdash u:C & \Gamma, y:B \vdash v:C}
    {\Gamma \vdash \elimor(t,\abstr{x}u,\abstr{y}v):C}
    {\mbox{$\vee$-e}}
  $$
  \caption{The typing rules of the in-left-right-+-calculus\label{figtypingrules}}
\end{figure}

\begin{figure}[t]
  \centering
  \parbox{0.45\textwidth}{
  \begin{align}
    \elimtop(\star, t) & \longrightarrow t \label{ruelimtop}\\
    (\lambda \abstr{x}t)~u & \longrightarrow  (u/x)t \label{rubeta}\\
    \elimand^1(\pair{t}{u}, \abstr{x}v) & \longrightarrow  (t/x)v \label{ruelimand1}
  \end{align}
}\quad \parbox{0.45\textwidth}{
  \begin{align}
    \elimand^2(\pair{t}{u}, \abstr{x}v) & \longrightarrow  (u/x)v \label{ruelimand2}\\
    \elimor(\inl(t),\abstr{x}v,\abstr{y}w) & \longrightarrow  (t/x)v \label{ruelimorinl}\\
    \elimor(\inr(u),\abstr{x}v,\abstr{y}w) & \longrightarrow (u/y)w \label{ruelimorinr}
  \end{align}}\vspace{-1\baselineskip}

  \parbox{0.6\textwidth}{
  \begin{align}
    \elimor(\inlr(t,u),\abstr{x}v,\abstr{y}w) & \longrightarrow  (t/x)v \plus (u/y)w 
    \label{ruelimorinlr1}
  \end{align}
}\vspace{-1\baselineskip}

  \hspace{-1cm}\parbox{0.5\textwidth}{
    \begin{align}
      \star \plus \star&\longrightarrow  \star \label{rusumstar}\\
      (\lambda \abstr{x}t) \plus (\lambda \abstr{x}u) & \longrightarrow  \lambda \abstr{x}(t \plus u) \label{rusumlam}\\
      \pair{t}{u} \plus \pair{v}{w} & \longrightarrow  \pair{t \plus v}{u \plus w} \label{rusumpair}\\
      \inl(t) \plus \inl(v) & \longrightarrow  \inl(t \plus v) \label{rusuminl}\\
      \inl(t) \plus \inr(w) & \longrightarrow  \inlr(t,w) \label{rusuminlinr} \\
      \inl(t) \plus \inlr(v,w) & \longrightarrow  \inlr(t \plus v,w) \label{rusuminlinlr}
    \end{align}
  }\parbox{0.56\textwidth}{
    \begin{align}
      \inr(u) \plus \inl(v) & \longrightarrow  \inlr(v,u) \label{rusuminrinl}\\
      \inr(u) \plus \inr(w) & \longrightarrow  \inr(u \plus w) \label{rusuminr}\\
      \inr(u) \plus \inlr(v,w) & \longrightarrow  \inlr(v,u \plus w) \label{rusuminrinlr}\\
      \inlr(t,u) \plus \inl(v) & \longrightarrow  \inlr(t \plus v,u) \label{rusuminlrinl}\\
      \inlr(t,u) \plus \inr(w) & \longrightarrow  \inlr(t,u \plus w) \label{rusuminlrinr}\\
      \inlr(t,u) \plus \inlr(v,w) & \longrightarrow  \inlr(t \plus v,u \plus w) \label{rusuminlr}
    \end{align}
  }
  \caption{The reduction rules of the in-left-right-+-calculus\label{figureductionrules}}
\end{figure}
The typing rules of the in-left-right-+-calculus are given in~\autoref{figtypingrules}.  The elimination rules for conjunction
are the so-called generalized elimination rules, but the usual
elimination rules could have been chosen instead. Indeed, 
$\pi_1(t)$ can be encoded as $\elimand^1(t,\abstr{x}x)$, and $\elimand^1(t,\abstr{x}u)$ can be encoded as $(\lambda x.u)\pi_1(t)$.
In the same way, we
have the generalized elimination rule for $\top$, but the usual rules,
with an introduction, but no elimination, could have been chosen
instead. However, we will reap the benefits of this choice in 
\autoref{sec:quantum}.
The elimination rule of implication, in contrast, is the simple
application rule, though the generalized elimination rule could also
have been chosen. In this paper, these choices are essentially
contingent. The reduction rules are given in
\autoref{figureductionrules}, they are the expression on
proof-terms of the reduction rules presented above.

The proof $\inl(t) \plus \inr(u)$, reduces with the rule \eqref{rusuminlinr} to $\inlr(t,u)$, that is,
\[
  \vcenter{
    \irule{\irule{\irule{t}
	{\Gamma \vdash A}
	{}
      }
      {\Gamma \vdash A \vee B}
      {\mbox{$\vee$-i1}}
      &\irule{\irule{u}
	{\Gamma \vdash B}
	{}
      }
      {\Gamma \vdash A \vee B}
      {\mbox{$\vee$-i2}}
    }
    {\Gamma \vdash A \vee B}
    {\mbox{sum}}
  }
  \qquad
  \lra
  \qquad
  \vcenter{
    \irule{\irule{t}
      {\Gamma \vdash A}
      {}
      &
      \irule{u}
      {\Gamma \vdash B}
      {}
    }
    {\Gamma \vdash A \vee B}
    {\mbox{$\vee$-i3.}}
  }
\]
And similar rules reduce the other sums of
introductions of the disjunction. For instance, with the rule
\eqref{rusuminlr}, the proof $\inlr(t_1,t_2) \plus \inlr(u_1,u_2)$
reduces to $\inlr(t_1 \plus u_1, t_2 \plus u_2)$.

With the rule \eqref{ruelimorinlr1}, the cut
$$
  \irule{\irule{\irule{t}
      {\Gamma \vdash A}
      {}
      &\irule{u}
      {\Gamma \vdash B}
      {}
    }
    {\Gamma \vdash A \vee B}
    {\mbox{$\vee$-i3}}
    &\irule{v}{\Gamma, A \vdash C}{}
    &\irule{w}{\Gamma, B \vdash C}{}
  }
  {\Gamma \vdash C}
  {\mbox{$\vee$-e}}
$$
that is 
$\elimor(\inlr(t,u), \abstr{x}v, \abstr{y}w)$,
reduces neither to $(t/x)v$ nor to $(u/y)w$, but to the sum of
these two proofs, preserving information and symmetry.

\begin{rem}
The in-left-right-+-calculus does not verify the so-called ``harmony
principle''
\cite{Gentzen,Prawitz,PrawitzEssay,Dummett,SchroederHeister,SchroederHeister2014,MillerPimentel,Read04,Read10,Read,MartinLof}. Indeed,
as we have three introduction rules
for disjuntion,
\[
  \irule{\Gamma \vdash t:A}
  {\Gamma \vdash \inl(t):A \vee B}
  {\mbox{$\vee$-i1,}}
  \qquad
  \irule{\Gamma \vdash u:B}
  {\Gamma \vdash \inr(u):A \vee B}
  {\mbox{$\vee$-i2, and}}
  \qquad
  \irule{\Gamma \vdash t:A & \Gamma \vdash u:B}
  {\Gamma \vdash \inlr(t,u):A \vee B}
  {\mbox{$\vee$-i3,}}
\]
a harmonious elimination rule would, for instance, have three minor premises
\[
  \irule{\Gamma \vdash t:A \vee B
    &
      \Gamma, x:A \vdash u:C &
      \Gamma, y:B \vdash v:C &
      \Gamma, x:A, y:B \vdash w:C
  }
  {\Gamma \vdash \elimor(t,\abstr{x}u,\abstr{y}v,\abstr{xy}w):C}
  {\mbox{$\vee$-e}}
\]
and the associated reduction rule would be
\(
  \elimor(\inlr(t_1,t_2),\abstr{x}u,\abstr{y}v,\abstr{xy} w) \longrightarrow
  (t_1/x,t_2/y)w
\).
Instead, we have the more specialized rule \eqref{ruelimorinlr1} where
the proof $w$ is $u \plus v$ by construction.  The choice of this
proof $w$ and its relation to the proofs $u$ and $v$ seems to be the
fine-tuning parameter, when designing a calculus with a $\vee$-i3
rule.
\end{rem}

\subsection{Properties}
\label{subsec:subrefd}

The in-left-right-+-calculus enjoys the subject reduction property (\autoref{thm:SR}), 
the introduction property (\autoref{introductions}), 
the termination property (\autoref{th:Term}), and
the confluence property (\autoref{th:Confluence})
(as it is left linear and has no critical pairs), 
as stated by the following theorems.  

\subsubsection{Subject reduction}

As usual, we start with a substitution lemma
(\autoref{prop:subst}) which is used in
the proof of subject reduction
(\autoref{thm:SR}).

The substitution lemma says that substituting a hypothesis in proof, by an
actual proof of that hypothesis, yields a valid proof.
\begin{prop}
  [Substitution]
  \label{prop:subst}
  Let $\Gamma,x:B\vdash t:A$ and $\Gamma\vdash u:B$.
  Then, $\Gamma\vdash (u/x)t:A$.
\end{prop}
\begin{proof}
  By induction on the structure of $t$.
  \begin{itemize}
    \item Let $t=x$, thus $B=A$ and $(u/x)t = u$, and since $u$ is a proof of
      $B$ in context $\Gamma$, we have that $(u/x)t$ is a proof of $A$ in
      context $\Gamma$.
    \item Let $t=y$, thus $y:A\in\Gamma$ and $(u/x)t=y$, hence $(u/x)t$ is a
      proof of $A$ in context $\Gamma$.
    \item Let $t=v \plus w$, 
      thus $(u/x)t=(u/x)v\plus (u/x)w$ and
      $v$ and $w$ are proofs of $A$ in context $\Gamma,x:B$. Hence, by the
      induction hypothesis,
      $(u/x)v$ is a proof of $A$ in context $\Gamma$ and
      $(u/x)w$ is a proof of $A$ in context $\Gamma$. Therefore, 
      $(u/x)t$ is a proof of $A$ in context $\Gamma$.
    \item Let $t=\star$, thus $A=\top$ and $(u/x)t=\star$, hence $(u/x)t$ is a
      proof of $A$ in context $\Gamma$.
    \item Let $t= \elimtop(v,w)$,  
      thus $(u/x)t=\elimtop((u/x)v,(u/x)w)$ and
      $v$ is a proof of $\top$ in the context $\Gamma,x:B$ and $w$ is a proof of $A$ in context $\Gamma,x:B$. Hence, by the induction
      hypothesis,
      $(u/x)v$ is a proof of $\top$ in context $\Gamma$ and
      $(u/x)w$ is a proof of $A$ in context $\Gamma$. Therefore, 
      $(u/x)t$ is a proof of $A$ in context $\Gamma$.
    \item Let $t= \elimbot{A}(v)$, 
      thus $(u/x)t=\elimbot{A}((u/x)v)$ and
      $v$ is a proof of $\bot$ in context $\Gamma,x:B$.
      Hence, by the induction hypothesis,
      $(u/x)v$ is a proof of $\bot$ in context $\Gamma$. Therefore,
      $(u/x)t$ is a proof of $A$ in context $\Gamma$.
    \item Let $t= \lambda \abstr{y}v$,
      thus $(u/x)t=\lambda\abstr{y}((u/x)v)$, $A=C\Rightarrow D$, and
      $v$ is a proof of $D$ in context $\Gamma,x:B,y:C$.
      Hence, by the induction hypothesis,
      $(u/x)v$ is a proof of $D$ in context $\Gamma,y:C$. Therefore,
      $(u/x)t$ is a proof of $A$ in context $\Gamma$.
    \item Let $t=v~w$,
      thus $(u/x)t=(u/x)~v(u/x)w$ and
      $v$ is a proof of $C\Rightarrow A$ in the context $\Gamma,x:B$ and $w$ is a proof of $C$ in context $\Gamma,x:B$. Hence, by the induction
      hypothesis,
      $(u/x)v$ is a proof of $C\Rightarrow A$ in context $\Gamma$ and
      $(u/x)w$ is a proof of $C$ in context $\Gamma$. Therefore, 
      $(u/x)t$ is a proof of $A$ in context $\Gamma$.
    \item Let $t= \pair{v}{w}$,
      thus $(u/x)t=\pair{(u/x)v}{(u/x)w}$, $A=C\wedge D$, 
      $v$ is a proof of $C$ in the context $\Gamma,x:B$, 
      and $w$ is a proof of $D$ in context $\Gamma,x:B$. Hence, by the
      induction hypothesis,
      $(u/x)v$ is a proof of $C$ in context $\Gamma$ and
      $(u/x)w$ is a proof of $D$ in context $\Gamma$. Therefore, 
      $(u/x)t$ is a proof of $A$ in context $\Gamma$.
    \item Let $t= \elimand^1(v,\abstr{y}w)$,
      thus $(u/x)t=\elimand^1((u/x)v,\abstr{y}{(u/x)w})$,
      $v$ is a proof of $C\wedge D$ in the context $\Gamma,x:B$, 
      and $w$ is a proof of $A$ in context $\Gamma,x:B,y:C$. Hence, by the
      induction hypothesis,
      $(u/x)v$ is a proof of $C\wedge D$ in context $\Gamma$ and
      $(u/x)w$ is a proof of $A$ in context $\Gamma,y:C$. Therefore, 
      $(u/x)t$ is a proof of $A$ in context $\Gamma$.
    \item Let $t= \elimand^2(v,\abstr{y}w)$. This case is  analogous to the previous case.
    \item Let $t= \inl(v)$,
      thus $(u/x)t=\inl((u/x)v)$, $A=C\vee D$, and
      $v$ is a proof of $C$ in context $\Gamma,x:B$.
      Hence, by the induction hypothesis,
      $(u/x)v$ is a proof of $C$ in context $\Gamma,y:C$. Therefore,
      $(u/x)t$ is a proof of $A$ in context $\Gamma$.
    \item Let $t= \inr(v)$. This case is  analogous to the previous case.
    \item Let $t= \inlr(v,w)$,
      thus $(u/x)t=\inlr({(u/x)v},{(u/x)w})$, $A=C\vee D$, 
      $v$ is a proof of $C$ in the context $\Gamma,x:B$, 
      and $w$ is a proof of $D$ in context $\Gamma,x:B$. Hence, by the
      induction hypothesis,
      $(u/x)v$ is a proof of $C$ in context $\Gamma$ and
      $(u/x)w$ is a proof of $D$ in context $\Gamma$. Therefore, 
      $(u/x)t$ is a proof of $A$ in context $\Gamma$.
    \item Let $t= \elimor(v,\abstr{y_1}w_1,\abstr{y_2}w_2)$,
      thus $(u/x)t=\elimor({(u/x)v},\abstr{y_1}{(u/x)w_1},\abstr{y_2}w_2)$, 
      $v$ is a proof of $C\vee D$ in the context $\Gamma,x:B$, 
      $w_1$ is a proof of $A$ in context $\Gamma,x:B,y_1:C$, and
      $w_2$ is a proof of $A$ in context $\Gamma,x:B,y_2:D$.
      Hence, by the induction hypothesis,
      $(u/x)v$ is a proof of $C\vee D$ in context $\Gamma$,
      $(u/x)w_1$ is a proof of $A$ in context $\Gamma,y_1;C$, and
      $(u/x)w_2$ is a proof of $A$ in context $\Gamma,y_2;D$.
      Therefore, $(u/x)t$ is a proof of $A$ in context $\Gamma$.
      \qedhere
  \end{itemize}
\end{proof}
\begin{thm}[Subject reduction]
\label{thm:SR}
  Let $\Gamma\vdash t:A$ and $t\lra u$, then $\Gamma\vdash u:A$.
\end{thm}
\begin{proof}
  By induction on the reduction relation. The inductive cases are trivial, so we
  only treat the basic cases corresponding to rules~\eqref{ruelimtop}
  to~\eqref{rusuminlr} from \autoref{figureductionrules}.
  \begin{enumerate}
    \item Let $t=\elimtop(\star, v)$ and $u=v$. Then, by inversion $u$ is a
      proof of $A$ in the context $\Gamma$.
    \item Let $t=(\lambda \abstr{x}v)~w$ and $u=(w/x)v$. Then $v$ is a proof
      of $A$ in the context $\Gamma,x:B$ and $w$ is a proof of $B$ in the
      context $\Gamma$. Thus, by \autoref{prop:subst}, we have that
      $(w/x)v$ is a proof of $A$ in the context $\Gamma$.
    \item Let $t=\elimand^1(\pair{v_1}{v_2}, \abstr{x}w)$ and $u=(v_1/x)w$.
      Then $v_1$ is a proof of $B$ in the context $\Gamma$ and $w$ is a proof
      of $A$ in the context $\Gamma,x:B$. Thus, by
      \autoref{prop:subst}, we have that $(v_1/x)w$ is a proof of $A$
      in the context $\Gamma$.
    \item Let $t=\elimand^2(\pair{v_1}{v_2}, \abstr{x}w)$ and $u=(v_2/x)w$.
      This case is analogous to the previous case.
    \item Let $t=\elimor(\inl(v),\abstr{x}w_1,\abstr{y}w_2)$ and $u=(v/x)w_1$.
      Then $v$ is a proof of $B$ in the context $\Gamma$ and $w_1$ is a proof
      of $A$ in the context $\Gamma,x:B$. Thus, by
      \autoref{prop:subst}, we have that $(v/x)w_1$ is a proof of $A$
      in the context $\Gamma$.
    \item Let $t=\elimor(\inr(v),\abstr{x}w_1,\abstr{y}w_2)$ and $u=(v/y)w_2$.
      This case is analogous to the previous case.
    \item Let $t=\elimor(\inlr(v_1,v_2),\abstr{x}w_1,\abstr{y}w_2)$ and $u=(v_1/x)w_1 \plus (v_2/y)w_2$.
      Then 
      $v_1$ is a proof of $B$ in the context $\Gamma$,
      $v_2$ is a proof of $C$ in the context $\Gamma$,
      $w_1$ is a proof of $A$ in the context $\Gamma,x:B$, and
      $w_2$ is a proof of $A$ in the context $\Gamma,y:C$.
      Thus, by \autoref{prop:subst}, we have that
      $(v_1/x)w_1$ is a proof of $A$ in the context $\Gamma$ and 
      $(v_2/x)w_2$ is a proof of $A$ in the context $\Gamma$.
      Thus, $(v_1/x)w_1 \plus (v_2/y)w_2$ is a proof of $A$ in the context $\Gamma$.
    \item Let $t=\star \plus \star$ and $u=\star$. Then $A=\top$, and we have
      that $\star$ is a proof of $\top$ in the context $\Gamma$.
    \item Let $t=(\lambda \abstr{x}v) \plus (\lambda \abstr{x}w)$ and $u=\lambda \abstr{x}(v \plus w)$.
      Then, $A=B\Rightarrow C$,
      $v$ is a proof of $C$ in the context $\Gamma,x:B$, and
      $w$ is a proof of $C$ in the context $\Gamma,x:B$.
      Therefore, $v\plus w$ is a proof of $C$ in the context $\Gamma,x:B$,
      and so $\lambda \abstr{x}(v \plus w)$ is a proof of $A$ in the context
      $\Gamma$.
    \item Let $t=\pair{v_1}{w_1} \plus \pair{v_2}{w_2}$ and $u=\pair{v_1\plus v_2}{w_1\plus w_2}$.
      Then, $A=B\wedge C$,
      $v_1$ is a proof of $B$ in the context $\Gamma$,
      $v_2$ is a proof of $B$ in the context $\Gamma$,
      $w_1$ is a proof of $C$ in the context $\Gamma$, and
      $w_2$ is a proof of $C$ in the context $\Gamma$.
      Thus, $\pair{v_1\plus v_2}{w_1\plus w_2}$ is a proof of $A$ in the context
      $\Gamma$.
    \item Let $t=\inl(v_1) \plus \inl(v_2)$ and $u=\inl(v_1\plus v_2)$
      Then, $A=B\vee C$,
      $v_1$ is a proof of $B$ in the context $\Gamma$,
      $v_2$ is a proof of $B$ in the context $\Gamma$,
      Thus, $\inl{v_1\plus v_2}$ is a proof of $A$ in the context $\Gamma$.
    \item Let $t=\inl(v) \plus \inr(w)$ and $u=\inlr(v,w) $
      Then, $A=B\vee C$,
      $v$ is a proof of $B$ in the context $\Gamma$, and
      $w$ is a proof of $C$ in the context $\Gamma$. 
      Thus, $\inlrpair{v}{w}$ is a proof of $A$ in the context $\Gamma$.
    \item Let $t=\inl(v_1) \plus \inlr(v_2,w)$ and $u=\inlr(v_1 \plus v_2,w) $
      Then, $A=B\vee C$,
      $v_1$ is a proof of $B$ in the context $\Gamma$, 
      $v_2$ is a proof of $B$ in the context $\Gamma$, and
      $w$ is a proof of $C$ in the context $\Gamma$. 
      Thus, $\inlrpair{v_1\plus v_2}{w}$ is a proof of $A$ in the context $\Gamma$.
    \item Let $t=\inr(w) \plus \inl(v)$ and $u=\inlr(v,w)$
      Then, $A=B\vee C$,
      $v$ is a proof of $B$ in the context $\Gamma$, and
      $w$ is a proof of $C$ in the context $\Gamma$.
      Thus, $\inlrpair{v}{w}$ is a proof of $A$ in the context $\Gamma$.
    \item Let $t=\inr(w_1) \plus \inr(w_2)$ and $u=\inr(w_1 \plus w_2) $
      Then, $A=B\vee C$,
      $w_1$ is a proof of $C$ in the context $\Gamma$, and
      $w_2$ is a proof of $C$ in the context $\Gamma$.
      Thus, $\inr{w_1\plus w_2}$ is a proof of $A$ in the context $\Gamma$.
    \item Let $t=\inr(w_1) \plus \inlr(v,w_2)$ and $u=\inlr(v,w_1 \plus w_2) $
      Then, $A=B\vee C$,
      $v$ is a proof of $B$ in the context $\Gamma$,
      $w_1$ is a proof of $C$ in the context $\Gamma$, and
      $w_2$ is a proof of $C$ in the context $\Gamma$.
      Thus, $\inlrpair{v}{w_1\plus w_2}$ is a proof of $A$ in the context $\Gamma$.
    \item Let $t=\inlr(v_1,w) \plus \inl(v_2) $ and $u= \inlr(v_1 \plus v_2,u) $
      Then, $A=B\vee C$,
      $v_1$ is a proof of $B$ in the context $\Gamma$,
      $v_2$ is a proof of $B$ in the context $\Gamma$, and
      $w$ is a proof of $C$ in the context $\Gamma$.
      Thus, $\inlrpair{v_1\plus v_2}{w}$ is a proof of $A$ in the context $\Gamma$.
    \item Let $t=\inlr(v,w_1) \plus \inr(w_2) $ and $u= \inlr(v,w_1 \plus w_2) $
      Then, $A=B\vee C$,
      $v$ is a proof of $B$ in the context $\Gamma$,
      $w_1$ is a proof of $C$ in the context $\Gamma$, and
      $w_2$ is a proof of $C$ in the context $\Gamma$.
      Thus, $\inlrpair{v}{w_1\plus w_2}$ is a proof of $A$ in the context $\Gamma$.
    \item Let $t=\inlr(v_1,w_1) \plus \inlr(v_2,w_2) $ and $u= \inlr(v_1 \plus v_2,w_1 \plus w_2)$
      Then, $A=B\vee C$,
      $v_1$ is a proof of $B$ in the context $\Gamma$,
      $v_2$ is a proof of $B$ in the context $\Gamma$,
      $w_1$ is a proof of $C$ in the context $\Gamma$, and
      $w_2$ is a proof of $C$ in the context $\Gamma$.
      Thus, $\inlrpair{v_1\plus v_2}{w_1\plus w_2}$ is a proof of $A$ in the context $\Gamma$.
      \qedhere
  \end{enumerate}
\end{proof}

\subsubsection{Introduction property}
\begin{thm}[Introduction]
\label{introductions}
A closed irreducible proof is an introduction.
\end{thm}
\begin{proof}
  Let $t$ be a closed irreducible proof of some proposition $A$. We prove,
  by induction on the structure of $t$ that $t$ is an introduction.

  As the proof $t$ is closed, it is not a variable.

  It cannot be a sum $u \plus v$, as if it were $u$ and $v$ would be
  closed irreducible proofs of the same proposition, hence, by induction
  hypothesis, they would either be both introductions of $\top$, both
  introductions of $\Rightarrow$, both introductions of $\wedge$, or
  both introductions of $\vee$, and the proof $t$ would be reducible.

  It cannot be an elimination as if it were of the form $\elimtop(u,v)$,
  $u~v$, $\elimand^1(u,\abstr{x}v)$, $\elimand^2(u,\abstr{x}v)$, or
  $\elimor(u,\abstr{x}v,\abstr{y}w)$, then $u$ would a closed
  irreducible proof, hence, by induction hypothesis, it would an
  introduction and the proof $t$ would be reducible.  If it were
  $\elimbot{A}(u)$, then $u$ would be a closed irreducible proof of $\bot$,
  by induction hypothesis, it would be an introduction of $\bot$ and no
  such introduction rule exists.

  Hence, it is an introduction. \qedhere
\end{proof}

\begin{rem}
In calculi where the sum rule commutes with the elimination rule of
the disjunction, rather than with its introduction rules
\cite{DiazcaroDowekTCS23,DiazcaroDowekMSCS24}, when $t$ and $u$ are
closed irreducible proofs, so is the proof $\inl(t) \plus
\inr(u)$. Hence the introduction theorem is weaker. In the
in-left-right-+-calculus, such a proof reduces to $\inlr(t,u)$.
\end{rem}

\subsubsection{Termination}
The termination proof proceeds, as usual, by defining, for each
proposition $A$, a set of strongly terminating proofs $\llbracket A
\rrbracket$ and then proving, by induction over proof structure, that
all proofs of $A$ are in the set $\llbracket A \rrbracket$.

\begin{rem}[Avoiding ultra-reduction]
In this proof, we need to prove that the proofs of the form
$\elimor(t,\abstr{x} u,\abstr{y} v)$, where $t$ is a proof of $B \vee
C$ and $u$ and $v$ are proofs of $A$, are in $\llbracket A \rrbracket$,
using the induction hypothesis that $t$ is in $\llbracket B \vee C
\rrbracket$ and (instances of) $u$ and $v$ are in $\llbracket A
\rrbracket$.  To do so, we prove that all the one-step reducts of this
proof are in $\llbracket A \rrbracket$.

When we have a rule commuting the sum with the elimination 
of the disjunction \cite{DiazcaroDowekTCS23,DiazcaroDowekMSCS24}, that is, reducing
  $\elimor(t_1 \plus t_2,\abstr{x} u,\abstr{y} v)$ to $
  \elimor(t_1,\abstr{x} u,\abstr{y} v) \plus \elimor(t_2,\abstr{x}
  u,\abstr{y} v)$,
we need to prove that the proofs $t_1$ and $t_2$ are in $\llbracket B
\vee C \rrbracket$, when $t = t_1 \plus t_2$ is, so that we can
conclude that $\elimor(t_1,\abstr{x} u,\abstr{y} v)$ and
$\elimor(t_2,\abstr{x} u,\abstr{y} v)$ and then $\elimor(t_1,\abstr{x}
u,\abstr{y} v) \plus \elimor(t_2,\abstr{x} u,\abstr{y} v)$ are in
$\llbracket A \rrbracket$.

But, deducing that $t_1$ and $t_2$ are in $\llbracket B \vee C
\rrbracket$, from the fact that $t_1 \plus t_2$ is, is not easy. A
usual way to deal with this issue is to add two reduction rules: $t_1 \plus t_2 \lra t_1$ and $t_1 \plus t_2 \lra t_2$
and prove the whole termination theorem for this so-called \emph{ultra-reduction} relation \cite{Girard,DW,DiazcaroDowekTCS23}.

Here, as we commute the sum rule with the introduction rules of the
disjunction rather than with its elimination rule, we do not need to
extend the reduction relation.
\end{rem}

\begin{defi}[Length of reduction]
If $t$ is a strongly terminating proof, we write $|t|$ for the
maximum length of a reduction sequence issued from $t$.
\end{defi}

\begin{prop}[Termination of a sum]
\label{terminationsum}
If $t$ and $u$ strongly terminate, then so does $t \plus u$. 
\end{prop}

\begin{proof}
We prove that all the one-step reducts of $t \plus u$ strongly
 terminate, by induction first on $|t| + |u|$ and then on the size of
$t$.

If the reduction takes place in $t$ or in $u$ we apply the induction
hypothesis.
Otherwise, the reduction is at the root and the rule used is one of
the rules
\eqref{rusumstar} to \eqref{rusuminlr}.

In the case \eqref{rusumstar}, the proof $\star$ is irreducible,
hence it strongly terminates.

In the case \eqref{rusumlam}, $t = \lambda \abstr{x} t_1$, $u =
\lambda \abstr{x} u_1$, by induction hypothesis, the proof $t_1 \plus
u_1$ strongly terminates, thus so does the proof $\lambda
\abstr{x}(t_1 \plus u_1)$.  The proof is similar for the cases
\eqref{rusumpair} to \eqref{rusuminlr}.  \qedhere
\end{proof}

\begin{defi}
\label{def:interpretation}
We define, by induction on the proposition $A$, a set of proofs
$\llbracket A \rrbracket$:
\begin{itemize}
\item $t \in \llbracket \top \rrbracket$ if $t$ strongly terminates,

\item $t \in \llbracket \bot \rrbracket$ if $t$ strongly terminates,

\item $t \in \llbracket A \Rightarrow B \rrbracket$ if $t$ strongly
  terminates and whenever it reduces to a proof of the form $\lambda
  \abstr{x}u$, then for every $v \in \llbracket A \rrbracket$, $(v/x)u \in
  \llbracket B \rrbracket$,

\item $t \in \llbracket A \wedge B \rrbracket$ if $t$ strongly
  terminates and whenever it reduces to a proof of the form
  $\pair{u}{v}$, then $u \in \llbracket A \rrbracket$ and $v \in \llbracket
  B \rrbracket$,

\item $t \in \llbracket A \vee B \rrbracket$ if $t$ strongly
  terminates and whenever it reduces to a proof of the form
  $\inl(u)$ (resp. $\inr(v)$, $\inlr(u,v)$), then $u \in
  \llbracket A \rrbracket$ (resp. $v \in \llbracket B \rrbracket$,
$u \in \llbracket A \rrbracket$ and $v \in \llbracket B \rrbracket$).
\end{itemize}
\end{defi}

\begin{prop}[Variables]
\label{Var}
For any $A$, the set $\llbracket A \rrbracket$ contains all the variables.
\end{prop}

\begin{proof}
A variable is irreducible, hence it strongly terminates. Moreover, it
never reduces to an introduction.
\qedhere
\end{proof}   

\begin{prop}[Closure by reduction]
\label{closure}
If $t \in \llbracket A \rrbracket$ and $t \longrightarrow^* t'$, then 
$t' \in \llbracket A \rrbracket$.
\end{prop}

\begin{proof}
If $t \longrightarrow^* t'$ and $t$ strongly terminates, then $t'$
strongly terminates.

Furthermore, if $A$ has the form $B \Rightarrow C$ and $t'$ reduces to
$\lambda \abstr{x}u$, then so does $t$, hence for every $v \in \llbracket B
\rrbracket$, $(v/x)u \in \llbracket C \rrbracket$.

If $A$ has the form $B \wedge C$ and $t'$ reduces to $\pair{u}{v}$,
then so does $t$, hence $u \in \llbracket B \rrbracket$ and $v \in
\llbracket C \rrbracket$.

If $A$ has the form $B \vee C$ and $t'$ reduces to $\inl(u)$
(resp. $\inr(v)$, $\inlr(u,v)$), then so does $t$, hence
$u \in   \llbracket A \rrbracket$ (resp. $v \in \llbracket B \rrbracket$,
$u \in   \llbracket A \rrbracket$ and $v \in \llbracket B \rrbracket$).
\qedhere
\end{proof}

\begin{prop}[Girard's lemma]
\label{CR3}
Let $t$ be a proof that is not an introduction, 
such that all the one-step reducts of $t$
are in $\llbracket A \rrbracket$. Then, $t \in \llbracket A \rrbracket$.
\end{prop}

\begin{proof}
Let $t, t_2, \ldots$ be a reduction sequence issued from $t$. If it has a
single element, it is finite. Otherwise, we have $t \longrightarrow
t_2$. As $t_2 \in \llbracket A \rrbracket$, it strongly terminates and
the reduction sequence is finite. Thus, $t$ strongly terminates.

Furthermore, if $A$ has the form $B \Rightarrow C$ and $t
\longrightarrow^* \lambda \abstr{x}u$, then let $t , t_2, \ldots, t_n$ be a
reduction sequence from $t$ to $\lambda \abstr{x}u$.  As $t_n$ is an
introduction and $t$ is not, $n \geq 2$. Thus, $t \longrightarrow t_2
\longrightarrow^* t_n$. We have $t_2 \in \llbracket A \rrbracket$,
thus for all $v \in \llbracket B \rrbracket$, $(v/x)u \in \llbracket C
\rrbracket$.

If $A$ has the form $B \wedge C$ and $t \longrightarrow^* \pair{u}{v}$, then let $t , t_2, \ldots, t_n$ be a reduction sequence
from $t$ to $\pair{u}{v}$.  As $t_n$ is an introduction and
$t$ is not, $n \geq 2$. Thus, $t \longrightarrow t_2 \longrightarrow^*
t_n$. We have $t_2 \in \llbracket A \rrbracket$, thus $u \in
\llbracket B \rrbracket$ and $v \in \llbracket C \rrbracket$.

And if $A$ has the form $B \vee C$ and $t \longrightarrow^* 
\inl(u)$ (resp. $\inr(v)$, $\inlr(u,v)$)
then
let $t , t_2, \ldots, t_n$ be a reduction sequence from $t$ to 
$\inl(u)$ (resp. $\inr(v)$,  $\inlr(u,v)$).
As $t_n$ is an introduction and $t$ is not, $n \geq 2$. Thus, $t
\longrightarrow t_2 \longrightarrow^* t_n$. We have $t_2 \in
\llbracket A \rrbracket$, thus 
$u \in \llbracket B \rrbracket$ (resp. $v \in \llbracket C \rrbracket$, 
$u \in \llbracket B \rrbracket$ and $v \in \llbracket C \rrbracket$).
\qedhere
\end{proof}

In Propositions~\ref{sum} to~\ref{elimor}, we prove the adequacy of each proof constructor.

\begin{prop}[Adequacy of $\plus$]
\label{sum}
If $t_1 \in \llbracket A \rrbracket$ and $t_2 \in \llbracket A
\rrbracket$, then $t_1 \plus t_2 \in \llbracket A \rrbracket$.
\end{prop}

\begin{proof}

  By induction on $A$.

  The proofs $t_1$ and $t_2$ strongly terminate.  Thus, by
  \autoref{terminationsum}, the proof $t_1 \plus t_2$ strongly
  terminates.

  Furthermore:
  \begin{itemize}
  \item If the proposition $A$ has the form $\top$ or $\bot$, then
    the strong termination of $t_1 \plus t_2$ is sufficient to conclude
    that $t_1 \plus t_2 \in \llbracket A \rrbracket$.
    
\item If the proposition $A$ has the form $B \Rightarrow C$, and $t_1
  \plus t_2 \lra^* \lambda \abstr{x} v$ then $t_1 \lra^* \lambda
  \abstr{x} u_1$, $t_2 \lra^* \lambda \abstr{x} u_2$, and $u_1 \plus
  u_2 \lra^* v$.

  As $t_1 \in \llbracket B \Rightarrow C \rrbracket$ and $t_1 \lra^*
  \lambda \abstr{x} u_1$, we have, for every $w$ in $\llbracket B
  \rrbracket$, $(w/x)u_1 \in \llbracket C \rrbracket$.  In the same
  way, for every $w$ in $\llbracket B \rrbracket$, $(w/x)u_2 \in
  \llbracket C \rrbracket$.  By induction hypothesis, $(w/x)(u_1 \plus
  u_2) = (w/x)u_1 \plus (w/x)u_2 \in \llbracket C \rrbracket$ and by
  \autoref{closure}, $(w/x)v \in \llbracket C \rrbracket$.

\item If the proposition $A$ has the form $B \wedge C$, and $t_1
  \plus t_2 \lra^* \pair{v}{v'}$ then
  $t_1 \lra^* \pair{u_1}{u'_1}$,
  $t_2 \lra^* \pair{u_2}{u'_2}$, $u_1 \plus u_2 \lra^* v$, and $u'_1
  \plus u'_2 \lra^* v'$.

  As $t_1 \in \llbracket B \wedge C \rrbracket$ and $t_1 \lra^*
  \pair{u_1}{u'_1}$, we have $u_1 \in \llbracket B \rrbracket$
  and $u'_1 \in \llbracket C \rrbracket$. In the same way, 
we have $u_2 \in \llbracket B \rrbracket$
  and $u'_2 \in \llbracket C \rrbracket$. 
  By induction hypothesis, $u_1 \plus u_2 \in \llbracket B \rrbracket$
and $u'_1 \plus u'_2 \in \llbracket C \rrbracket$, and, by
\autoref{closure}, $v \in \llbracket B \rrbracket$ and $v' \in
\llbracket C \rrbracket$.

\item If the proposition $A$ has the form $B \vee C$, and $t_1 \plus
  t_2 \lra^* \inl(v)$ then $t_1 \lra^* \inl(u_1)$, $t_2 \lra^*
  \inl(u_2)$, and $u_1 \plus u_2 \lra^* v$.  As $t_1 \in \llbracket B
  \vee C \rrbracket$ and $t_1 \lra^* \inl(u_1)$, we have $u_1 \in
  \llbracket B \rrbracket$. In the same way, $u_2 \in \llbracket B
  \rrbracket$.  By induction hypothesis, $u_1 \plus u_2 \in \llbracket
  B \rrbracket$ and, by \autoref{closure}, $v \in \llbracket B
  \rrbracket$.

\item If the proposition $A$ has the form $B \vee C$, and $t_1 \plus
  t_2 \lra^* \inr(v)$ the proof is similar.

\item
    If the proposition $A$ has the form $B \vee C$, and $t_1 \plus
    t_2 \lra^* \inlr(v,v')$ then, either
    \begin{itemize}
      \item
	$t_1 \lra^* \inl(u_1)$, $t_2 \lra^* \inr(u'_2)$, and
	$u_1 \lra^* v$ and  $u'_2 \lra^* v'$;

      \item
	$t_1 \lra^* \inl(u_1)$, $t_2 \lra^* \inlr(u_2,u'_2)$, and
	$u_1 \plus u_2 \lra^* v$ and  $u'_2 \lra^* v'$;

      \item
	$t_1 \lra^* \inr(u'_1)$, $t_2 \lra^* \inl(u_2)$, and
	$u_2 \lra^* v$ and  $u'_1 \lra^* v'$;

      \item
	$t_1 \lra^* \inr(u'_1)$, $t_2 \lra^* \inlr(u_2,u'_2)$, and
	$u_2 \lra^* v$ and  $u'_1 \plus u'_2 \lra^* v'$;

      \item
	$t_1 \lra^* \inlr(u_1,u'_1)$, $t_2 \lra^* \inl(u_2)$, and
	$u_1 \plus u_2 \lra^* v$ and  $u'_1 \lra^* v'$;

      \item
	$t_1 \lra^* \inlr(u_1,u'_1)$, $t_2 \lra^* \inr(u'_2)$, and
	$u_1 \lra^* v$ and  $u'_1 \plus u'_2 \lra^* v'$;

      \item
	or $t_1 \lra^* \inlr(u_1,u'_1)$, $t_2 \lra^* \inlr(u_2,u'_2)$, and
	$u_1 \plus u_2 \lra^* v$ and  $u'_1 \plus u'_2 \lra^* v'$.
    \end{itemize}

    As these seven cases are similar, we consider only the last one.

    As $t_1 \in \llbracket B \vee C \rrbracket$ and $t_1 \lra^*
    \inlr(u_1,u'_1)$, we have 
    $u_1 \in \llbracket B \rrbracket$
    and
    $u'_1 \in \llbracket C \rrbracket$.
    In the same way, 
    $u_2 \in \llbracket B \rrbracket$
    and
    $u'_2 \in \llbracket C \rrbracket$.
    By induction hypothesis,
    $u_1 \plus u_2 \in \llbracket B \rrbracket$
    and
    $u'_1 \plus u'_2 \in \llbracket C \rrbracket$.
    As $u_1 \plus u_2 \lra^* v$ and $u'_1\plus u'_2\to v'$,
    we get by  \autoref{closure}, 
    $v \in \llbracket B \rrbracket$ and
    $v' \in \llbracket C \rrbracket$.
    \qedhere
\end{itemize}
\end{proof}

\begin{prop}[Adequacy of $\star$]
\label{star}
We have $\star \in \llbracket \top \rrbracket$.
\end{prop}

\begin{proof}
As $\star$ is irreducible, it strongly terminates, hence
$\star \in \llbracket \top \rrbracket$. 
\qedhere
\end{proof}

\begin{prop}[Adequacy of $\lambda$]
\label{abstraction}
If, for all $u \in \llbracket A \rrbracket$, $(u/x)t \in \llbracket B
\rrbracket$, then $\lambda \abstr{x}t \in \llbracket A \Rightarrow B
\rrbracket$.
\end{prop}

\begin{proof}
By \autoref{Var}, $x \in \llbracket A \rrbracket$, thus
$t = (x/x)t \in \llbracket B \rrbracket$. Hence, $t$ strongly
terminates.  Consider a reduction sequence issued from $\lambda
\abstr{x}t$.  This sequence can only reduce $t$ hence it is finite. Thus,
$\lambda \abstr{x}t$ strongly terminates.

Furthermore, if $\lambda \abstr{x}t \longrightarrow^* \lambda \abstr{x}t'$, then
$t \lra^* t'$.  Let $u \in \llbracket A \rrbracket$,
$(u/x)t \lra^* (u/x)t'$. 
By \autoref{closure}, $(u/x)t' \in \llbracket B \rrbracket$.
\qedhere
\end{proof}

\begin{prop}[Adequacy of $\pair{}{}$]
\label{pair}
If $t_1 \in \llbracket A \rrbracket$ and $t_2 \in \llbracket B
\rrbracket$, then $\pair{t_1}{t_2} \in \llbracket A \wedge B
\rrbracket$.
\end{prop}

\begin{proof}
The proofs $t_1$ and $t_2$ strongly terminate. Consider a reduction
sequence issued from $\pair{t_1}{t_2}$.  This sequence can only
reduce $t_1$
and $t_2$, hence it is finite.  Thus, $\pair{t_1}{t_2}$
strongly terminates.

Furthermore, if $\pair{t_1}{t_2} \longrightarrow^* \pair
{t'_1}{t'_2}$, then $t_1 \lra^* t'_1$ and $t_2 \lra^* t'_2$.  By
\autoref{closure}, $t'_1 \in \llbracket A \rrbracket$ and
$t'_2 \in \llbracket B \rrbracket$.
\qedhere
\end{proof}

\begin{prop}[Adequacy of $\inl$]
\label{inl}
If $t \in \llbracket A \rrbracket$, then $\inl(t) \in \llbracket A
\vee B \rrbracket$.
\end{prop}

\begin{proof}
The proof $t$ strongly terminates. Consider a reduction
sequence issued from $\inl(t)$.  This sequence can only
reduce $t$, hence it is finite.  Thus, $\inl(t)$
strongly terminates.

Furthermore, if $\inl(t) \longrightarrow^* \inl
(t')$, then $t \lra^* t'$. By 
\autoref{closure}, $t' \in \llbracket A \rrbracket$.
And the proof $\inl(t)$ never reduces to a proof of the form $\inr(t'_2)$ or
$\inlr(t'_1,t'_2)$.
\qedhere
\end{proof}

\begin{prop}[Adequacy of $\inr$]
\label{inr}
If $t \in \llbracket B
\rrbracket$, then $\inr(t) \in \llbracket A \vee B
\rrbracket$.
\end{prop}

\begin{proof}
Similar to that of \autoref{inl}.
\qedhere
\end{proof}

\begin{prop}[Adequacy of $\inlr$]
\label{inlr}
If $t_1 \in \llbracket A \rrbracket$ and $t_2 \in \llbracket B
\rrbracket$, then $\inlr(t_1,t_2) \in \llbracket A \vee B
\rrbracket$.
\end{prop}

\begin{proof}
    The proofs $t_1$ and $t_2$ strongly terminate. Consider a reduction
    sequence issued from $\inlr(t_1,t_2)$.  This sequence can only
    reduce $t_1$ and $t_2$, hence it is finite.  Thus, $\inlr(t_1,t_2)$
    strongly terminates.

    Furthermore, if $\inlr(t_1,t_2) \longrightarrow^*
    \inlr(t'_1,t'_2)$, then $t_1 \lra^* t'_1$ and $t_2\lra^* t'_2$. By
    \autoref{closure}, $t'_1 \in \llbracket A \rrbracket$ and
    $t'_2 \in \llbracket B \rrbracket$.  And the proof $\inlr(t)$
    never reduces to a proof of the form $\inl(t'_1)$ or $\inr(t'_2)$.
      \qedhere
\end{proof}

\begin{prop}[Adequacy of $\elimtop$]
\label{elimtop}
If $t_1 \in \llbracket \top \rrbracket$ and $t_2 \in \llbracket C \rrbracket$, 
then $\elimtop(t_1,t_2) \in \llbracket C \rrbracket$.
\end{prop}

\begin{proof}
The proofs $t_1$ and $t_2$ strongly terminate.  We prove, by
induction on $|t_1| + |t_2|$, that $\elimtop(t_1,t_2)
\in \llbracket C \rrbracket$.  Using \autoref{CR3}, we only
need to prove that every of its one step reducts is in $\llbracket C
\rrbracket$.  If the reduction takes place in $t_1$ or $t_2$, then we
apply \autoref{closure} and the induction hypothesis.

Otherwise, the proof $t_1$ is $\star$ and the
reduct is $t_2$. 
\qedhere
\end{proof}

\begin{prop}[Adequacy of $\elimbot{C}$]
\label{elimbot}
If $t \in \llbracket \bot \rrbracket$, 
then $\elimbot{C}(t) \in \llbracket C \rrbracket$.
\end{prop}

\begin{proof}
The proof $t$ strongly terminates.  Consider a reduction sequence
issued from $\elimbot{C}(t)$.  This sequence can only reduce $t$, hence it
is finite.  Thus, $\elimbot{C}(t)$ strongly terminates.  Moreover, it
never reduces to an introduction.
\qedhere
\end{proof}

\begin{prop}[Adequacy of application]
\label{application}
If $t_1 \in \llbracket A \Rightarrow B \rrbracket$ and $t_2 \in
\llbracket A \rrbracket$, then $t_1~t_2 \in \llbracket B
\rrbracket$.
\end{prop}

\begin{proof}
The proofs $t_1$ and $t_2$ strongly terminate. We prove, by induction
on $|t_1| + |t_2|$, that $t_1~t_2 \in \llbracket B \rrbracket$. Using
\autoref{CR3}, we only need to prove that every of its one
step reducts is in $\llbracket B \rrbracket$.  If the reduction takes
place in $t_1$ or in $t_2$, then we apply \autoref{closure}
and the induction hypothesis.

Otherwise, the proof $t_1$ has the form $\lambda \abstr{x}u$ and the reduct
is $(t_2/x)u$.  As $\lambda \abstr{x}u \in \llbracket A \Rightarrow B
\rrbracket$, we have $(t_2/x)u \in \llbracket B \rrbracket$.
\qedhere
\end{proof}

\begin{prop}[Adequacy of $\elimand^1$]
\label{elimand1}
If $t_1 \in \llbracket A \wedge B \rrbracket$ and,
for all $u$ in $\llbracket A \rrbracket$,
$(u/x)t_2 \in \llbracket C \rrbracket $, 
then $\elimand^1(t_1, \abstr{x}t_2) \in \llbracket C \rrbracket$.
\end{prop}

\begin{proof}
By \autoref{Var}, $x \in \llbracket A \rrbracket$
thus $t_2 = (x/x)t_2 \in \llbracket C
\rrbracket$.  Hence, $t_1$ and $t_2$ strongly terminate.  We prove, by
induction on $|t_1| + |t_2|$, that $\elimand^1(t_1, \abstr{x}t_2)
\in \llbracket C \rrbracket$.  Using \autoref{CR3}, we only
need to prove that every of its one step reducts is in $\llbracket C
\rrbracket$.  If the reduction takes place in $t_1$ or $t_2$, then we
apply \autoref{closure} and the induction hypothesis.

Otherwise, the proof $t_1$ has the form $\pair{u}{v}$ and the
reduct is $(u/x)t_2$.  As $\pair{u}{v} \in \llbracket A
\wedge B \rrbracket$, we have $u \in \llbracket A \rrbracket$.
Hence, $(u/x)t_2 \in \llbracket C \rrbracket$.
\qedhere
\end{proof}

\begin{prop}[Adequacy of $\elimand^2$]
\label{elimand2}
If $t_1 \in \llbracket A \wedge B \rrbracket$ and,
for all $u$ in $\llbracket B \rrbracket$, 
$(u/x)t_2 \in \llbracket C \rrbracket $, 
then $\elimand^2(t_1,\abstr{x}t_2) \in \llbracket C \rrbracket$.
\end{prop}

\begin{proof}
Similar to that of \autoref{elimand1}.
\qedhere
\end{proof}

\begin{prop}[Adequacy of $\elimor$]
\label{elimor}
If $t_1 \in \llbracket A \vee B \rrbracket$, for all $u$ in $\llbracket A
\rrbracket$, $(u/x)t_2 \in \llbracket C \rrbracket $, and, for all $v$
in $\llbracket B \rrbracket$, $(v/y)t_3 \in \llbracket C \rrbracket $,
then $\elimor(t_1, \abstr{x}t_2, \abstr{y}t_3) \in \llbracket C \rrbracket$.
\end{prop}

\begin{proof}
By \autoref{Var}, $x \in \llbracket A \rrbracket$, thus $t_2 =
(x/x)t_2 \in \llbracket C \rrbracket$. In the same way, $t_3 \in
\llbracket C \rrbracket$.  Hence, $t_1$, $t_2$, and $t_3$ strongly
terminate.  We prove, by induction on $|t_1| + |t_2| + |t_3|$, that
$\elimor(t_1, \abstr{x}t_2, \abstr{y}t_3) \in \llbracket C
\rrbracket$.  Using \autoref{CR3}, we only need to prove that
every of its one step reducts is in $\llbracket C \rrbracket$.  If the
reduction takes place in $t_1$, $t_2$, or $t_3$, then we apply
\autoref{closure} and the induction hypothesis.

Otherwise, the proof $t_1$ has the form $\inlr(w_2,w_3)$
and the reduct is $(w_2/x)t_2 \plus (w_3/x)t_3$,
or the proof $t_1$ has the form
$\inl(w_2)$ and the reduct is $(w_2/x)t_2$, or the proof $t_1$ has the form
$\inr(w_3)$ and the reduct is $(w_3/x)t_3$.

  In the first case, as $\inlr(w_2,w_3) \in \llbracket A \vee B \rrbracket$ we
  have $w_2 \in \llbracket A \rrbracket$ and $w_3 \in \llbracket B \rrbracket$.
  Hence, $(w_2/x)t_2 \in \llbracket C \rrbracket$ and $(w_3/y)t_3 \in
  \llbracket C \rrbracket$.  Moreover, by \autoref{sum}, $(w_2/x)t_2
  \plus (w_3/x)t_3 \in \llbracket C \rrbracket$.

In the second, as $\inl(w_2) \in \llbracket A \vee B \rrbracket$, we
have $w_2 \in \llbracket A \rrbracket$.  Hence, $(w_2/x)t_2 \in
\llbracket C \rrbracket$.

In the third, as $\inr(w_3) \in \llbracket A \vee B \rrbracket$, we
have $w_3 \in \llbracket B \rrbracket$.  Hence, $(w_3/x)t_3 \in
\llbracket C \rrbracket$.
\qedhere
\end{proof}

\begin{prop}[Adequacy]
\label{prop:adequacy}
Let $t$ be a proof of $A$ in a context $\Gamma = x_1:A_1, \ldots, x_n:A_n$ and
$\sigma$ be a substitution mapping each variable $x_i$ to an element
of $\llbracket A_i \rrbracket$, then $\sigma t \in \llbracket A
\rrbracket$.
\end{prop}

\begin{proof}
By induction on the structure of $t$.

If $t$ is a variable, then, by definition of $\sigma$, $\sigma t \in
\llbracket A \rrbracket$.  For the other proof constructors, we use
the Propositions~\ref{sum} to~\ref{elimor}.
\qedhere
\end{proof}

\begin{thm}[Termination]
\label{th:Term} 
Let $\Gamma\vdash t:A$, then $t$ strongly terminates.
\end{thm}
\begin{proof}
  Let $\sigma$ be the substitution mapping each variable $x_i:A_i$ of
  $\Gamma$ to
  itself. Note that, by \autoref{Var}, this variable is an
  element of $\llbracket A_i \rrbracket$.  Then, $t = \sigma t$ is an
  element of $\llbracket A \rrbracket$. Hence, it strongly terminates.
\qedhere
\end{proof}

\subsubsection{Confluence}
\begin{thm}[Confluence]
\label{th:Confluence}
The in-left-right-+-calculus is confluent, that is whenever $u \llas t
\lras v$, then there exists a proof $w$, such that $u \lras w \llas v$.
\end{thm}
\begin{proof}
  The reduction rules of \autoref{figureductionrules} is left linear
  and has no critical pairs~\cite[Section 6]{Nipkow}.  By \cite[Theorem
  6.8]{Nipkow}, it is confluent.  \qedhere
\end{proof}

\section{Application to quantum computing}
\label{sec:quantum}

We show how the $\inlr$ symbol can be used to build a
programming language for quantum computing, similar to the
language of \cite{DiazcaroDowekTCS23,DiazcaroDowekMSCS24}, but without
the extra $\odot$ connective used there.

\subsection{The quantum in-left-right-calculus}
We extend the in-left-right-+-calculus into a quantum computing
language, in three steps: we first introduce a linear version of the
system of \autoref{sec:inlrcalculus} with the symbols $\multimap$,
$\one$, and $\oplus$, with the extra introduction rule for the symbol
$\oplus$ 
that is derivable in linear logic:
\[
  \irule{\Gamma \vdash A & \Gamma \vdash B}
  {\Gamma \vdash A \oplus B}
  {\mbox{$\oplus$-i3.}}
\]

Second, besides the rule sum, we add another interstitial rule,
\[
  \irule{\Gamma \vdash A} {\Gamma \vdash A} {\mbox{prod,}} 
\]
and scalars: we consider the field $\langle {\mathbb C}, +, \times\rangle$
of scalars, take a different rule $\one$-i for each scalar and a
different rule prod for each scalar.  The proof-term for each $\one$-i
axiom is $a.\star$, with $a\in {\mathbb C}$, and the proof-term for
each prod rule is $a\bullet t$, with $a\in {\mathbb C}$ and $t$ a
proof of $A$.  The rule prod and the scalars will enable us to build
linear combinations of proofs
\cite{Lineal,Vaux2009}. This
will be useful when we express vectors and matrices as proofs in this
logic.

We adapt the reduction rule \eqref{ruelimtop} of
\autoref{figureductionrules} to make the proof $\elimone(a.\star,
t)$ reduce to $a \bullet t$, and the rule \eqref{rusumstar} so that
the proof $a.\star \plus b.\star$ reduces to $(a + b).\star$, with
scalar addition. Moreover, we introduce rules that commute the product
with all introduction rules. The proof $a \bullet b.\star$ reduces to
$(a \times b).\star$, with rule \eqref{ll:rubulletstar}, where the scalars are multiplied.

Third, as quantum measurement implies
information-erasure, we must include reduction rules that differ from
\eqref{ruelimorinlr1}.  Thus, we add an extra elimination symbol
$\elimplus^{nd}$ for the disjunction
\[
  \irule{\Gamma \vdash t:A \oplus B & \Delta, x:A \vdash u:C & \Delta,
  y:B \vdash v:C}
  {\Gamma, \Delta \vdash \elimplus^{nd}(t,\abstr{x}u,\abstr{y}v):C}
  {\mbox{$\oplus$-e$^{nd}$}}
\]
and non-deterministic reduction rules
\begin{align*}
  \elimplus^{nd}(\inl(t),\abstr{x}v,\abstr{y}w) & \longrightarrow (t/x)v
  \\
  \elimplus^{nd}(\inlr(t,u),\abstr{x}v,\abstr{y}w) & \longrightarrow (t/x)v
  \\
  \elimplus^{nd}(\inr(u),\abstr{x}v,\abstr{y}w) & \longrightarrow (u/y)w
  \\
  \elimplus^{nd}(\inlr(t,u),\abstr{x}v,\abstr{y}w) & \longrightarrow (u/y)w 
\end{align*}

Because of these two last rules, this calculus is another
non-deterministic extension of the $\lambda$-calculus. Note that
non-determinism is handled in at least two ways in the $\lambda$-calculus
\cite{deLiguoroPiperno}. One approach is to define a deterministic calculus
that returns the set of possible values (see, for instance,
\cite{deGroote}). Another is to introduce a random choice operator, such
that $choice(t,u)$ reduces non-deterministically to either $t$ or $u$ (see, for instance,
\cite{Ong93,deLiguoroPiperno}).

Quantum physics---and our calculus---involves
both approaches: superposition (and our $\plus$ operator) can be seen as
a way to collect all possible values of a physical quantity, while
measurement is a random operation returning one of them. In our
calculus, however, it is not the case that the proof $\inlr(t,u)$
reduces randomly to $\inl(t)$ or $\inr(u)$, it is only when reducing
the redex $\elimplus^{nd}(\inlr(t,u),\abstr{x}v,\abstr{y}w)$ that it
behaves randomly like $\inl(t)$ or $\inr(u)$. As in quantum physics,
randomness is controlled and occurs only in measurement operations.

\begin{figure}[t]
  $$
    \irule{}
    {x:A \vdash x:A}
    {\mbox{ax}}
    \quad
    \irule{\Gamma \vdash t:A & \Gamma \vdash u:A}
    {\Gamma \vdash t \plus u:A}
    {\mbox{sum}}
    \quad
    \irule{\Gamma \vdash t:A}
    {\Gamma \vdash a \bullet t:A}
    {\mbox{prod}(a)}
    \quad
    \irule{}
    {\vdash a.\star:\one}
    {\mbox{$\one$-i}(a)}
  $$
  $$
    \irule{\Gamma \vdash t:\one & \Delta \vdash u:A}
    {\Gamma, \Delta \vdash \elimone(t,u):A}
    {\mbox{$\one$-e}}
    \qquad
    \irule{\Gamma, x:A \vdash t:B}
    {\Gamma \vdash \lambda \abstr{x}t:A \multimap B}
    {\mbox{$\multimap$-i}}
    \qquad
    \irule{\Gamma \vdash t:A\multimap B & \Delta \vdash u:A}
    {\Gamma, \Delta \vdash t~u:B}
    {\mbox{$\multimap$-e}}
  $$
  $$
    \irule{\Gamma \vdash t:A}
    {\Gamma \vdash \inl(t):A \oplus B}
    {\mbox{$\oplus$-i1}}
    \qquad
    \irule{\Gamma \vdash t:B}
    {\Gamma \vdash \inr(t):A \oplus B}
    {\mbox{$\oplus$-i2}}
    \qquad
    \irule{\Gamma \vdash t:A & \Gamma \vdash u:B}
    {\Gamma \vdash \inlr(t,u):A \oplus B}
    {\mbox{$\oplus$-i3}}
  $$
  $$
    \irule{\Gamma \vdash t:A \oplus B & \Delta, x:A \vdash u:C & \Delta, y:B \vdash v:C}
    {\Gamma, \Delta \vdash \elimplus(t,\abstr{x}u,\abstr{y}v):C}
    {\mbox{$\oplus$-e}}
  $$
  $$
    \irule{\Gamma \vdash t:A \oplus B & \Delta, x:A \vdash u:C & \Delta,
    y:B \vdash v:C} {\Gamma, \Delta \vdash
    \elimplus^{nd}(t,\abstr{x}u,\abstr{y}v):C} {\mbox{$\oplus$-e$^{nd}$}}
  $$
  \caption{The typing rules of the quantum in-left-right-calculus\label{linear}}
\end{figure}

\begin{figure}[t]
  \centering
    \parbox{0.6\textwidth}{
    \begin{align}
      \elimone(a.\star,t) & \longrightarrow  a \bullet t
    \label{ll:ruelimone}\\
      (\lambda \abstr{x}t)~u & \longrightarrow  (u/x)t
    \label{ll:rubeta}\\
      \elimplus(\inl(t),\abstr{x}v,\abstr{y}w) & \longrightarrow  (t/x)v
    \label{ll:ruelimoplusinl}\\
      \elimplus(\inr(u),\abstr{x}v,\abstr{y}w) & \longrightarrow  (u/y)w
    \label{ll:ruelimoplusinr}\\
      \elimplus(\inlr(t,u),\abstr{x}v,\abstr{y}w) & \longrightarrow  (t/x)v \plus (u/y)w
    \label{ll:ruelimoplusinlr}\\
    \elimplus^{nd}(\inl(t),\abstr{x}v,\abstr{y}w) & \longrightarrow  (t/x)v 
    \label{ll:ruelimoplusndinl}\\
    \elimplus^{nd}(\inr(u),\abstr{x}v,\abstr{y}w) & \longrightarrow
    (u/y)w 
    \label{ll:ruelimoplusndinr}\\
    \elimplus^{nd}(\inlr(t,u),\abstr{x}v,\abstr{y}w) & \longrightarrow  (t/x)v 
    \label{ll:ruelimoplusndinlr1}\\
    \elimplus^{nd}(\inlr(t,u),\abstr{x}v,\abstr{y}w) & \longrightarrow
    (u/y)w 
    \label{ll:ruelimoplusndinlr2}
    \end{align}}
    \vspace{-1\baselineskip}

    \parbox{0.56\textwidth}{
      \begin{align}
	{a.\star} \plus b.\star&\longrightarrow  (a+b).\star
	\label{ll:rusumstar}\\
	(\lambda \abstr{x}t) \plus (\lambda \abstr{x}u) & \longrightarrow  \lambda \abstr{x}(t \plus u)
	\label{ll:rusumlam}\\
	\inl(t) \plus \inl(v) & \longrightarrow  \inl(t \plus v)
	\label{ll:rusuminl}\\
	\inl(t) \plus \inr(w) & \longrightarrow  \inlr(t,w)
	\label{ll:rusuminlinr} \\
	\inl(t) \plus \inlr(v,w) & \longrightarrow  \inlr(t \plus v,w)
	\label{ll:rusuminlinlr}\\
	\inr(u) \plus \inl(v) & \longrightarrow  \inlr(v,u)
	\label{ll:rusuminrinl}\\
	\inr(u) \plus \inr(w) & \longrightarrow  \inr(u \plus w) 
	\label{ll:rusuminr}\\
	\inr(u) \plus \inlr(v,w) & \longrightarrow  \inlr(v,u \plus w)
	\label{ll:rusuminrinlr}\\
	\inlr(t,u) \plus \inl(v) & \longrightarrow  \inlr(t \plus v,u)
	\label{ll:rusuminlrinl}\\
	\inlr(t,u) \plus \inr(w) & \longrightarrow  \inlr(t,u \plus w)
	\label{ll:rusuminlrinr}\\
	\inlr(t,u) \plus \inlr(v,w) & \longrightarrow  \inlr(t \plus v,u \plus w)
	\label{ll:rusuminlr}
      \end{align}
    }\hspace{-0.02\textwidth}
    \parbox{0.45\textwidth}{
      \begin{align}
	a \bullet b.\star&\longrightarrow  (a \times b).\star
	\label{ll:rubulletstar}\\
	a \bullet \lambda \abstr{x} t &\longrightarrow  \lambda \abstr{x} a \bullet t
	\label{ll:rubulletlam}\\
	a \bullet \inl(t) & \longrightarrow  \inl(a \bullet t)
	\label{ll:rubulletinl}\\
	a \bullet \inr(t) & \longrightarrow  \inr(a \bullet t)
	\label{ll:rubulletinr}\\
	a \bullet \inlr(t,u)  & \longrightarrow  \inlr(a \bullet t,a \bullet u)
	\label{ll:rubulletinlr}
      \end{align}
    }
  \caption{The reduction rules of the quantum in-left-right-calculus\label{linearreductionrules}}
\end{figure}

This leads to a logic with the connectives $\one$, $\multimap$, and $\oplus$. 
The syntax of proof-terms is
\begin{align*}
  t = ~ & x  \mid t \plus t \mid a \bullet t
  \mid a.\star \mid \elimone(t,t)
  \mid \lambda \abstr{x}t \mid t~t\\
  & \mid \inl(t) \mid \inr(t) \mid \inlr(t,t)
  \mid \elimplus(t,\abstr{x}t,\abstr{x}t)
  \mid \elimplus^{nd}(t,\abstr{x}t,\abstr{x}t),
\end{align*}
the typing rules are those of \autoref{linear}, and the reduction
rules those of \autoref{linearreductionrules}.

\begin{rem}
  A reduction rule always reducing the proof 
$\elimplus(\inlr(t,u),\abstr{x}v,\abstr{y}w)$ to $(t/x)v$ or to
$(u/y)w$ would imply a loss of information and a break in symmetry.  In the
same way, the solutions to the Buridan's donkey dilemma where the
donkey always chooses the left pile of hay or always the right pile of
hay introduces a loss of hay and a break in symmetry.

The non-deterministic reduction rules of $\elimplus^{nd}$, just
like quantum measurement, imply an erasure of information, but
preserve symmetry, at the cost of a loss of determinism.  In the same
way, \emph{the non-deterministic solution to the Buridan's donkey
dilemma} where the donkey randomly chooses the left pile of hay or
the right pile of hay preserves symmetry, at the cost of a loss of
determinism.  
Instead of always evolving to a state $\ket{\Phi}$ or to a state $\ket{\Psi}$, the donkey randomly evolves to the state $\ket{\Phi}$ or to the state $\ket{\Psi}$.

In the rule reducing $\elimplus(\inlr(t,u),\abstr{x}v,\abstr{y}w)$ to
$(t/x)v \plus (u/y)w$, the random choice is avoided by keeping both
proofs.  This can be called \emph{the quantum solution to the Buridan's
  donkey dilemma}: instead of evolving to the state
$\ket{\Phi}$ or to the state $\ket{\Psi}$, the donkey
evolves to the state $\ket{\Phi} + \ket{\Psi}$, without introducing a
loss of information or a breaking in the symmetry.  But the problems
of information preservation, symmetry, and determimism are delayed to
measurement.
\end{rem}

\begin{rem}
The quantum in-left-right-calculus bears some resemblance to the
$\odot$-calculus \cite{DiazcaroDowekTCS23,DiazcaroDowekMSCS24},
without the need to introduce the new connective ``sup'' $\odot$.
Indeed, ignoring the scalars that are treated in a different way, the
fragment formed with the symbols $\plus$, $\inlr$, $\elimplus$, and
$\elimplus^{nd}$, without $\inl$ and $\inr$, contains the same symbols
and the same rules as those, in the first version of the
$\odot$-calculus \cite{odotICTAC}, of the connective $\odot$, that has
the introduction rule of the conjunction and the elimination rule of
the disjunction, written there $\|$, $+$, $\elimsup^{\parallel}$, and
$\elimsup$. The rules \ref{ll:ruelimoplusinlr},
\ref{ll:ruelimoplusndinlr1}, and \ref{ll:ruelimoplusndinlr2} were
written
\begin{align*}
  \elimsup^\|(t + u,\abstr{x}v,\abstr{y}w) &\longrightarrow  (t/x)v~\|~(u/y)w\\
  \elimsup(t+u,\abstr{x}v,\abstr{y}w) &\longrightarrow  (t/x)v \\
  \elimsup(t+u,\abstr{x}v,\abstr{y}w) &\longrightarrow  (u/y)w
\end{align*}
there.  The careful reader might notice that in later versions of
this calculus \cite{DiazcaroDowekTCS23},
the symbol $\elimsup^{\parallel}$
was replaced, for a better clarity,
with $\elimsup^1$ and
$\elimsup^2$ that permitted to access to the proofs $t$ and $u$, in
the proof $t + u$, that is $\inlr(t,u)$, written there $[t,u]$.  For
instance, the proof $\elimsup^1([t,u],\abstr{x} v)$ reduced to
$(t/x)v$.  Now that we have the constructors $\inl$ and $\inr$, we
cannot have these proof constructors $\elimsup^1$ and $\elimsup^2$ any more, as there would be no
way to reduce the proof $\elimsup^1(\inr(u),\abstr{x} v)$, so we keep
the symbol $\elimplus$, that corresponds to $\elimsup^{\parallel}$.
\end{rem}

\subsection{Properties}\label{subsec:properties}

The subject reduction theorem (\autoref{thm:SR}),
the introduction theorem (\autoref{introductions}), and the termination
theorem (\autoref{th:Term}) generalize.
The confluence theorem (\autoref{th:Confluence}) does
not, because we now have non-deterministic rules.  If we remove the
symbol $\elimplus^{nd}$ and the four associated reduction rules, then
the obtained calculus is left linear and has no critical pairs, hence
it is confluent.  

The proofs of subject reduction (\autoref{thm:SRlin}) and the introduction property (\autoref{introductionslinear}) are omitted here and given in \autoref{app:properties}.
\begin{thm}[Subject reduction]\label{thm:SRlin}
  Let $\Gamma\vdash t:A$ and $t\lra u$, then $\Gamma\vdash u:A$.
  \qed
\end{thm}

\begin{thm}[Introduction]
  \label{introductionslinear}
A closed irreducible proof is an introduction.
\qed
\end{thm}

The proof of termination (\autoref{thm:termination}) could be adapted from that
of \autoref{th:Term}, taking into account scalars, the product
rule, and the new elimination symbol $\elimplus^{nd}$.  But, we prefer
to give a simpler proof, reaping the benefit of linearity, using a
measure.

Our goal is to build a measure function such that if $t$ is proof of
$B$ in a context $\Gamma, x:A$ and $u$ is a proof of $A$, then
$\mu((u/x)t) = \mu(t) + \mu(u)$. This would be the case, for the usual
notion of size, if $x$ had exactly one occurrence in $t$. But, due to
additive connectives, it may have several occurrences.  For example,
if $u \plus v$ is proof of $B$ in a context $\Gamma, x:A$, then $x$
may occur both in $u$ and in $v$.  Thus, we modify the definition of
the measure and take $\mu(t \plus u) = 1 + \max(\mu(t), \mu(u))$,
instead of $\mu(t \plus u) = 1 + \mu(t) + \mu(u)$, making the function
$\mu$ a mix between a size function and a depth function.  This leads
to the definition of the measure $\mu$ below.

\begin{defi}[The measure $\mu$]
\label{measureofaproof}~
  \begin{align*}
    \mu(x) &= 0 \\
    \mu(t \plus u) &= 1 + \max(\mu(t), \mu(u)) \\
    \mu(a \bullet t) &= 1 + \mu(t) \\
    \mu(a.\star) &= 1 \\
    \mu(\elimone(t,u)) &= 1 + \mu(t) + \mu(u) \\
    \mu(\lambda \abstr{x} t) &= 1 + \mu(t) \\
    \mu(t~u) &= 1 + \mu(t) + \mu(u) \\
    \mu(\inl(t)) &= 1 + \mu(t) \\
    \mu(\inr(t)) &= 1 + \mu(t) \\
    \mu(\inlr(t,u)) &= 1 + \max(\mu(t), \mu(u)) \\
    \mu(\elimplus(t,\abstr{y} u,\abstr{z} v)) &= 1 + \mu(t) + \max(\mu(u), \mu(v)) \\
    \mu(\elimplus^{nd}(t,\abstr{y} u,\abstr{z} v)) &= 1 + \mu(t) + \max(\mu(u), \mu(v))
  \end{align*}
\end{defi}

With this measure we have $\mu((u/x)t) = \mu(t)+\mu(u)$.

\begin{lem}\label{lem:msubst}
If $\Gamma, x:A \vdash t:B$ and $\Delta \vdash u:A$ then
  $\mu((u/x)t) = \mu(t)+\mu(u)$.
\end{lem}

\begin{proof}
  By induction on $t$.
  \begin{itemize}
    \item If $t$ is a variable, then $\Gamma$ is empty, $t = x$, 
      $(u/x)t = u$ and $\mu(t) = 0$.
      Thus, $\mu((u/x)t) = \mu(u) = \mu(t)+\mu(u)$.

    \item If $t = t_1 \plus t_2$, then $\Gamma, x:A \vdash t_1:B$,
      $\Gamma, x:A \vdash t_2:B$.  Using the induction hypothesis, we get 
      $\mu((u/x)t)
      = 1 + 
      \max(\mu((u/x)t_1),\mu((u/x)t_2))$ $= 1 + \max(\mu(t_1) + \mu(u),
      \mu(t_2) + \mu(u))
      = \mu(t) + \mu(u)$.

    \item If $t = a \bullet t_1$, then $\Gamma, x:A \vdash t_1:B$.
      Using the induction hypothesis, we get 
      $\mu((u/x)t)
      = 1 + \mu((u/x)t_1)
      = 1 + \mu(t_1) + \mu(u) = \mu(t) + \mu(u)$.

    \item The proof $t$ cannot be of the form
      $a.\star$, that is not a proof in $\Gamma, x:A$.

      \item If $t = \elimone(t_1,t_2)$, then
      $\Gamma = \Gamma_1, \Gamma_2$ and there are two cases.
      \begin{itemize}
	\item 
	  If $\Gamma_1, x:A \vdash t_1:\one$ and $\Gamma_2 \vdash t_2:B$,
	  then, using the induction hypothesis, we get
	  $\mu((u/x)t)
	  = 1 + \mu((u/x)t_1) + \mu(t_2)
	  = 1 + \mu(t_1) + \mu(u) + \mu(t_2)
	  = \mu(t) + \mu(u)$.

	\item 
	  If $\Gamma_1 \vdash t_1:\one$ and $\Gamma_2, x:A \vdash t_2:B$,
	  then, using the induction hypothesis, we get
	  $\mu((u/x)t)
	  = 1 + \mu(t_1) + \mu((u/x)t_2)
	  = 1 + \mu(t_1) + \mu(t_2) + \mu(u)
	  = \mu(t) + \mu(u)$.
      \end{itemize}

    \item If $t = \lambda \abstr{y} t_1$,
we apply the same method as for the case $t = a \bullet t_1$.

    \item If $t = t_1~t_2$, we apply the same method as for the
      case $t = \elimone(t_1,t_2)$.

    \item If $t = \inlr(t_1,t_2)$, we apply the same method as for the
      case $t = t_1 \plus t_2$.

    \item If $t = \inl(t_1)$ or $t = \inr(t_1)$,
     we apply the same method as for the case $t = a \bullet t_1$.

    \item   
      If $t = \elimplus(t_1,\abstr{y}t_2,\abstr{z}t_3)$ then
      $\Gamma = \Gamma_1, \Gamma_2$ and there are two cases.
      \begin{itemize}
	\item 
	  If $\Gamma_1, x:A \vdash t_1:C_1 \oplus C_2$,
	  $\Gamma_2, y:C_1 \vdash t_2:A$,
	  $\Gamma_2, z:C_2 \vdash t_3:A$, then using the induction hypothesis,
	  we get
	  $\mu((u/x)t)
	  = 1 + \mu((u/x)t_1) + \max(\mu(t_2),\mu(t_3))
          = 1 + \mu(t_1) + \mu(u) + \max(\mu(t_2),\mu(t_3))
	  = \mu(t) + \mu(u)$.

	\item 
	  If $\Gamma_1 \vdash t_1:C_1 \oplus C_2$,
	  $\Gamma_2, y:C_1, x:A \vdash t_2:A$,
	  $\Gamma_2, z:C_2, x:A \vdash t_3:A$, then using the induction hypothesis,
	  we get
	  $\mu((u/x)t)
	  = 1 + \mu(t_1) + \max(\mu((u/x)t_2),\mu((u/x)t_3))
	 =1 + \mu(t_1) + \max(\mu(t_2) + \mu(u),\mu(t_3) + \mu(u))
	 =1 + \mu(t_1) + \max(\mu(t_2),\mu(t_3)) + \mu(u)
	  = \mu(t) + \mu(u)$.
      \end{itemize}
      If $t = \elimplus^{nd}(t_1,\abstr{y}t_2,\abstr{z}t_3)$ the proof is
      similar.
	  \qedhere
  \end{itemize}
\end{proof}

And we can prove that if $t$ reduces to $u$ at the root with one of the rules
\eqref{ll:ruelimone} to \eqref{ll:ruelimoplusndinlr2}, then $\mu(t) >
\mu(u)$.

\begin{lem} \label{lem:mured1}
For the rules \eqref{ll:ruelimone} to \eqref{ll:ruelimoplusndinlr2},
if $\Gamma \vdash t:A$ and $t$ reduces to $u$ by a reduction step at
the root, then $\mu(t) > \mu(u)$.
\end{lem}
\begin{proof}
  We check the rules one by one, using \autoref{lem:msubst}.

  \begin{itemize}
    \item \eqref{ll:ruelimone} $\mu(\elimone(a.\star,t)) = 2 + \mu(t) > 1
      + \mu(t) = \mu(a \bullet t)$.

    \item \eqref{ll:rubeta} $\mu((\lambda \abstr{x}t)~u) = 2 + \mu(t) +
      \mu(u) > \mu(t) + \mu(u) = \mu((u/x)t)$.

    \item \eqref{ll:ruelimoplusinl}
      $\mu(\elimplus(\inl(t),\abstr{x}v,\abstr{y}w)) = 2 + \mu(t) +
      \max(\mu(v), \mu(w)) > \mu(t) + \mu(v) = \mu((t/x)v)$.

    \item \eqref{ll:ruelimoplusinr}
      $\mu(\elimplus(\inr(u),\abstr{x}v,\abstr{y}w)) = 2 + \mu(u) +
      \max(\mu(v),\mu(w)) > \mu(u) + \mu(w) = \mu((u/y)w)$.

    \item \eqref{ll:ruelimoplusinlr}
      $\mu(\elimplus(\inlr(t,u),\abstr{x}v,\abstr{y}w)) = 2 +
      \max(\mu(t),\mu(u)) + \max(\mu(v),\mu(w)) > 1 + \max(\mu(t),\mu(u))
      + \max(\mu(v),\mu(w)) \geq 1 + \max(\mu(t) + \mu(v),\mu(u) + \mu(w))
      = 1 + \max(\mu((t/x)v), \mu((u/y)w)) = \mu((t/x)v) \plus
      \mu((u/y)w)$.

    \item \eqref{ll:ruelimoplusndinl}
      $\mu(\elimplus^{nd}(\inl(t),\abstr{x}v,\abstr{y}w)) = 2 + 2 \mu(t) +
      2 \max(\mu(v), \mu(w)) > \mu(t) + \mu(v) = \mu((t/x)v)$.

    \item \eqref{ll:ruelimoplusndinr}
      $\mu(\elimplus^{nd}(\inr(u),\abstr{x}v,\abstr{y}w)) = 2 + \mu(u) +
      \max(\mu(v),\mu(w)) > \mu(u) + \mu(w) = \mu((u/y)w)$.

    \item \eqref{ll:ruelimoplusndinlr1}
      $\mu(\elimplus^{nd}(\inlr(t,u),\abstr{x}v,\abstr{y}w)) = 2 +
      \max(\mu(t),\mu(u)) + \max(\mu(v),\mu(w)) > \mu(t) + \mu(v) =
      \mu((t/x)v)$.

    \item \eqref{ll:ruelimoplusndinlr2}
      $\mu(\elimplus^{nd}(\inlr(t,u),\abstr{x}v,\abstr{y}w)) = 2 +
      \max(\mu(t),\mu(u)) + \max(\mu(v),\mu(w)) > \mu(u) + \mu(w) =
      \mu((u/y)w)$. \qedhere
  \end{itemize}
\end{proof}

But for the rules \eqref{ll:rusumstar} to
\eqref{ll:rubulletinlr}, we do not have this property. For example
$(\lambda \abstr{x}x) \plus (\lambda \abstr{x}x)$ reduces to $\lambda
\abstr{x}(x \plus x)$, but $\mu((\lambda \abstr{x}x) \plus (\lambda
\abstr{x}x)) = 2 = \mu(\lambda \abstr{x}(x \plus x))$. We only have a
similar property with the non strict order.

\begin{lem} \label{lem:mured2}
  For the rules \eqref{ll:rusumstar} to \eqref{ll:rubulletinlr}, if $t$
  reduces to $u$ by a reduction step at the root, then $\mu(t) \geq
  \mu(u)$.
\end{lem}
\begin{proof}
  We check the rules one by one.
  \begin{itemize}
    \item \eqref{ll:rusumstar} $\mu({a.\star} \plus b.\star) = 2 > 1 =
      \mu((a+b).\star)$

    \item \eqref{ll:rusumlam} $\mu((\lambda \abstr{x}t) \plus (\lambda
      \abstr{x}u)) = 1 + \max(1 + \mu(t), 1 + \mu(u)) = 1 + 1 +
      \max(\mu(t),\mu(u)) = \mu(\lambda \abstr{x}(t \plus u))$

    \item \eqref{ll:rusuminl} $\mu(\inl(t) \plus \inl(v)) = 1 + \max (1 +
      \mu(t), 1 + \mu(v)) = 1 + 1 + \max (\mu(t), \mu(v)) = \mu(\inl(t
      \plus v))$

    \item \eqref{ll:rusuminlinr} $\mu(\inl(t) \plus \inr(w)) = 1 + \max (1
      + \mu(t), 1 + \mu(w)) = 1 + 1 + \max (\mu(t), \mu(w)) > 1 + \max
      (\mu(t),\mu(w)) = \mu(\inlr(t,w))$

    \item \eqref{ll:rusuminlinlr} $\mu(\inl(t) \plus \inlr(v,w)) = 1 +
      \max (1 + \mu(t), 1 + \max(\mu(v),\mu(w))) = \max (2 + \mu(t),2 +
    \mu(v),2 + \mu(w)) \geq \max(2 + \mu(t), 2 + \mu(v)), 1 + \mu(w)) =
    1 + \max(1+ \max(\mu(t),\mu(v)), \mu(w)) = \mu(\inlr(t \plus v,w))$

  \item \eqref{ll:rusuminrinl} $\mu(\inr(u) \plus \inl(v)) = 1 + \max (1
    + \mu(u), 1 + \mu(w)) > 1 + \max(\mu(v), \mu(u)) = \mu(\inlr(v,u))$

  \item \eqref{ll:rusuminr} $\mu(\inr(u) \plus \inr(w)) = 1 + \max (1 +
    \mu(u), 1 + \mu(w)) = 1 + 1 + \max (\mu(u), \mu(w)) = \mu(\inr(u
    \plus w))$

  \item \eqref{ll:rusuminrinlr} $\mu(\inr(u) \plus \inlr(v,w)) = 1 +
    \max (1 + \mu(u), 1 + \max(\mu(v),\mu(w)))$\\ $ = \max (2 + \mu(u),
    2 + \mu(v), 2 + \mu(w)) \geq \max(1 + \mu(v),2 + \mu(u), 2 + \mu(w))
    = 1 + \max(\mu(v), 1 + \max(\mu(u), \mu(w))) = \mu(\inlr(v,u \plus
    w))$

  \item \eqref{ll:rusuminlrinl} $\mu(\inlr(t,u) \plus \inl(v)) = 1 +
  \max (1 + \max(\mu(t),\mu(u)), 1 + \mu(v)))$\\ $ = \max(2 + \mu(t),2
  + \mu(u),2 + \mu(v)) \geq \max(2 + \mu(t),2 + \mu(v), 1 + \mu(u)) =
  1 + \max(1+ \max(\mu(t),\mu(v)), \mu(u)) = \mu(\inlr(t \plus v,u))$

\item \eqref{ll:rusuminlrinr} $\mu(\inlr(t,u) \plus \inr(w)) = 1 +
\max (1 + \max(\mu(t),\mu(u)), 1 + \mu(w)))$\\ $ = \max(2 + \mu(t),2
+ \mu(u),2 + \mu(w)) \geq \max(1 + \mu(t),2 + \mu(u), 2 + \mu(w)) =
  1 + \max(\mu(t), 1+ \max(\mu(u),\mu(w)))) = \mu(\inlr(t, u \plus
  w))$

\item \eqref{ll:rusuminlr} $\mu(\inlr(t,u) \plus \inlr(v,w)) = 1 +
  \max (1 + \max(\mu(t),\mu(u)), 1 + \max(\mu(v),\mu(w)))$\\ $= 2 +
  \max (\mu(t),\mu(u), \mu(v),\mu(w))$\\ $= 1 + \max(1+
  \max(\mu(t),\mu(v)), 1 + \max(\mu(u),\mu(w))) = \mu(\inlr(t \plus
  v,u \plus w))$

\item \eqref{ll:rubulletstar} $\mu(a \bullet b.\star) = 2 > 1 = \mu((a
  \times b).\star)$

\item \eqref{ll:rubulletlam} $\mu(a \bullet \lambda \abstr{x} t) = 2 +
  \mu(t) =\mu(\lambda \abstr{x} a \bullet t)$

\item \eqref{ll:rubulletinl} $\mu(a \bullet \inl(t)) = 2 + \mu(t)
  =\mu(\inl(a \bullet t))$

\item \eqref{ll:rubulletinr} $\mu(a \bullet \inr(t)) = 2 + \mu(t) =
  \mu(\inr(a \bullet t))$

\item \eqref{ll:rubulletinlr} $\mu(a \bullet \inlr(t,u)) = 2 +
  \max(\mu(t),\mu(u)) = \mu(\inlr(a \bullet t,a \bullet u))$
  \qedhere
  \end{itemize} 
\end{proof}

Thus, we
introduce a second measure $\nu$ that gives more weight to the lower
$\lambda$ that to the lower $\plus$, so that $\nu((\lambda \abstr{x}x)
\plus (\lambda \abstr{x}x)) = 3$ and $\nu(\lambda \abstr{x}(x \plus
x)) = 2$.

\begin{defi}[The measure $\nu$]
  \label{measureofaproofNu}~
  \begin{align*}
    \nu(x) &= 0 \\
    \nu(t \plus u) &= 1 + 2 \max(\nu(t), \nu(u)) \\
    \nu(a \bullet t) &= 1 + 2 \nu(t) \\
    \nu(a.\star) &= 1 \\
    \nu(\elimone(t,u)) &= 1 \\
    \nu(\lambda \abstr{x} t) &= 1 + \nu(t) \\
    \nu(t~u) &= 1 \\
    \nu(\inl(t)) &= 1 + \nu(t) \\
    \nu(\inr(t)) &= 1 + \nu(t) \\
    \nu(\inlr(t,u)) &= 1 + \max(\nu(t), \nu(u)) \\
    \nu(\elimplus(t,\abstr{y} u,\abstr{z} v)) &= 1 \\
    \nu(\elimplus^{nd}(t,\abstr{y} u,\abstr{z}v)) &= 1
  \end{align*}
\end{defi}

And we prove that, for the rules \eqref{ll:rusumstar} to
\eqref{ll:rubulletinlr}, if $t$ reduces to $u$ by a reduction step
at the root, then $\nu(t) > \nu(u)$.

  \begin{lem}
  \label{lem:mured3}
  For the rules \eqref{ll:rusumstar} to \eqref{ll:rubulletinlr}, if $t$
  reduces to $u$ by a reduction step at the root, then $\nu(t) >
  \nu(u)$.
\end{lem}
\begin{proof}
  We check the rules one by one.
  \begin{itemize}

    \item \eqref{ll:rusumstar} $\nu({a.\star} \plus b.\star) = 3 > 1 =
      \nu((a+b).\star)$

    \item \eqref{ll:rusumlam} $\nu((\lambda \abstr{x}t) \plus (\lambda
      \abstr{x}u)) = 1 + 2 \max(1 + \nu(t), 1 + \nu(u)) = 3 + 2
      \max(\nu(t),\nu(u)) > 1 + 1 + 2 \max(\nu(t),\nu(u)) = \nu(\lambda
      \abstr{x}(t \plus u))$

    \item \eqref{ll:rusuminl} $\nu(\inl(t) \plus \inl(v)) = 1 +2 \max (1 +
      \nu(t), 1 + \nu(v))$\\ $= \max (3 + 2 \nu(t), 3 + 2 \nu(v)) > \max
      (2 + 2 \nu(t), 2 + 2 \nu(v)) = 1 + 1 + 2 \max (\nu(t), \nu(v)) =
      \nu(\inl(t \plus v))$

    \item \eqref{ll:rusuminlinr} $\nu(\inl(t) \plus \inr(w)) = 1 +2 \max
      (1 + \nu(t), 1 + \nu(w))$\\ $ = \max (3 + 2 \nu(t), 3 + 2 \nu(w)) >
      \max (1 + \nu(t),1 + \nu(w)) = 1 + \max (\nu(t),\nu(w)) =
      \nu(\inlr(t,w))$

    \item \eqref{ll:rusuminlinlr} $\nu(\inl(t) \plus \inlr(v,w)) = 1 + 2
      \max (1 + \nu(t), 1 + \max(\nu(v),\nu(w)))$\\ $= \max (3 + 2
      \nu(t),3 + 2 \nu(v),3 + 2 \nu(w)) > \max(2 + 2 \nu(t), 2 + 2\nu(v)),
    1 + \nu(w)) = 1 + \max(1+ 2 \max(\nu(t),\nu(v)), \nu(w)) =
    \nu(\inlr(t \plus v,w))$

  \item \eqref{ll:rusuminrinl} $\nu(\inr(u) \plus \inl(v)) = 1 + 2 \max
    (1 + \nu(u), 1 + \nu(w))$\\ $ = \max(3 + 2 \nu(u),3 + 2 \nu(v)) >
    \max(1 + \nu(v), 1 + \nu(u)) = 1 + \max(\nu(v), \nu(u)) =
    \nu(\inlr(v,u))$

  \item \eqref{ll:rusuminr} $\nu(\inr(u) \plus \inr(w)) = 1 +2 \max (1 +
    \nu(u), 1 + \nu(w))$\\ $= \max (3 + 2 \nu(u), 3 + 2 \nu(w)) > \max
    (2 + 2 \nu(u), 2 + 2 \nu(w)) = 1 + 1 + 2 \max (\nu(u), \nu(w)) =
    \nu(\inr(u \plus w))$

  \item \eqref{ll:rusuminrinlr} $\nu(\inr(u) \plus \inlr(v,w)) = 1 +2
    \max (1 + \nu(u), 1 + \max(\nu(v),\nu(w)))$\\ $ = \max (3 + 2
    \nu(u), 3 + 2 \nu(v), 3 + 2 \nu(w)) > \max(1 + \nu(v),2 + 2 \nu(u),
    2 + 2 \nu(w)) = 1 + \max(\nu(v), 1 + 2 \max(\nu(u), \nu(w))) =
    \nu(\inlr(v,u \plus w))$

  \item \eqref{ll:rusuminlrinl} $\nu(\inlr(t,u) \plus \inl(v)) = 1 +2
  \max (1 + \max(\nu(t),\nu(u)), 1 + \nu(v)))$\\ $ = \max(3 + 2
  \nu(t),3 + 2 \nu(u),3 + 2 \nu(v)) > \max(2 + 2 \nu(t),2 + 2 \nu(v),
  1 + \nu(u)) = 1 + \max(1+ 2 \max(\nu(t),\nu(v)), \nu(u)) =
  \nu(\inlr(t \plus v,u))$

\item \eqref{ll:rusuminlrinr} $\nu(\inlr(t,u) \plus \inr(w)) = 1 +2 \max (1
+ \max(\nu(t),\nu(u)), 1 + \nu(w)))$\\ $ = \max(3 + 2 \nu(t),3 + 2
\nu(u),3 + 2 \nu(w)) > \max(1 + \nu(t),2 + 2 \nu(u), 2 + 2 \nu(w)) =
 1 + \max(\nu(t), 1+ 2 \max(\nu(u),\nu(w)))) = \nu(\inlr(t, u \plus
 w))$

\item \eqref{ll:rusuminlr} $\nu(\inlr(t,u) \plus \inlr(v,w)) = 1 +2
  \max (1 + \max(\nu(t),\nu(u)), 1 + \max(\nu(v),\nu(w)))$\\ $ =
  \max(3 + 2 \nu(t), 3 + 2 \nu(u),3 + 2 \nu(v), 3 + 2 \nu(w)) > \max(2
  + 2 \nu(t),2 + 2 \nu(v),2 + 2 \nu(u), 2 + 2 \nu(w)) = 1 + \max(1+ 2
  \max(\nu(t),\nu(v)), 1 + 2 \max(\nu(u),\nu(w))) = \nu(\inlr(t \plus
  v,u \plus w))$

\item \eqref{ll:rubulletstar} $\nu(a \bullet b.\star) = 3 > 1 = \nu((a
  \times b).\star)$

\item \eqref{ll:rubulletlam} $\nu(a \bullet \lambda \abstr{x} t) = 3 +
  2 \nu(t) > 2 + 2 \nu(t) = \nu(\lambda \abstr{x} a \bullet t)$

\item \eqref{ll:rubulletinl} $\nu(a \bullet \inl(t)) = 3 + 2 \nu(t) >
  2 + 2 \nu(t) = \nu(\inl(a \bullet t))$

\item \eqref{ll:rubulletinr} $\nu(a \bullet \inr(t)) = 3 + 2 \nu(t) >
  2 + 2 \nu(t) = \nu(\inr(a \bullet t))$

\item \eqref{ll:rubulletinlr} $\nu(a \bullet \inlr(t,u)) = 3 + 2
  \max(\nu(t),\nu(u)) > 2 + 2 \max(\nu(t),\nu(u)) = \nu(\inlr(a
  \bullet t,a \bullet u))$ \qedhere
  \end{itemize} 
\end{proof}

Finally, we consider the lexicographic order $\succ$ on $\mu$ and $\nu$.
\medskip

\begin{defi}[Lexicographic order]~
\begin{itemize}
\item
The preorder relation $\succcurlyeq$ is defined as follows: $t
\succcurlyeq u$ if either $\mu(t) > \mu(u)$ or ($\mu(t) = \mu(u)$ and
$\nu(t) \geq \nu(u)$),
\item
and the well-founded strict order $\succ$ as follows: $t \succ u$ if
either $\mu(t) > \mu(u)$ or ($\mu(t) = \mu(u)$ and $\nu(t) > \nu(u)$).
\end{itemize}
\end{defi}

It is easy to prove that if $t$ reduces to $u$ by a reduction step at
the root, then $t \succ u$.  This is sufficient to prove that all
sequences $t_0, t_1, \ldots$ such that $t_i$ reduces to $t_{i+1}$ by a
reduction step at the root are finite. To generalize this result to
all reduction sequences, we need to prove that the relation
$\succcurlyeq$ is a simplification pre-order, that is a pre-order
verifying the monotony and sub-term properties \cite[Section 3,
Definition 3]{Nachum}. Then, we can conclude with \cite[Second
Termination Theorem]{Nachum}.  Yet, as the relation $\succ$ is
well-founded, we also have a direct proof, using a multi-set ordering
on the multi-sets of terms.

\begin{lem}\label{lex}
If $t$
reduces to $u$ by a reduction step at the root, then
$t \succ u$.
\end{lem}

\begin{proof}
From Lemmas~\ref{lem:mured1}, \ref{lem:mured2}, and
\ref{lem:mured3}. \qedhere
\end{proof}

We write $t_{|p}$ the sub-term of $t$ at position $p$, and $[u/p]t$ for the
grafting of $u$ in $t$ at the position $p$.

\begin{lem}[Simplification]
  \label{lem:simplification}~
  \begin{enumerate}
    \item For all positions $p$ of $t$, $u\succcurlyeq u'$ implies
      $[u/p]t \succcurlyeq [u'/p]t$ (monotony).
    \item For all positions $p$ of $t$, $t \succcurlyeq t_{|p}$ (sub-term).
  \end{enumerate}
\end{lem}
\begin{proof}
  We first prove by induction on the depth of the position $p$ that 
  \begin{itemize}
    \item For all positions $p$ of $t$, $\mu(u) \geq \mu(u')$ implies
      $\mu([u/p]t) \geq \mu([u'/p]t)$.

    \item For all positions $p$ of $t$, $\mu(t) \geq \mu(t_{|p})$. 

    \item For all positions $p$ of $t$, $\nu(u) \geq \nu(u')$ implies
      $\nu([u/p]t) \geq \nu([u'/p]t)$.

    \item For all positions $p$ of $t$, $\nu(t) \geq \nu(t_{|p})$. 
  \end{itemize}

  As all cases are trivial, we give only one: in the proof of the first
  property, if the term $t$ is $\inlr(t_1,t_2)$ and the position $p$ in $t$
  is the position $p'$ in $t_1$ (that is $p = 1 \cdot p'$), by induction
  hypothesis, we get
  $\mu([u/p']t_1) \geq \mu([u'/p']t_1)$.
  Thus,
  $\mu([u/p]t) = 1 + \max(\mu([u/p']t_1),\mu(t_2)) \geq
  1 + \max(\mu([u'/p']t_1),\mu(t_2))
  = \mu([u'/p]t)$.

  Thus, the preorder $\succcurlyeq$ also verifies the monotony and sub-term
  properties. \qedhere
\end{proof}

\begin{rem}
  Note that in a higher-order language, that is in a language containing
  bound variables, the notion of sub-term at position $p$ is not
  completely well-defined.

  Indeed, $\lambda \abstr{x} (f x)$ is the same term as $\lambda
  \abstr{y} (f y)$, but the sub-term at the position $[0]$ of the first
  is $f x$ and the sub-term at the same position of the second is $f
  y$.  In this case this is not an issue here as bound and free variable
  names are immaterial to the measures $\mu$ and $\nu$, hence to the
  relations $\succ$ and $\succcurlyeq$.

  A rigorous formulation \cite{IwamiToyama}, would use a translation that
  neutralizes the bound variables by adding a constant $c_A$ for each
  proposition $A$ and replacing the bound variables of type $A$ with
  $c_A$. Then the names of bound variables can be omitted and a proof
  such as $\lambda \abstr{x} (f x)$ is transformed into the proof
  $\lambda (f c)$.  More generally the translation of a term is defined
  as follows:
  \begin{align*}
    \| x \| &= x &
    \| t \plus u \| &= \| t \| \plus \| u \|\\
    \| a \bullet t \| &= a \bullet \| t \| &
    \| a.\star \| &= a . \star\\
    \| \elimone(t,u) \| &= \elimone (\| t \|, \| u \|) &
    \| \lambda \abstr{x} t \| &= \lambda \| (c/x)t \|\\
    \| t~u \| &= \| t \| \| u \| &
    \| \inl(t) \| &= \inl(\| t\|)\\
    \| \inr(t) \| &= \inr(\| t\|)&
    \| \inlr(t,u) \| &= \inlr(\| t\|,\|u\|)\\
    \| \elimplus(t,\abstr{y}u,\abstr{z}v) \| &= \elimplus(\| t \|,\| (c/y)u \|,\| (c/z)v \|)&
    \| \elimplus^{nd}(t,\abstr{y}u,\abstr{z}v) \| &= \elimplus^{nd}(\| t \|,\| (c/y)u \|,\| (c/z)v \|)
  \end{align*}
  And \autoref{lem:simplification} can then be rigorously formulated on
  translated terms.
\end{rem}

\begin{thm}[Termination]  
\label{thm:termination}
Let $\Gamma\vdash t:A$, then $t$ strongly terminates.
\end{thm}
\begin{proof}
  To each term $t$, we associate the multi-set $sub(t)$ of the sub-terms of
  (the translation of) $t$.

  If $t$ reduces to $u$ then there exists a position $p$ and a rewrite
  rule $l \lra r$, such that $t_{|p} = \sigma l$ and $u = [\sigma r
  /p]t$. Thus, from \autoref{lex}, $t_{|p} = \sigma l \succ \sigma r$.

  The multi-set $sub(u)$ is obtained from the multi-set $sub(t)$ in the
  following way:
  \begin{itemize}
    \item the sub-term $t_{|p}$ of $t$ whose position is $p$ 
      is replaced by $\sigma r$ and $t_{|p} \succ \sigma r$, 

    \item  all the strict sub-terms $v$ of $\sigma r$ are added,
      for such sub-terms, by
      the sub-term property (\autoref{lem:simplification}),
      $ t_{|p} \succ \sigma r \succcurlyeq v $, hence
      $ t_{|p} \succ v$,

    \item the sub-terms of $t$ whose position is below $p$ 
      are removed,

    \item the sub-terms of $t$ whose position is not comparable to $p$
      remain in $sub(u)$,

    \item the sub-terms $t_{|p_1}$ of $t$ such that $p_1$ is above $p$, that is
      $p = p_1 \cdot p_2$ 
      are such that ${t_{|p_1}}_{|p_2} = t_{|p}$ 
      and they are replaced with a sub-term of the form
      $[\sigma r/p_2]t_{|p_1}$, we have $
      t_{|p_1}
      = 
      [t_{|p}/p_2]t_{|p_1}
      \succcurlyeq
      [\sigma r/p_2]t_{|p_1}$ by monotony
      (\autoref{lem:simplification}).
  \end{itemize}

  To sum up:
  the term $t_{|p}$ is replaced with terms that are strictly smaller, 
  some terms are removed,
  some terms are kept,
  and
  some terms are replaced with terms
  that are smaller or equal.
  Thus,
  $sub(u)$ is strictly smaller than $sub(t)$ for the multi-set ordering.

  As the multi-set ordering associated to $\succ$ is well-founded the
  sequence is finite. \qedhere
\end{proof}

\subsection{Vectors and matrices}\label{subsec:vectorsandmatrices}

This section shows how vectors and matrices can be
expressed as proofs in the fragment of the quantum in-left-right-calculus without the non-deterministic symbol $\elimplus^{nd}$. It is
kept short as it is just a transposition of results presented in
slightly different contexts in
\cite{DiazcaroDowekTCS23,DiazcaroDowekMSCS24}.

\begin{defi}[Vector propositions]
The set $\mathcal{V}$ is inductively defined as follows: $\one \in
\mathcal{V}$, and if $A$ and $B$ are in $\mathcal{V}$, then so is $A
\oplus B$.
\end{defi}

\begin{defi}[Dimension of a vector propositon in $\mathcal{V}$]
To each proposition $A \in \mathcal{V}$, we associate a positive
natural number $d(A)$, which is the number of occurrences of the
symbol $\one$ in A: $d(\one) = 1$ and $d(B \oplus C) = d(B) + d(C)$.
\end{defi}

If $A \in \mathcal{V}$ and $d(A) = n$, then the closed irreducible
proofs of $A$ and the vectors of ${\mathbb C}^n$ are in
correspondence: to each closed irreducible proof $t$ of $A$ we
associate a vector $\underline{t}$ of ${\mathbb C}^n$ and to each
vector {\bf u} of ${\mathbb C}^n$, we associate a closed irreducible
proof $\overline{\bf u}^A$ of $A$.

\begin{defi}[Correspondance]
Let $A \in \mathcal{V}$ and $d(A) = n$.  To each closed irreducible
proof $t$ of $A$ we associate a vector $\underline{t}$ of ${\mathbb
  C}^n$ as follows.
\begin{itemize}
\item If $A = \one$, then $t = a.\star$. We let $\underline{t} =
\left(\begin{smallmatrix} a \end{smallmatrix}\right)$.

\item If $A = A_1 \oplus A_2$, then

    --- if $t = \inlr(u,v)$, we let $\underline{t}$ be the vector
      with two blocks $\underline{u}$ and $\underline{v}$:
      $\underline{t} = \left(\begin{smallmatrix}
        \underline{u}\\\underline{v} \end{smallmatrix}\right)$,

    --- if $t = \inl(u)$, we let $\underline{t}$ be the vector with
      two blocks $\underline{u}$ and ${\bf 0}$: $\underline{t} =
      \left(\begin{smallmatrix} \underline{u}\\{\bf 0} \end{smallmatrix}\right)$,

    --- if $t = \inr(v)$, we let $\underline{t}$ be the vector with
      two blocks ${\bf 0}$ and $\underline{v}$: $\underline{t} =
      \left(\begin{smallmatrix} {\bf 0}\\\underline{v} \end{smallmatrix}\right)$.
\end{itemize}
To each vector ${\bf u} \in {\mathbb C}^n$, we associate a closed irreducible
proof $\overline{\bf u}^A$ without $\inl$ or $\inr$.
\begin{itemize}
\item If $n = 1$, then ${\bf u} =
  \left(\begin{smallmatrix} a \end{smallmatrix}\right)$.
  We let $\overline{\bf u}^{\one} = a.\star$.

\item If $n > 1$, then $A = A_1 \oplus A_2$, let $n_1$ and $n_2$ be
  the dimensions of $A_1$ and $A_2$.  Let ${\bf u}_1$ and ${\bf u}_2$
  be the two blocks of ${\bf u}$ of $n_1$ and $n_2$ lines, so ${\bf u}
  = \left(\begin{smallmatrix} {\bf u}_1\\ {\bf
      u}_2\end{smallmatrix}\right)$.  We let $\overline{\bf u}^{A_1
      \oplus A_2} = \inlr(\overline{{\bf u}_1}^{A_1}, \overline{{\bf
        u}_2}^{A_2})$.
\end{itemize}
\end{defi}

Let $A \in \mathcal{V}$.  We extend the definition of $\underline{t}$
to any closed proof of $A$, $\underline{t}$ is by definition
$\underline{t'}$, where $t'$ is the irreducible form of $t$.

\begin{rem}
This correspondence is not one-to-one as the proofs $\inl(1.\star)$
and $\inlr(1.\star,0.\star)$ both represent the vector
$\left(\begin{smallmatrix} 1 \\ 0 \end{smallmatrix}\right)$.

More generally, if $d(A) = n$ and $d(B) = p$, the closed proofs of $A$
are in correspondence with the vectors of ${\mathbb C}^n$, those of
$B$ with the vectors of ${\mathbb C}^p$, and those of $A \oplus B$
with the vectors of ${\mathbb C}^{n+p} = {\mathbb C}^n \times {\mathbb
  C}^p$.  The proofs of the form $\inl(u)$ are in correspondence with
the vectors of ${\mathbb C}^n \times {\mathbb C}^p$ that are in the
image of the first injection of ${\mathbb C}^n$ into ${\mathbb C}^n +
{\mathbb C}^p\simeq{\mathbb C}^n \times {\mathbb C}^p$, that is the
vectors of the form $\left(\begin{smallmatrix} {\bf u} \\ {\bf
    0} \end{smallmatrix}\right)$.  In the same way, the proofs of the
form $\inr(v)$ are in correspondence with the vectors of ${\mathbb
  C}^n \times {\mathbb C}^p$ that are in the image of the second
injection of ${\mathbb C}^p$ into ${\mathbb C}^n + {\mathbb
  C}^p\simeq{\mathbb C}^n \times {\mathbb C}^p$, that is the vectors
of the form $\left(\begin{smallmatrix} {\bf 0} \\ {\bf
    u} \end{smallmatrix}\right)$.

Having two representations of the vector $\left(\begin{smallmatrix} 1
  \\ 0 \end{smallmatrix}\right)$ may be a way to distinguish between
the proof $\inl(1.\star)$ that represents the classical bit
$\boolzero$, for example the result of a measure, and the proof
$\inlr(1.\star,0.\star)$ that represents Qbit $1. \ket{0} +
0. \ket{1} = \ket{0}$, for instance the result of a computation, that
happens to have a zero coordinate.
\end{rem}

The next proposition shows that the symbol $\plus$ expresses the sum
of vectors and the symbol $\bullet$, the product of a vector by a
scalar.

\begin{prop}
\label{parallelsum}
\label{parallelprod}
Let $A \in \mathcal{V}$ and $d(A) = n$.  Let $u$ and $v$ be two closed
proofs of $A$ and $a$ an scalar.  Then, 
  $\underline{u \plus v} = \underline{u} + \underline{v}$.
  and
  $\underline{a \bullet u} = a \underline{u}$.
\end{prop}
\begin{proof}
  ~
  \begin{enumerate}
    \item 
      By induction on $n$. 

      \begin{itemize}
	\item 
	  If $n = 1$, then $u \lra^* a.\star$, $v \lra^* b.\star$, $\underline{u} =
	  \left(\begin{smallmatrix} a \end{smallmatrix}\right)$, $\underline{v}
	  = \left(\begin{smallmatrix} b \end{smallmatrix}\right)$.  Thus,
	  $\underline{u \plus v} = \underline{{a.\star} \plus b.\star} =
	  \underline{(a + b).\star} = \left(\begin{smallmatrix} a +
	  b \end{smallmatrix}\right) = \left(\begin{smallmatrix}
	  a \end{smallmatrix}\right) + \left(\begin{smallmatrix}
	  b \end{smallmatrix}\right) = \underline {u} + \underline{v}$.

	\item 
	  If $n > 1$, then $A = A_1 \oplus A_2$, let $n_1$ and $n_2$ be the
	  dimensions of $A_1$ and $A_2$.

	  We have nine cases to consider.
	  \begin{itemize}
	    \item 
	      If  $u \lra^* \inlr(u_1,u_2)$, $v \lra^* \inlr(v_1,v_2)$,
	      $\underline{u} = \left(\begin{smallmatrix}
		  \underline{u_1} \\ \underline{u_2}
	      \end{smallmatrix}\right)$ and 
	      $\underline{v} = \left(\begin{smallmatrix} \underline{v_1} \\ \underline{v_2} 
	      \end{smallmatrix}\right)$, using the induction hypothesis,
	      and the rule \eqref{ll:rusuminlr},
	      $\underline{u \plus v}
	      = \underline{\inlr(u_1,u_2) \plus \inlr(v_1,v_2)}
	      = \underline{\inlr(u_1 \plus v_1,u_2 \plus v_2)}
	      = \left(\begin{smallmatrix} \underline{u_1 \plus v_1} \\
		  \underline{u_2 \plus v_2}
	      \end{smallmatrix}\right)
	      = \left(\begin{smallmatrix} \underline{u_1} + \underline{v_1} \\
		  \underline{u_2} + \underline{v_2}
	      \end{smallmatrix}\right)
	      = \left(\begin{smallmatrix} \underline{u_1} \\
		  \underline{u_2}
	      \end{smallmatrix}\right)
	      +
	      \left(\begin{smallmatrix} \underline{v_1} \\
		  \underline{v_2}
	      \end{smallmatrix}\right)
	      =  \underline{u} + \underline{v}$.

	    \item
	      If  $u \lra^* \inl(u_1)$, $v \lra^* \inr(v_2)$,
	      $\underline{u} = \left(\begin{smallmatrix}
		  \underline{u_1} \\ {\bf 0}
	      \end{smallmatrix}\right)$ and 
	      $\underline{v} = \left(\begin{smallmatrix} {\bf 0} \\ \underline{v_2} 
	      \end{smallmatrix}\right)$, using 
	      the rule \eqref{ll:rusuminlinr},
	      $\underline{u \plus v}
	      = \underline{\inl(u_1) \plus \inr(v_2)}
	      = \underline{\inlr(u_1, v_2)}
	      = \left(\begin{smallmatrix} \underline{u_1} \\
		  \underline{v_2}
	      \end{smallmatrix}\right)
	      = \left(\begin{smallmatrix} \underline{u_1} \\
		  {\bf 0}
	      \end{smallmatrix}\right)
	      +
	      \left(\begin{smallmatrix} {\bf 0} \\
		  \underline{v_2}
	      \end{smallmatrix}\right)
	      =  \underline{u} + \underline{v}$.

	    \item The other cases are similar using
	      rules \eqref{ll:rusuminl} to \eqref{ll:rusuminlr}.
	  \end{itemize}
      \end{itemize}
    \item
      By induction on $n$.

      \begin{itemize}
	\item 
	  If $n = 1$, then $u \lra^* b.\star$, $\underline{u} =
	  \left(\begin{smallmatrix} b \end{smallmatrix}\right)$.  Thus,
	  $\underline{a \bullet u} = \underline{a \bullet b.\star} =
	  \underline{(a \times b).\star} = \left(\begin{smallmatrix} a \times
	  b \end{smallmatrix}\right) = a \left(\begin{smallmatrix}
	  b \end{smallmatrix}\right) = a \underline {u}$.

	\item 
	  If $n > 1$, then $A = A_1 \oplus A_2$, let $n_1$ and $n_2$ be the
	  dimensions of $A_1$ and $A_2$. We have
	  three cases to consider. If 
	  $u \lra^* \inlr(u_1,u_2)$,
	  $\underline{u} = \left(\begin{smallmatrix} \underline{u_1}
	      \\ \underline{u_2}
	  \end{smallmatrix}\right)$,
	  then, using the induction hypothesis and rule \eqref{ll:rubulletinlr},
	  $\underline{a \bullet u} = 
	  \underline{a \bullet \inlr(u_1,u_2)}
	  =
	  \underline{\inlr(a \bullet u_1,a \bullet u_2)}
	  =
	  \left(\begin{smallmatrix} \underline{a \bullet u_1} \\
	      \underline{a \bullet u_2}
	  \end{smallmatrix}\right)
	  = 
	  \left(\begin{smallmatrix} a \underline{u_1}  \\
	      a \underline{u_2} 
	  \end{smallmatrix}\right)
	  =
	  a \left(\begin{smallmatrix} \underline{u_1} \\
	      \underline{u_2}
	  \end{smallmatrix}\right)
	  = a \underline{u}$.

	  The other cases are similar using rules \eqref{ll:rubulletinl} to
	  \eqref{ll:rubulletinlr}. 
	  \qedhere
      \end{itemize}
  \end{enumerate}
\end{proof}

\begin{thm}[Matrices]
\label{matrices}
Let $A, B \in \mathcal{V}$ with $d(A) = m$ and $d(B) = n$ and let $M$
be a matrix with $m$ columns and $n$ lines, then there exists a closed
proof $t$ of $A \multimap B$ such that, for all vectors ${\bf u} \in
{\mathbb C}^m$, $\underline{t~\overline{\bf u}^A} = M {\bf u}$.
\end{thm}
\begin{proof}
  By induction on $A$.
  \begin{itemize}

    \item If $A = \one$, then $M$ is a matrix of one column and
      $n$ lines. Hence, it is also a vector of $n$ lines.
      We take
      $$
	t = \lambda \abstr{x} \elimone(x,\overline{M}^B)
      $$
      Let ${\bf u} \in {\mathbb C}^1$, ${\bf u}$ has the form
      $\left(\begin{smallmatrix}  a \end{smallmatrix}\right)$ and 
      $\overline{\bf u}^A = a.\star$.

      Then, using
      \autoref{parallelprod}, we have
      \[
	\underline{t~\overline{\bf u}^A} 
	= \underline{\elimone(\overline{\bf u}^A,\overline{M}^B)}
	= \underline{\elimone(a.\star,\overline{M}^B)}
	= \underline{a \bullet \overline{M}^B}
	= a \underline{\overline{M}^B} 
	= a M = M
	\left(\begin{smallmatrix}  a\end{smallmatrix}\right) =
	M {\bf u}
      \]

    \item If $A = A_1 \oplus A_2$, then let $m_1 = d(A_1)$ and $m_2 =
      d(A_2)$.  Let $M_1$ and $M_2$ be the two blocks of $M$ of $m_1$ and
      $m_2$ columns, so $M = \left(\begin{smallmatrix} M_1 &
      M_2\end{smallmatrix}\right)$.

      By induction hypothesis, there exist closed proofs $t_1$ and $t_2$ of
      the proposition
      $A_1 \multimap B$ and $A_2 \multimap B$
      such that, for all
      vectors ${\bf u}_1 \in {\mathbb C}^{m_1}$ and 
      ${\bf u}_2 \in {\mathbb C}^{m_2}$,
      we have
      $\underline{t_1~\overline{{\bf u}_1}^{A_1}} = M_1 {\bf u}_1$ and
      $\underline{t_2~\overline{{\bf u}_2}^{A_2}} = M_2 {\bf u}_2$.  We take
      $$
	t = \lambda \abstr{x} (\elimplus(x, \abstr{y} (t_1~y),
	\abstr{z} (t_2~z)))
      $$
      Let ${\bf u} \in {\mathbb C}^m$, and ${\bf u}_1$ and ${\bf
      u}_2$ be the two blocks of $m_1$ and $m_2$ lines of ${\bf u}$, so ${\bf
      u} = \left(\begin{smallmatrix} {\bf u}_1 \\ {\bf
      u}_2 \end{smallmatrix}\right)$, and $\overline{\bf u}^A =
      \inlr(\overline{{\bf u}_1}^{A_1},\overline{{\bf u}_2}^{A_2})$.

      Then, using \autoref{parallelsum}, 
      \begin{align*}
	\underline{t~\overline{\bf u}^A}
	&= \underline{\elimplus(\inlr(\overline{{\bf u}_1}^{A_1},\overline{{\bf u}_2}^{A_2}), \abstr{y} (t_1~y), \abstr{z} (t_2~z))}
	\\
	&= \underline{t_1~\overline{{\bf u}_1}^{A_1} \plus t_2~\overline{{\bf u}_2}^{A_2}}
	= \underline{t_1~\overline{{\bf u}_1}^{A_1}} + \underline{t_2~\overline{{\bf u}_2}^{A_2}}
	\\
	&= M_1 {\bf u}_1 + M_2 {\bf u}_2
	= \left(\begin{smallmatrix} M_1 & M_2 \end{smallmatrix}\right) \left(\begin{smallmatrix} {\bf u}_1 \\ {\bf u}_2  \end{smallmatrix}\right)
	= M {\bf u}
      \qedhere
      \end{align*}
  \end{itemize}
\end{proof}

\subsection{Linearity}
\label{seclinearity}

\autoref{matrices} shows that all linear maps can be
expressed as proofs in the quantum in-left-right-calculus.
We now prove the converse: that every proof in the fragment of the quantum in-left-right-calculus without $\elimplus^{nd}$ expresses a linear map.
  This proof is an
  adaptation to the quantum in-left-right-calculus of the proof given
  in \cite{DiazcaroDowekMSCS24}.
  The omitted proofs are given in \autoref{app:linearity}.

\begin{thm}\label{thm:converse}
  Let $A, B \in {\mathcal V}$, such that $d(A) = m$ and $d(B) = n$, 
  and  $t$ be a closed proof of $A \multimap B$.
  Then the function $F$ from ${\mathbb C}^m$ to ${\mathbb C}^n$,
  defined as
  $F({\bf u}) = \underline{t~\overline{\bf u}^A}$ is linear.
  \qed
\end{thm}

We could attempt to generalize this statement and prove that these
properties hold for all proofs, whatever the proved proposition. But
this generalization is too strong, as in the introduction rule of
$\multimap$, the subproof $t$ may contain more free variables than the
proof $\lambda \abstr{x} t$.  Thus, if, for example, $A = \one$, $B =
(\one \multimap \one) \multimap \one$, and $t = \lambda \abstr{x}
\lambda \abstr{y}y~x$, we have
\(
    t~({1.\star} \plus 2.\star) \lra^* \lambda \abstr{y} y~3.{\star},
  \) and
  \(
    (t~1.\star) \plus (t~2.\star) \lra^* \lambda \abstr{y} {(y~1.\star)} \plus {(y~2.\star)}
  \)
and these two irreducible proofs are different.

However, while the proofs $\lambda \abstr{y} y~3.\star$ and $\lambda
\abstr{y} {(y~1.\star)} \plus {(y~2.\star)}$ are different, if we put
them in the context $\_~\lambda \abstr{z} z$, then both proofs
$(\lambda \abstr{y} y~3.\star)~\lambda \abstr{z} z$ and $(\lambda
\abstr{y} {(y~1.\star)} \plus {(y~2.\star)})~\lambda \abstr{z} z$
reduce to $3.\star$.  This leads us to introduce a notion of
observational equivalence ($\sim$) and prove the general theorem.

\begin{defi}[Observational equivalence]\label{def:obseq}
Two proofs $t_1$ and $t_2$ of a proposition $B$ are \emph{observationally equivalent}, $t_1 \sim t_2$, if for all propositions
$C$ in $\mathcal{V}$ and for all proofs $c$ such that $\_:B \vdash
c:C$, we have $(t_1/\_)c \equiv (t_2/\_)c$.
\end{defi}

\begin{thm}
\label{corollary2}
Let $A$ and $B$ be propositions, $t$ a closed proof of $A \multimap
B$, and $u_1$ and $u_2$ two closed proofs of $A$.  Then
\(
  t~(u_1 \plus u_2) \sim t~u_1 \plus t~u_2
\) and
\(
  t~(a\bullet u_1) \sim a\bullet t~u_1
\).
\qed
\end{thm}

\subsection{Measurement}
\label{sec:quantumcomputing}

  Now that we have linear maps, we almost have a quantum programming
language, but we still need measurement operators. We see, in this
Section, how such operators can be defined, using
$\elimplus^{nd}$.
We consider specific vector
propositions $\Q_n$, inductively defined as $\Q_0 = \one$ and
$\Q_{n+1} = \Q_n \oplus \Q_n$, that represent vectors of ${\mathbb
  C}^{2^n}$.

The expressions $\ket 0$ and $\ket 1$ are notations for 
$\inlr(1.\star,0.\star)$ and $\inlr(0.\star,1.\star)$, proofs of the
proposition $\Q_1 = \one \oplus \one$, that represent the vectors
$\left(\begin{smallmatrix}1\\0\end{smallmatrix}\right)$ and
$\left(\begin{smallmatrix}0\\1\end{smallmatrix}\right)$
respectively.  Note that the proof $a \bullet \ket 0 \plus b \bullet
\ket 1$ reduces to $\inlr(a.\star,b.\star)$ that represent the vector
$\left(\begin{smallmatrix}a\\b\end{smallmatrix}\right)$.

We define ${\bf 0}_n$ as ${\bf 0}_0 = 0.\star$ and ${\bf 0}_{n+1} =
\inlr({\bf 0}_n,{\bf 0}_n)$.

Booleans, as usual, can be encoded by the proposition
$\B=\one\oplus\one$, which coincides with $\Q_1$.  Note that, unlike
in \cite{DiazcaroDowekTCS23,DiazcaroDowekMSCS24}, where the Booleans
are proofs of $\one \oplus \one$ and the Qbits of $\one \odot \one$,
introducing an artificial distinction between classical and quantum
objects, the Boolean and the Qbits here are proofs of the same
proposition.

The Boolean $\boolzero$ and $\boolone$ could be expressed as the Qbits
$\ket 0$ and $\ket 1$, that is $\inlr(1.\star, 0.\star)$ and
$\inlr(0.\star,1.\star)$, respectively, but we prefer to use the
richness of the proof language to express them as the proofs
$\inl(1.\star)$ and $\inr(1.\star)$, respectively.

\begin{defi}[Norm]
Let $t$ be a closed irreducible proof of $\Q_n$, then
we define the norm of $t$ by induction on $n$.
If $n = 0$, then $t = a.\star$ and we let $\|t\|^2 = |a|^2$.
If $n > 0$ and $t = \inlr(u,v)$, we let $\|t\|^2 = \|u\|^2 +
  \|v\|^2$. In the same way, if $t = \inl(u)$, we let $\|t\|^2 =
  \|u\|^2$ and if $t = \inr(v)$, we let $\|t\|^2 = \|v\|^2$.
\end{defi}

This notion of norm can be used to assign probabilities to the
non-deterministic reduction rules~\eqref{ll:ruelimoplusndinlr1} and
\eqref{ll:ruelimoplusndinlr2}. If $t$ and $u$ are closed irreducible
proofs of propositions in $\Q_n$ and $\|t\|^2+\|u\|^2\neq 0$, we
assign the probabilities $\tfrac{\|t\|^2}{\|t\|^2 + \|u\|^2}$ and
$\tfrac{\|u\|^2}{\|t\|^2 + \|u\|^2}$ to the rules
\[
  \elimplus^{nd}(\inlr(t,u),\abstr{x}v,\abstr{y}w)  \longrightarrow  (t/x)v
  \qquad
  \mbox{and}
  \qquad
  \elimplus^{nd}(\inlr(t,u),\abstr{x}v,\abstr{y}w)  \longrightarrow
  (u/y)w.
\]

\noindent
respectively. Otherwise, the proof cannot be reduced.

We take the convention that any closed irreducible proof $u$ of
$\Q_n$, expressing a non-zero vector $\underline{u} \in {\mathbb
C}^{2^n}$, is an alternative expression of the vector or norm $1$
$\tfrac{u}{\|u\|}$. For example, the vector
$\frac 1{\sqrt 2}\left(\begin{smallmatrix}1\\ 
1\end{smallmatrix}\right)$
is expressed as the proof
$\inlr(\nicefrac{1}{\sqrt{2}}.\star,\nicefrac{1}{\sqrt{2}}.\star)$, but also
as the proof $\inlr(1.\star,1.\star)$.

The information erasing and non-deterministic proof
constructor $\elimplus^{nd}$ permits to define the measurement
operators as follows, one returning a Boolean value and
the other the resulting state after measuring the first Qbit.
\begin{align*}
  \pi_n(t)  &=\elimplus^{nd}(t, \abstr{x} \delta^{\Q_{n-1}}({x},{\boolzero}),
  \abstr{y}\delta^{\Q_{n-1}}(y,{\boolone}))\\
  \pi'_n(t)  &=\elimplus^{nd}(t, \abstr{x} \inlr(x,{\bf 0}_n), 
  \abstr{y}\inlr({\bf 0}_n,y))
\end{align*}
where
\begin{align*}
  \delta^{\Q_0}(x,{\bf b}) &= \elimone(x,{\bf b})\\
  \delta^{\Q_{n+1}}(x,{\bf b}) &= \elimplus^{nd}(x, \abstr{y} \delta^{\Q_{n}}(y,{\bf b}), \abstr{z} \delta^{\Q_{n}}(z,{\bf b})).
\end{align*}

If $t$ is an irreducible proof of $\Q_1$ of the form
$\inlr(a.\star,b.\star)$, where $a$ and $b$ are not both $0$, then the
proof $\pi_1(t)$ of $\B$ reduces, with probabilities
$\tfrac{|a|^2}{|a|^2 + |b|^2}$ to $a\bullet\boolzero$, which, under our
convention, represent the same vector as $\boolzero$, as $a \neq 0$
(because the reduction rule has a non-zero probability) and with
probability $\tfrac{|b|^2}{|a|^2 + |b|^2}$, to $b\bullet\boolone$, which,
under our convention, represent the same vector as $\boolone$.  It is
the result of the measurement.

The proof $\pi'_1(t)$ of $\Q_1$ reduces, with the same
probabilities as above, to $\inlr(a . \star,0 . \star)$ and to
$\inlr(0 . \star,b . \star)$, which are the states after the measure.

In the same way, if $t$ is an irreducible proof of $\Q_n$ of the form
$\inlr(u,v)$ where $\|u\|$ and $\|v\|$ are not both $0$, then the
proof $\pi_n(t)$ of $\B$ reduces, with probabilities
$\tfrac{\|u\|^2}{\|u\|^2 + \|v\|^2}$ and $\tfrac{\|v\|^2}{\|u\|^2 +
  \|v\|^2}$ to proofs representing the same vectors as $\boolzero$ and
to $\boolone$ respectively.  It is the result of the partial
measurement of the first Qbit.  The proof $\pi'_n(t)$ of 
$\Q_n$ reduces, with the same probabilities as above, to
proofs representing the same vectors as $\inlr(u,{\bf 0}_{n-1})$ and
$\inlr({\bf 0}_{n-1},v)$.  These are the states after the partial
measure of the first Qbit.

With matrices and measurement operators, we can express all the quantum
algorithms in the quantum in-left-right-calculus.

\section{Application to commuting cuts}\label{sec:commutingcuts}

In calculi with no interstitial rules, there are also commuting cuts,
where an introduction and an elimination rule of some connective
are separated with a blocking rule.  To eliminate those cuts, the
blocking rule is often commuted with the elimination rule below, when
we have the rule $\vee$-i3, it can also be commuted with the
introduction rule above. However, this situation is more complex than with
the sum rule.

\subsection{Blocking rules}
  
%(marking in red the blocking rule)
When an introduction and an elimination rule of some connective $c$
are separated with an elimination rule 
blocking the reduction, we have a commuting cut.
For example
\[
  \irule{\irule{\irule{\pi_1}{\Gamma \vdash A \wedge B}{}
      &
      \irule{\irule{\pi_2}
	{\Gamma, A, C \vdash D}
	{}
      }
      {\Gamma, A \vdash C \Rightarrow D}
      {\mbox{$\Rightarrow$-i}}
    }
    {\Gamma \vdash C \Rightarrow D}
    {\mbox{\color{red} $\wedge$-e1}}
    &
    \irule{\pi_3}
    {\Gamma \vdash C}
    {}
  }
  {\Gamma \vdash D}
  {\mbox{$\Rightarrow$-e}}
\]

Eliminating such commuting cuts is not needed to obtain the
introduction property. Indeed, if $\Gamma$ is empty, the proof $\pi_1$
reduces to a proof ending with a $\wedge$-i rule, and reducing the cut
formed with this rule and the blocking rule yields an
ordinary cut between the $\Rightarrow$-i and the $\Rightarrow$-e
rules.  But it is needed to prove the sub-formula property for the
(open) irreducible proofs and to relate the irreducible proofs in
Natural Deduction and the cut free proofs in Sequent Calculus.
Remind that the sub-formula property expresses that all the
propositions in an irreducible proof are sub-formulas of the
conclusion, this is not the case in the example above. For example,
when $\Gamma = A\wedge B,C,D$, as the proposition $C\Rightarrow D$ is
not a sub-formula of the sequent $A\wedge B,C,D\vdash D$

Commuting the blocking rule with the introduction rule above is not too
difficult when, like in this case, the blocking rule has only one minor
premise. For example, the proof
\[
  \irule{\irule{\pi_1}{\Gamma \vdash A \wedge B}{}
    &
    \irule{\irule{\pi_2}
      {\Gamma, A, C \vdash D}
      {}
    }
    {\Gamma, A \vdash C \Rightarrow D}
    {\mbox{$\Rightarrow$-i}}
  }
  {\Gamma \vdash C \Rightarrow D}
  {\mbox{\color{red} $\wedge$-e1}}
\]
reduces to
\[
  \irule{\irule{\irule{\pi_1}{\Gamma \vdash A \wedge B}{}
      &
      \irule{\pi_2}
      {\Gamma, A, C \vdash D}
      {}
    }
    {\Gamma, C \vdash D}
    {\mbox{$\wedge$-e1}}
  }
  {\Gamma \vdash C \Rightarrow D}
  {\mbox{\color{red} $\Rightarrow$-i}}
\]
leading to the rule 
$\elimand^1(t,\abstr{x} \lambda \abstr{y} u)\lra 
\lambda \abstr{y} \elimand^1(t,\abstr{x} u)$
(rule \eqref{nnnruleelimandlambda} of
\autoref{nnnfigureductionrules-two}).

The only difficult case is when the connective $c$ has two
introduction rules and the blocking rule two premises, that is when
$c$ is the disjunction and the blocking rule the elimination rule of
the disjunction, that is when the proof has the form 
  $$
    \irule{\irule{\irule{t}{\Gamma \vdash A_1 \vee A_2}{}
	& \irule{\irule{u_1}{\Gamma, A_1 \vdash B_1}{}}
	{\Gamma, A_1 \vdash B_1 \vee B_2}
	{\mbox{$\vee$-i1}}
	& \irule{\irule{u_2}{\Gamma, A_2 \vdash B_2}{}}
	{\Gamma, A_2 \vdash B_1 \vee B_2}
	{\mbox{$\vee$-i2}}
      }
      {\Gamma \vdash B_1 \vee B_2}
      {\mbox{\color{red} $\vee$-e}}
      & \irule{v_1}{\Gamma, B_1 \vdash C}{}
      & \irule{v_2}{\Gamma, B_2 \vdash C}{}
    }
    {\Gamma \vdash C}
    {\mbox{$\vee$-e}}
  $$
  that is,
\(
  \elimor(\elimorint(t,\abstr{x_1} \inl(u_1), \abstr{x_2} \inr(u_2)),
  \abstr{y_1} v_1, \abstr{y_2} v_2)
\).

\begin{figure*}[!ht]
  \centering
  \begin{align}
    \elimtop(\star, t) & \longrightarrow t \label{nnnruelimtop}\\
    (\lambda \abstr{x}t)~u & \longrightarrow  (u/x)t \label{nnnrubeta}\\
    \elimand^1(\pair{t}{u}, \abstr{x}v) & \longrightarrow  (t/x)v \label{nnnruelimand1}\\
    \elimand^2(\pair{t}{u}, \abstr{x}v) & \longrightarrow  (u/x)v \label{nnnruelimand2}\\
    \elimor(\inl(t),\abstr{x}v,\abstr{y}w) & \longrightarrow  (t/x)v 
    \label{nnnruelimorinl}\\
    \elimor(\inr(u),\abstr{x}v,\abstr{y}w) & \longrightarrow  (u/y)w
    \label{nnnruelimorinr}\\
\elimor(\inlr(t,\abstr{x_1} u_1,\abstr{x_2} u_2),
\abstr{y_1} v_1, \abstr{y_2} v_2)
&\lra \elimor(t,\abstr{x_1}(u_1/y_1)v_1,\abstr{x_2}(u_2/y_2)v_2)
\label{nnnruelimorinlr}
  \end{align}
  \caption{The reduction rules of the in-left-right-calculus (I)\label{nnnfigureductionrules-one}}
\end{figure*}

\begin{figure*}[!ht]
  \centering
  \begin{align}
    \elimbot{\top}(t) &\lra \star \label{nnnruleelimbotstar} \\
    \elimbot{A\Rightarrow B}(t) &\lra\lambda \abstr{x}\elimbot{B}(t) \label{nnnruleelimbotlambda}\\
    \elimbot{A\wedge B}(t) &\lra\pair{\elimbot{A}(t)}{\elimbot{B}(t)} \label{nnnruleelimbotpair}\\
    \elimbot{A\vee B}(t) &\lra \inl(\elimbot{A}(t)) \label{nnnruleelimbotinl}\\
    \elimbot{A\vee B}(t) &\lra \inr(\elimbot{B}(t)) \label{nnnruleelimbotinr}\\
    \elimtop(t, \star) &\lra \star\label{nnnruleelimtopstar}\\
    \elimtop(t, \lambda \abstr{x} u) &\lra \lambda \abstr{x} \elimtop(t, u)\\
    \elimtop(t, \pair{u_1}{u_2}) &\lra\pair{\elimtop(t, u_1)}{\elimtop(t,u_2)}\\
    \elimtop(t, \inl(u)) &\lra \inl(\elimtop(t,u))\\
    \elimtop(t, \inr(u)) &\lra \inr(\elimtop(t,u))\\
    \elimtop(t, \inlr(u,\abstr{x_1} v_1,
    \abstr{x_2} v_2))
    &\lra \inlr(u,\abstr{x_1} \elimtop(t,v_1),
    \abstr{x_2} \elimtop(t,v_2))
    \label{nnnruleelimtopinlr} \\
    \elimand^i(t,\abstr{x}\star) &\lra \star\label{nnnruleelimandstar}\\
    \elimand^i(t,\abstr{x} \lambda \abstr{y} u) &\lra \lambda \abstr{y} \elimand^i(t,\abstr{x} u)\label{nnnruleelimandlambda}\\
    \elimand^i(t,\abstr{x} \pair{u_1}{u_2}) &\lra\pair{\elimand^i(t,\abstr{x} u_1)}{\elimand^i(t,\abstr{x}u_2)}\\
    \elimand^i(t,\abstr{x} \inl(u)) &\lra \inl(\elimand^i(t,\abstr{x}u))\\
    \elimand^i(t,\abstr{x} \inr(u)) &\lra \inr(\elimand^i(t,\abstr{x}u))\\
    \elimand^i(t,\abstr{x} \inlr(u,\abstr{y_1} v_1,\abstr{y_2} v_2))
    &\lra
    \inlr(u, \abstr{y_1} \elimand^i(t,\abstr{x}v_1),
  \abstr{y_2} \elimand^i(t,\abstr{x}v_2))
\label{nnnruleelimandinlr}
  \end{align}
  \caption{The reduction rules of the in-left-right-calculus (II)\label{nnnfigureductionrules-two}}
\end{figure*}

\begin{figure*}[!ht]
  \centering
  \begin{align}
    \elimor(t,\abstr{x_1} \star, \abstr{x_2} \star) &\lra \star \label{nnnruleelimorstar}
    \\
    \elimor(t,\abstr{x_1} \lambda \abstr{y} u_1, \abstr{x_2} \lambda \abstr{y} u_2) &\lra \lambda \abstr{y} \elimor(t,\abstr{x_1} u_1, \abstr{x_2} u_2)
    \\
    \elimor(t,\abstr{x_1} \pair{u_1}{v_1}, \abstr{x_2} \pair{u_2}{v_2})&\label{nnnruleelimorpair} \\
    \omit\rlap{\hspace{2cm}$\lra \pair{\elimor(t,\abstr{x_1} u_1, \abstr{x_2} u_2)} {\elimor(t,\abstr{x_1} v_1, \abstr{x_2} v_2)}$}\nonumber
    \\
    \elimor(t,\abstr{x_1} \inl(u_1), \abstr{x_2} \inl(u_3)) &\lra \inl (\elimor(t, \abstr{x_1} u_1, \abstr{x_2} u_3)) \label{nnnrusuminlinl} 
    \\
    \elimor(t,\abstr{x_1} \inl(u_1), \abstr{x_2} \inr(u_4)) &\lra \inlr (t, \abstr{x_1} u_1, \abstr{x_2} u_4)\label{nnnrusuminlinr}
    \\
    \elimor(t,\abstr{x_1} \inl(u_1), \abstr{x_2} \inlr(t_2, \abstr{y_3}u_3,\abstr{y_4} u_4))\label{nnnrusuminlinlr} \\
    \omit\rlap{\hspace{2cm}$\lra \inlr(\pi_{\ref{nnnrusuminlinlr}}, \abstr{z_1}\elimor(z_1, \abstr{x_1} u_1, \abstr{w_2} (w_2/\pair{x_2}{y_3})u_3), \abstr{z_2} (z_2/\pair{x_2}{y_4})u_4)$}\nonumber
    \\
    \elimor(t,\abstr{x_1} \inr(u_2), \abstr{x_2} \inl(u_3)) &\lra \inlr (\pi_{\ref{nnnrusuminrinl}}, \abstr{x_2} u_3, \abstr{x_1} u_2)\label{nnnrusuminrinl}
    \\
    \elimor(t,\abstr{x_1} \inr(u_2), \abstr{x_2} \inr(u_4)) &\lra \inr (\elimor(t, \abstr{x_1} u_2, \abstr{x_2} u_4))\label{nnnrusuminrinr}
    \\
    \elimor(t,\abstr{x_1} \inr(u_2), \abstr{x_2} \inlr(t_2, \abstr{y_3}u_3,\abstr{y_4} u_4)) \label{nnnrusuminrinlr}\\
    \omit\rlap{\hspace{2cm}$\lra \inlr ( \pi_{\ref{nnnrusuminrinlr}}, \abstr{z_1} (z_1/\pair{x_2}{y_3})u_3, \abstr{z_2}\elimor(z_2, \abstr{x_1} u_2, \abstr{w_2} (w_2/\pair{x_2}{y_4})u_4))$}\nonumber
    \\
    \elimor(t, \abstr{x_1} \inlr(t_1, \abstr{y_1}u_1,\abstr{y_2} u_2), \abstr{x_2} \inl(u_3))  \label{nnnrusuminlrinl}\\
    \omit\rlap{\hspace{2cm}$\lra \inlr ( \pi_{\ref{nnnrusuminlrinl}}, \abstr{z_1}\elimor(z_1, \abstr{w_1} (w_1/\langle x_1, y_1\rangle)u_1, \abstr{x_2} u_3), \abstr{z_2} (z_2/\langle x_1, y_2 \rangle)u_2)$}\nonumber
    \\
    \elimor(t, \abstr{x_1} \inlr(t_1, \abstr{y_1}u_1,\abstr{y_2} u_2), \abstr{x_2} \inr(u_4)) \label{nnnrusuminlrinr}\\
    \omit\rlap{\hspace{2cm}$\lra \inlr ( \pi_{\ref{nnnrusuminlrinr}}, \abstr{z_1} (z_1/\langle x_1, y_1 \rangle)u_1, \abstr{z_2}\elimor(z_2, \abstr{w_1} (w_1/\langle x_1, y_2\rangle)u_2, \abstr{x_2} u_4))$}\nonumber
    \\
    \omit\rlap{$\elimor(t,\abstr{x_1} \inlr(t_1, \abstr{y_1}u_1,\abstr{y_2} u_2), \abstr{x_2} \inlr(t_2, \abstr{y_3}u_3,\abstr{y_4} u_4))$} \label{nnnrusuminlr}\\
    \omit\rlap{\hspace{2cm}$\lra \inlr (\pi_{\ref{nnnrusuminlr}}, \abstr{z_1}\elimor(z_1, \abstr{w_1} (w_1/\pair{x_1}{y_1})u_1, \abstr{w_2} (w_2/\pair{x_2}{y_3})u_3),$}\nonumber\\
    \omit\rlap{\hspace{4cm}$\abstr{z_2}\elimor(z_2, \abstr{w_1} (w_1/\pair{x_1}{y_2})u_2, \abstr{w_2} (w_2/\pair{x_2}{y_4})u_4))$} \nonumber
  \end{align}
  \caption{The reduction rules of the in-left-right-calculus (III)\label{nnnfigureductionrules-three}}
\end{figure*}

\subsection{The \texorpdfstring{$\inlr$}{inlr} symbol}

With a sum rule, there were many possibilities to reduce the proof
$$
  \irule{\irule{\irule{\irule{u_1}{\Gamma \vdash B_1}{}}
      {\Gamma \vdash B_1 \vee B_2}
      {\mbox{$\vee$-i1}}
      & \irule{\irule{u_2}{\Gamma \vdash B_2}{}}
      {\Gamma \vdash B_1 \vee B_2}
      {\mbox{$\vee$-i2}}
    }
    {\Gamma \vdash B_1 \vee B_2}
    {\mbox{\color{red} sum}}
    & \irule{v_1}{\Gamma, B_1 \vdash C}{}
    & \irule{v_2}{\Gamma, B_2 \vdash C}{}
  }
  {\Gamma \vdash C}
  {\mbox{$\vee$-e}}
$$
that is
\(
  \elimor(\inl(u_1) \plusr \inr(u_2) ,\abstr{y_1} v_1, \abstr{y_2} v_2)
\).

First, we could decide to drop the proof $\inr(u_2)$ and transform it
into \linebreak $\elimor(\inl(u_1),\abstr{y_1} v_1, \abstr{y_2} v_2)$
and then reduce it to $(u_1/y_1)v_1$. We could also drop the proof
$\inl(u_1)$ and transform it into $\elimor(\inr(u_2),\abstr{y_1} v_1,$
$\abstr{y_2} v_2)$ and then reduce it to $(u_2/y_2)v_2$. We could
also, as we did in \autoref{sec:inlrcalculus}, introduce a symbol
$\inlr$ and transform this proof into $\elimor(\inlr(u_1,u_2),\abstr{y_1} v_1,
\abstr{y_2} v_2)$.  We could finally decide to reduce this proof
always to $(u_1/y_1)v_1$ or $(u_2/y_2)v_2$, either to $(u_1/y_1)v_1$
or $(u_2/y_2)v_2$ in a non-deterministic way, or to $(u_1/y_1)v_1
\plus (u_2/y_2)v_2$, as we did with the rule \eqref{ruelimorinlr1}.

But, with the proof of the previous section, %of \autoref{fig:commutingcut}, that is
\(
  \elimor(\elimorint(t,\abstr{x_1} \inl(u_1), \abstr{x_2} \inr(u_2)),
  \abstr{y_1} v_1, \abstr{y_2} v_2)
\)
(where the $\plusr$ is replaced with the blocking rule
{\color{red} $\vee$-e}),
we have much fewer possibilities.

Indeed, in the first case, the proofs $\inl(u_1)$, $\inr(u_2)$, and
$\inl(u_1) \plusr \inr(u_2)$ are all proofs of $B_1 \vee B_2$ in the
context $\Gamma$.  In the second, the proofs $\inl(u_1)$, $\inr(u_2)$,
and $\elimorint(t,\abstr{x_1} \inl(u_1), \abstr{x_2} \inr(u_2))$ are
also proofs of $B_1 \vee B_2$ but in different contexts: $\inl(u_1)$
in the context $\Gamma, A_1$, $\inr(u_2)$ in the
context $\Gamma, A_2$, and $\elimorint(t,\abstr{x_1} \inl(u_1),
\abstr{x_2} \inr(u_2))$ in the context $\Gamma$.

Thus, if we want
to transform the proof $\elimorint(t,\abstr{x_1}$
$\inl(u_1),\abstr{x_2}$ $\inr(u_2))$ with a new introduction 
rule
$\inlr$ we must keep record of the proof $t$, and we must bind the
variables $x_1$ in $u_1$ and $x_2$ in $u_2$, leading to the admissible
rule
\[
  \irule{\Gamma \vdash t:A_1 \vee A_2 & \Gamma, x_1:A_1 \vdash u_1:B_1 &
  \Gamma, x_2:A_2 \vdash u_2:B_2}
  {\Gamma \vdash \inlr(t,\abstr{x_1} u_1,\abstr{x_2} u_2):B_1 \vee B_2}
  {}
\]

\subsection{Reducing ordinary cuts}

Adding the new introduction rule $\inlr$ produces new ordinary cuts
for the disjunction: $\elimor(\inlr(t,\abstr{x_1} u_1,\abstr{x_2}
u_2), \abstr{y_1} v_1, \abstr{y_2} v_2)$.  To reduce this proof, the
proof $(u_1/y_1)v_1$ is indeed a proof of $C$, but in the context
$\Gamma, x_1:A_1$, the proof $(u_2/y_2)v_2$ is indeed a proof of $C$,
but in the context $\Gamma, x_2:A_2$. Thus, the only option is to
reduce this proof to $\elimor(t,\abstr{x_1} (u_1/y_1)v_1,
\abstr{x_2}(u_2/y_2)v_2)$
(rule \eqref{nnnruelimorinlr} of \autoref{nnnfigureductionrules-one}), analogous to the rule
\eqref{ruelimorinlr1}. More generally, rules 
  \eqref{ruelimtop} to \eqref{ruelimorinlr1}, reducing ordinary cuts, yields
  the rules of \autoref{nnnfigureductionrules-one}.

  \subsection{Commutation rules}

The rules~\eqref{nnnruleelimbotstar} to~\eqref{nnnruleelimandinlr} are
the easy commutation rules, where the blocking rule has zero or one
minor premise: the rules~\eqref{nnnruleelimbotstar}
to~\eqref{nnnruleelimbotinr} refine the rules~\eqref{rusumstar}
to~\eqref{rusuminlr} when the blocking rule is the rule $\bot$-e.  In
particular, rules~\eqref{nnnruleelimbotinl}
and~\eqref{nnnruleelimbotinr} introduce non-determinism to avoid
breaking the symmetry.  The rules~\eqref{nnnruleelimtopstar}
to~\eqref{nnnruleelimtopinlr} refine the same rules when the blocking
rule is the rule $\top$-e.  The rules~\eqref{nnnruleelimandstar}
to~\eqref{nnnruleelimandinlr} refine these rules when the
blocking rule is the rule $\wedge$-e.

The rules~\eqref{nnnruleelimorstar} to~\eqref{nnnruleelimorpair} are
the easy commutation rules where the blocking rule is the rule
$\vee$-e, but the connective $c$ has only one introduction rule: they
refine the rules~\eqref{rusumstar} to~\eqref{rusumpair}.

Let us now turn to the only difficult case, where the blocking rule is
the elimination of the disjunction and the blocked connective is the
disjunction. We must refine the rules \eqref{rusuminl} to
\eqref{rusuminlr}, for example, we must refine the rule
\eqref{rusuminlr}:
\[
  \inlr(u_1,u_2) \plus \inlr(u_3,u_4)
  \longrightarrow \inlr(u_1 \plus u_3,u_2 \plus u_4)
\]
into a rule reducing the proof
\[
  \elimor(t,\abstr{x_1} \inlr(t_1, \abstr{y_1}u_1,\abstr{y_2} u_2),
  \abstr{x_2} \inlr(t_2, \abstr{y_3}u_3,\abstr{y_4} u_4)) 
\]
to a proof of the form 
\[
  \inlr
  (..., 
    \abstr{z_1}\elimor(..., \abstr{w_1} u_1, \abstr{w_2} u_3),
  \abstr{z_2}\elimor(..., \abstr{w_3} u_2, \abstr{w_3} u_4)),
\]
where
\begin{align*}
  & \Gamma \vdash t:A_1 \vee A_2\\
  & \Gamma, x_1:A_1 \vdash t_1:B_1 \vee B_2\\
  & \Gamma, x_2:A_2 \vdash t_2:B_3 \vee B_4\\
  & \Gamma, x_1:A_1, y_1:B_1 \vdash u_1:C\\
  & \Gamma, x_1:A_1, y_2:B_2 \vdash u_2:D\\
  & \Gamma, x_2:A_2, y_3:B_3 \vdash u_3:C\\
  & \Gamma, x_2:A_2, y_4:B_4 \vdash u_4:D
\end{align*}

Note that the proofs $u_1$ and $u_3$ are in different contexts:
besides those of $\Gamma$, the free variables of $u_1$ are $x_1$ and
$y_1$ while those of $u_3$ are $x_2$ a $y_3$. So there is no way to
chose $z_1$, $w_1$, and $w_2$ among $x_1$, $x_2$, $y_1$, and $y_3$ to
obtain a well-scoped proof. Instead, we take $w_1$ to be a proof of
$A_1 \wedge B_1$ so that, in $u_1$, the first and the second component
of $w_1$ may be substituted for the variables $x_1$ and $y_1$. In the
same way, $w_2$ is taken to be a proof of $A_2 \wedge B_3$.

Thus, the first argument of the first $\elimor$ symbol must be a proof
of $(A_1 \wedge B_1) \vee (A_2 \wedge B_3)$ and the first argument of
the second $\elimor$ symbol must be a proof of $(A_1 \wedge B_2) \vee
(A_2 \wedge B_4)$. A possibility is to assign these types to the
variables $z_1$ and $z_2$ and reduce the proof to a proof of the form
$
\inlr
(
  ...,
  \abstr{z_1}\elimor(z_1,
    \abstr{w_1} (w_1/\pair{x_1}{y_1})u_1, 
  \abstr{w_2} (w_2/\pair{x_2}{y_3})u_3),
  \abstr{z_2}\elimor(z_2,
    \abstr{w_1} (w_1/\pair{x_1}{y_2})u_2,
  \abstr{w_2}$ $(w_2/\pair{x_2}{y_4})u_4)
)$
where $(u/\pair{x}{y})t$ is a notation for
$(\elimand^1(u,\abstr{z}z)/x,
  \elimand^2(u,\abstr{z}z)/y)t$.

Finally, the first argument of the $\inlr$ symbol must be a proof of
$((A_1 \wedge B_1) \vee (A_2 \wedge B_3)) \vee ((A_1 \wedge B_2) \vee
(A_2 \wedge B_4))$. Note that as we have proofs $t$, $t_1$, and $t_2$,
either $B_1$, $B_2$, $B_3$, or $B_4$ is true, where $A_1$ is true in
the two first cases and $A_2$ is true in the last two.  Thus, we can
build a proof $\pi$ of this proposition.
This leads to the rule
$$\begin{array}{r@{\,}l}
  &\elimor(t,\abstr{x_1} \inlr(t_1, \abstr{y_1}u_1,\abstr{y_2} u_2), \abstr{x_2} \inlr(t_2, \abstr{y_3}u_3,\abstr{y_4} u_4)) \\
  &\lra
  \inlr
  \begin{aligned}[t]
    (
      \pi,
      &
      \abstr{z_1}\elimor(z_1,
	\abstr{w_1} (w_1/\pair{x_1}{y_1})u_1, 
      \abstr{w_2} (w_2/\pair{x_2}{y_3})u_3),\\
      &
      \abstr{z_2}\elimor(z_2,
	\abstr{w_1} (w_1/\pair{x_1}{y_2})u_2, 
      \abstr{w_2} (w_2/\pair{x_2}{y_4})u_4)
    )
  \end{aligned}
\end{array}$$
where $\pi$ is the proof of $((A_1 \wedge B_1) \vee (A_2 \wedge B_3))
\vee ((A_1 \wedge B_2) \vee (A_2 \wedge B_4))$
$$\begin{array}{r@{\,}l}
  \elimor(t,
&[x_1]\elimor(t_1,[y_1]\inl(\inl(\pair{x_1}{y_1})), [y_2]\inr(\inl(\pair{x_1}{y_2}))),\\
&[x_2]\elimor(t_2,[y_3]\inl(\inr(\pair{x_2}{y_3})), [y_4]\inr(\inr(\pair{x_2}{y_4}))))
\end{array}$$

This way the commuting cut 
  $\elimor(
      \elimorint(t,\abstr{x_1} \inlr(t_1, \abstr{y_1}u_1,\abstr{y_2} u_2), \abstr{x_2} \inlr(t_2, \abstr{y_3}u_3,\abstr{y_4} u_4)),$
    $\abstr{\alpha_1}s_1, \abstr{\alpha_2}s_2)$
reduces to
$\begin{aligned}[t]
  \elimor(
    \begin{aligned}[t]
      &\inlr
      \begin{aligned}[t]
	(
	  \pi,
	  &
	  \abstr{z_1}\elimor(z_1,
	    \abstr{w_1} (w_1/\pair{x_1}{y_1})u_1, 
	  \abstr{w_2} (w_2/\pair{x_2}{y_3})u_3),\\
	  &
	  \abstr{z_2}\elimor(z_2,
	    \abstr{w_1} (w_1/\pair{x_1}{y_2})u_2, 
	  \abstr{w_2} (w_2/\pair{x_2}{y_4})u_4)
	),\\
      &\abstr{\alpha_1}s_1,\abstr{\alpha_2}s_2)
      \end{aligned}
  \end{aligned}
\end{aligned}$

\noindent
where $\pi$ is as above, and then, by the rule explained in the previous section, to 
$$\begin{array}{r@{\,}l}
\elimor(\pi,
&\abstr{z_1}
(\elimor(z_1,
	  \abstr{w_1} (w_1/\pair{x_1}{y_1})u_1, 
	\abstr{w_2} (w_2/\pair{x_2}{y_3})u_3)/\alpha_1)s_1,\\
&\abstr{z_2}(\elimor(z_2,
	  \abstr{w_1} (w_1/\pair{x_1}{y_2})u_2, 
	\abstr{w_2} (w_2/\pair{x_2}{y_4})u_4)/\alpha_2)s_2)
\end{array}$$
 leading to the rules \eqref{nnnrusuminlinl}
to \eqref{nnnrusuminlr}.  where the proof $\pi_{\ref{nnnrusuminlr}}$
is the proof $\pi$ above and the proofs $\pi_{\ref{nnnrusuminlinlr}}$
to $\pi_{\ref{nnnrusuminlrinr}}$, given in \autoref{nnnfigurepi},
are analogous.  These rules form a rewrite system eliminating both
ordinary and commuting cuts.

We leave the termination of this system as an open problem.

\begin{figure*}[!ht]
  \centering
  \columnwidth=\linewidth
  $$
    \begin{array}{l}
      \vdash\pi_{\ref{nnnrusuminlinlr}}:
      (A_1 \vee (A_2 \wedge B_3)) \vee (A_2 \wedge B_4)\\
      \pi_{\ref{nnnrusuminlinlr}} = 
	\elimor(t,
	  \abstr{x_1}\inl(\inl(x_1)),
	\abstr{x_2}\elimor(t_2,\abstr{y_3}\inl(\inr(\pair{x_2}{y_3})), \abstr{y_4}\inr(\pair{x_2}{y_4})))
	\\[1.5ex]
      \vdash\pi_{\ref{nnnrusuminrinl}}:A_2 \vee A_1\\
      \pi_{\ref{nnnrusuminrinl}}=
	\elimor(t,\abstr{x_1} \inr(x_1),
	\abstr{x_2} \inl(x_2))
	\\[1.5ex]
      \vdash\pi_{\ref{nnnrusuminrinlr}}:(A_2 \wedge B_3) \vee (A_1 \vee (A_2 \wedge B_4))\\
      \pi_{\ref{nnnrusuminrinlr}}= 
	\elimor(t,
	  \abstr{x_1}\inr(\inl(x_1)),
	\abstr{x_2}\elimor(t_2,\abstr{y_3}\inl(\pair{x_2}{y_3}), \abstr{y_4}\inr(\inr(\pair{x_2}{y_4}))))
      \\[1.5ex]
      \vdash\pi_{\ref{nnnrusuminlrinl}}:((A_1 \wedge B_1)\vee A_2) \vee (A_1 \wedge B_2)\\
      \pi_{\ref{nnnrusuminlrinl}}=
	\elimor(t,
	  \abstr{x_1}\elimor(t_1,\abstr{y_1}\inl(\inl(\pair{x_1}{y_1})), \abstr{y_2}\inr(\pair{x_1}{y_2})),
	\abstr{x_2}\inl(\inr(x_2)))
      \\[1.5ex]
      \vdash\pi_{\ref{nnnrusuminlrinr}}:(A_1 \wedge B_1) \vee ((A_1 \wedge B_2) \vee A_2)\\
      \pi_{\ref{nnnrusuminlrinr}}= 
	\elimor(t,
	  \abstr{x_1}\elimor(t_1,\abstr{y_1}\inl(\pair{x_1}{y_1}),\abstr{y_2}\inr(\inl(\pair{x_1}{y_2}))),
	\abstr{x_2}\inr(\inr(x_2)))
      \\[1.5ex]
      \vdash\pi_{\ref{nnnrusuminlr}}:((A_1 \wedge B_1) \vee (A_2 \wedge B_3)) \vee ((A_1 \wedge B_2) \vee (A_2 \wedge B_4))\\
      \pi_{\ref{nnnrusuminlr}}= 
      \begin{aligned}[t]
	\elimor(t,&
	  [x_1]\elimor(t_1,[y_1]\inl(\inl(\pair{x_1}{y_1})), [y_2]\inr(\inl(\pair{x_1}{y_2}))),
	  \\
	  &[x_2]\elimor(t_2,[y_3]\inl(\inr(\pair{x_2}{y_3}), [y_4]\inr(\inr(\pair{x_2}{y_4}))))
	\end{aligned}
      \end{array}
    $$
    \caption{Definition of the $\pi$ terms of \autoref{nnnfigureductionrules-three}\label{nnnfigurepi}}
\end{figure*}

\subsection{An example of application to program optimization}
\label{sec:optimization}

\newcommand\tletp[4]{\mathsf{let}\ {#2}=\pi_{#1}{#3}\ \mathsf{in}\ {#4}}
\newcommand\tmatch[5]{\mathsf{match}\ {#1}\ \mathsf{in}\ \{{#2}\mapsto{#3},{#4}\mapsto{#5}\}}
To discuss programs, we can give a more familiar name to the elimination proof of the conjunction. In this section, we write $\tletp 1 x t u$ instead of $\elimand^1(t,\abstr{x}u)$, and we assume we have natural numbers and arithmetic operations.

We can see that the commuting cuts allow optimizing programs. For example, the following program, where $x$ is a parameter yet to be defined
\(
  (\tletp 1 yx{\lambda\abstr{z}(z+y)})~u
\)
is commuted, by rule \eqref{nnnruleelimandlambda}, into
\(
  \lambda\abstr{z}(\tletp 1 yx{(z+y)})~u
\)
and thus reduced, by rule \eqref{nnnrubeta}, to
\(
  \tletp 1 yx{(u+y)}
\).
Yet, the same result would be obtained if we had chosen to commute
with the elimination rule, yielding first
\(
  \tletp 1 yx{((\lambda\abstr{z}(z+y))~u)}
\)
and finally
\(
  \tletp 1 yx{(u+y)}
\).

However, the main difference is that even the program
\(
  \tletp 1 yx{\lambda\abstr{z}(z+y)}
\)
before it is applied to $u$, can be commuted in our case into
\(
  \lambda\abstr{z}(\tletp 1 yx{(z+y)})
\)
while it is irreducible if the commutation is made with the
elimination rule, since we need first to have an application
(elimination) in order to commute.

Now suppose this program is used in a bigger one, such as
$
  (\lambda\abstr{f}(f~0+f~1+\cdots+f~999))~{(\tletp 1 yx{(\lambda\abstr{z}(z+y)}))}
  $.
In our case, this term can be reduced first by rule \eqref{nnnruleelimandlambda} and then 
by rule \eqref{nnnrubeta} 
and finally, by doing $1000$ reductions with rule \eqref{nnnrubeta}, to 
$
  {\tletp 1 yx{(0+y)}}
  +
  {\tletp 1 yx{(1+y)}}
  +\cdots+
  {\tletp 1 yx{(999+y)}}
$
while commuting with the elimination, we can only reduce by rule \eqref{nnnrubeta} once to get
$
  (\tletp 1 yx{\lambda\abstr{z}(z+y)})~0
  +
  (\tletp 1 yx{\lambda\abstr{z}(z+y)})~1
  +\cdots+(\tletp 1 yx{\lambda\abstr{z}(z+y)})~999$
and then it would require $1000$ commutations, before the $1000$ reductions by rule \eqref{nnnrubeta} to get the same result.

\section{Conclusion}

In \cite{DiazcaroDowekTCS23}, a non-harmonious connective $\odot$ was
introduced. This connective was characterized by the absence of
symmetry between its introduction and elimination rules. Specifically,
this connective had the introduction rule of the conjunction and the
elimination rule of the disjunction.  In contrast, in this paper, we
have merely added the introduction rule of the conjunction as an extra
rule of the disjunction, the rule $\vee$-i3, and obtained a calculus
with similar properties.  As this rule is admissible, we do not obtain
a new connective, but it remains the disjunction. The cuts formed with
this new introduction rule can be reduced (rule \ref{ruelimorinlr1}).
The termination proof of the calculus with the rule $\vee$-i3 is also
simpler than that of the $\odot$-calculus: as we commute the sum rule
with the introduction rules of the disjunction rather that with its
elimination rules, we do not need to add the so-called \emph{ultra-reduction} rules. We also get a stronger introduction
property: closed irreducible proofs end with an introduction rule and
are not sums, or linear combinations, of proofs ending with an
introduction rule.

This calculus, in its linear version, has applications to quantum
computing, where this new introduction rule $\vee$-i3 can be seen as a
way to commute the sum rule with the introduction rule of the
disjunction in the proof
\[
  \irule{\irule{\irule{\pi_1}
      {\Gamma \vdash A}
      {}
    }
    {\Gamma \vdash A \oplus B}
    {\mbox{$\oplus$-i1}}
    &\irule{\irule{\pi_2}
      {\Gamma \vdash B}
      {}
    }
    {\Gamma \vdash A \oplus B}
    {\mbox{$\oplus$-i2}}
  }
  {\Gamma \vdash A \oplus B}
  {\mbox{sum}}
\]

But, even in calculi without interstitial rules, the rule $\vee$-i3 is
useful to reduce commuting cuts, although, in this paper, we leave the
termination of such reduction as an open problem. This reduction of
commuting cuts itself has applications to program optimization.

\bibliographystyle{alphaurl}
\bibliography{inlr}

\appendix

\section{Omitted proofs in \autoref{subsec:properties}}\label{app:properties}
\subsection{Proof of \autoref{thm:SRlin}}
\begin{prop}
  [Substitution]\label{lem:substitutionLin}
  Let $t$ be a proof of $A$ in context $\Gamma,x:B$ and $u$ be a proof of $B$ in context $\Delta$. Then, $(u/x)t$ is a proof of $A$ in context $\Gamma,\Delta$.
\end{prop}
\begin{proof}
  By induction on the structure of $t$.
  \begin{itemize}
      \item Let $t=x$, thus $B=A$, $\Gamma$ is empty, and $(u/x)t = u$, and since $u$ is a proof of $B$ in context $\Delta$, we have that $(u/x)t$ is a proof of $A$ in context $\Delta$.

      \item Let $t=y\neq x$, this case is impossible since the context $\Gamma,x:B$ cannot be equal to $y:A$.

     \item Let $t=t_1\plus t_2$. By induction hypothesis, $(u/x)t_1$ and $(u/x)t_2$ are proofs of $A$ in context $\Gamma,\Delta$. Thus, $(u/x)t$ is a proof of $A$ in context $\Gamma,\Delta$.
     
     \item Let $t=a\bullet t_1$. By induction hypothesis, $(u/x)t_1$ is a proof of $A$ in context $\Gamma,\Delta$. Thus, $(u/x)t$ is a proof of $A$ in context $\Gamma,\Delta$.
     
     \item Let $t=a.\star$, this case is impossible since the context $\Gamma,x:B$ is not empty.
     
     \item Let $t=\elimone(t_1,t_2)$, then $\Gamma,x:B=\Gamma_1,\Gamma_2$, $t_1$ is a proof of $\one$ in the context $\Gamma_1$, and $t_2$ is a proof of $A$ in the context $\Gamma_2$. Let $x:B\in\Gamma_1$, then by the induction hypothesis, $(u/x)t_1$ is a proof of $\one$ in the context $\Gamma_1,\Delta$. Thus, $(u/x)t$ is a proof of $A$ in the context $\Gamma_1,\Delta,\Gamma_2=\Gamma,\Delta$.
     Let $x:B\in\Gamma_2$, then by the induction hypothesis, $(u/x)t_2$ is a proof of $A$ in the context $\Gamma_2,\Delta$. Thus, $(u/x)t$ is a proof of $A$ in the context $\Gamma_1,\Gamma_2,\Delta=\Gamma,\Delta$.

     \item Let $t=\lambda\abstr{y}t_1$, then $A=A_1\multimap A_2$ and $t_1$ is a proof of $A_2$ in the context $\Gamma,y:A_1,x:B$. Then, by the induction hypothesis, $(u/x)t_1$ is a proof of $A_2$ in the context $\Gamma,y:A_1,\Delta$. Thus, $(u/x)t$ is a proof of $A_1\multimap A_2$ in the context $\Gamma,\Delta$. 

     \item Let $t=t_1\,t_2$, then $\Gamma,x:B = \Gamma_1,\Gamma_2$, $t_1$ is a proof of $A_1\multimap A$ in the context $\Gamma_1$, and $t_2$ is a proof of $A_1$ in the context $\Gamma_2$. Let $x:B\in\Gamma_1$, then by the induction hypothesis, $(u/x)t_1$ is a proof of $A_1\multimap A$ in the context $\Gamma_1,\Delta$. Thus, $(u/x)t$ is a proof of $A$ in the context $\Gamma_1,\Delta,\Gamma_2=\Gamma,\Delta$.
     Let $x:B\in\Gamma_2$, then by the induction hypothesis, $(u/x)t_2$ is a proof of $A_1$ in the context $\Gamma_2,\Delta$. Thus, $(u/x)t$ is a proof of $A$ in the context $\Gamma_1,\Gamma_2,\Delta=\Gamma,\Delta$.

     \item Let $t=\inl(t_1)$, then $A=A_1\oplus A_2$ and $t_1$ is a proof of $A_1$ in the context $\Gamma,x:B$. By the induction hypothesis, $(u/x)t_1$ is a proof of $A_1$ in the context $\Gamma,\Delta$. Thus, $(u/x)t$ is a proof of $A_1\oplus A_2$ in the context $\Gamma,\Delta$.
     
     \item Let $t=\inr(t_1)$, this case is analogous to the previous one.
     
     \item Let $t=\inlr(t_1,t_2)$, then $A=A_1\oplus A_2$, $t_1$ is a proof of $A_1$ in the context $\Gamma,x:B$, and $t_2$ is a proof of $A_2$ in the context $\Gamma,x:B$. By the induction hypothesis, $(u/x)t_1$ is a proof of $A_1$ in the context $\Gamma,\Delta$, and $(u/x)t_2$ is a proof of $A_2$ in the context $\Gamma,\Delta$. Thus, $(u/x)t$ is a proof of $A_1\oplus A_2$ in the context $\Gamma,\Delta$.

     \item Let $t=\elimplus(t_1,\abstr{y}v,\abstr{z}w)$, then $\Gamma,x:B=\Gamma_1,\Gamma_2$, $t_1$ is a proof of $C_1\oplus C_2$ in the context $\Gamma_1$, $v$ is a proof of $A$ in the context $\Gamma_2,y:C_1$, and $w$ is a proof of $A$ in the context $\Gamma_2,z:C_2$. Let $x:B\in\Gamma_1$, then by the induction hypothesis, $(u/x)t_1$ is a proof of $C_1\oplus C_2$ in the context $\Gamma_1,\Delta$. Thus, $(u/x)t$ is a proof of $A$ in the context $\Gamma_1,\Delta,\Gamma_2=\Gamma,\Delta$.
     Let $x:B\in\Gamma_2$, then by the induction hypothesis, $(u/x)v$ is a proof of $A$ in the context $\Gamma_2,\Delta$ and $(u/x)w$ is a proof of $A$ in the context $\Gamma_2,\Delta$. Thus, $(u/x)t$ is a proof of $A$ in the context $\Gamma_1,\Gamma_2,\Delta=\Gamma,\Delta$.

     \item Let $t=\elimplus^{nd}(t_1,\abstr{y}v,\abstr{z}w)$, this case is analogous to the previous one.
     \qedhere
  \end{itemize}
\end{proof}

\begin{proof}[Proof of \autoref{thm:SRlin}]
  By induction on the reduction relation. The inductive cases are trivial, so we
  only treat the basic cases corresponding to rules \eqref{ll:ruelimone} to
  \eqref{ll:rubulletinlr} of \autoref{linearreductionrules}.
  \begin{enumerate}
    \item Let $t=\elimone(a.\star,v)$ and $u=v$. Then, by inversion, $u$ is a
    proof of $A$ in context $\Gamma$.

    \item Let $t=(\lambda \abstr{x}v)~w$ and $(v/x)w$. Then
    $\Gamma=\Gamma_1,\Gamma_2$, $v$ is a proof of $A$ in the context
    $\Gamma_1,x:B$ and $w$ is a proof of $B$ in the context $\Gamma_2$. Thus, by
    \autoref{lem:substitutionLin}, $(v/x)w$ is a proof of $A$ in the
    context $\Gamma_1,\Gamma_2=\Gamma$.

		\item\label{caseeliminl} Let $t=\elimplus(\inl(v),\abstr{x}w_1,\abstr{y}w_2)$ and $u=(v/x)w_1$.
		Then $\Gamma=\Gamma_1,\Gamma_2$, $v$ is a proof of $B$ in the context
		$\Gamma_1$, and $w_1$ is a proof of $A$ in the context $\Gamma_2,x:B$. Thus,
		by \autoref{lem:substitutionLin}, $(v/x)w_1$ is a proof of $A$ in
		the context $\Gamma_1,\Gamma_2=\Gamma$.

		\item\label{caseeliminr} Let $t=\elimplus(\inr(v),\abstr{x}w_1,\abstr{y}w_2)$ and $u=(v/y)w_2$.
		This case is analogous to the previous one.

		\item Let $t=\elimplus(\inlr(v_1,v_2),\abstr{x}w_1,\abstr{y}w_2)$ and $u=(v_1/x)w_1 \plus (v_2/y)w_2$.
     Then $\Gamma=\Gamma_1,\Gamma_2$, $v_1$ is a proof of $B_1$ in the context $\Gamma_1$, $v_2$ is a proof of $B_2$ in the context $\Gamma_1$, $w_1$ is a proof of $A$ in the context $\Gamma_2,x:B_1$, and $w_2$ is a proof of $A$ in the context $\Gamma_2,y:B_2$. Thus, by \autoref{lem:substitutionLin}, $(v_1/x)w_1$ is a proof of $A$ in the context $\Gamma_1,\Gamma_2=\Gamma$ and $(v_2/y)w_2$ is a proof of $A$ in the context $\Gamma_1,\Gamma_2=\Gamma$. Therefore, $(v_1/x)w_1 \plus (v_2/y)w_2$ is a proof of $A$ in the context $\Gamma$.

		\item Let $t=\elimplus^{nd}(\inl(v),\abstr{x}w_1,\abstr{y}w_2)$ and $u=(v/x)w_1$. This case is analogous to the \autoref{caseeliminl}.

		\item Let $t=\elimplus^{nd}(\inr(v),\abstr{x}w_1,\abstr{y}w_2)$ and $u=(v/y)w_2$. This case is analogous to the \autoref{caseeliminr}.

		\item Let $t=\elimplus^{nd}(\inlr(v_1,v_2),\abstr{x}w_1,\abstr{y}w_2)$ and $u=(v_1/x)w_1$.
     Then $\Gamma=\Gamma_1,\Gamma_2$, $v_1$ is a proof of $B_1$ in the context $\Gamma_1$, and $w_1$ is a proof of $A$ in the context $\Gamma_2,x:B_1$. Thus, by \autoref{lem:substitutionLin}, $(v_1/x)w_1$ is a proof of $A$ in the context $\Gamma_1,\Gamma_2=\Gamma$.

		\item Let $t=\elimplus^{nd}(\inlr(v_1,v_2),\abstr{x}w_1,\abstr{y}w_2)$ and $u=(v_2/y)w_2$. This case is analogous to the previous one.

		\item Let $t={a.\star} \plus b.\star$ and $u=(a+b).\star$. Then $\Gamma$ is empty and $A=\one$, and we have that $(a+b).\star$ is a proof of $\one$ in empty context.

		\item Let $t=(\lambda \abstr{x}v) \plus (\lambda \abstr{x}w)$ and $u=\lambda \abstr{x}(v \plus w)$. Then $A=A_1\multimap A_2$, $v$ is a proof of $A_2$ in the context $\Gamma,x:A_1$, and $w$ is a proof of $A_2$ in the context $\Gamma,x:A_1$. Hence, $(v \plus w)$ is a proof of $A_2$ in the context $\Gamma,x:A_1$, and thus, $\lambda \abstr{x}(v \plus w)$ is a proof of $A_1\multimap A_2$ in the context $\Gamma$.

		\item\label{caseinlinl} Let $t=\inl(v_1) \plus \inl(v_2)$ and $u=\inl(v_1 \plus v_2)$. Then $A=B\oplus C$, $v_1$ is a proof of $B$ in the context $\Gamma$, and $v_2$ is a proof of $B$ in the context $\Gamma$. Hence, $v_1 \plus v_2$ is a proof of $B$ in the context $\Gamma$, and thus, $\inl(v_1 \plus v_2)$ is a proof of $B\oplus C$ in the context $\Gamma$.

		\item Let $t=\inl(v) \plus \inr(w)$ and $u=\inlr(v,w)$. Then $A=B\oplus C$, $v$ is a proof of $B$ in the context $\Gamma$, and $w$ is a proof of $C$ in the context $\Gamma$. Hence, $\inlr(v,w)$ is a proof of $B\oplus C$ in the context $\Gamma$.

		\item\label{caseinlinlr} Let $t=\inl(v_1) \plus \inlr(v_2,w)$ and $u=\inlr(v_1 \plus v_2,w)$. Then $A=B\oplus C$, $v_1$ is a proof of $B$ in the context $\Gamma$, $v_2$ is a proof of $B$ in the context $\Gamma$, and $w$ is a proof of $C$ in the context $\Gamma$. Hence, $v_1 \plus v_2$ is a proof of $B$ in the context $\Gamma$, and thus, $\inlr(v_1 \plus v_2,w)$ is a proof of $B\oplus C$ in the context $\Gamma$.

		\item Let $t=\inr(w) \plus \inl(v)$ and $u=\inlr(v,w)$. Then $A=B\oplus C$, $v$ is a proof of $B$ in the context $\Gamma$, and $w$ is a proof of $C$ in the context $\Gamma$. Hence, $\inlr(v,w)$ is a proof of $B\oplus C$ in the context $\Gamma$.

		\item Let $t=\inr(w_1) \plus \inr(w_2)$ and $u=\inr(w_1 \plus w_2)$. This case is analogous to the \autoref{caseinlinl}.

		\item Let $t=\inr(w_1) \plus \inlr(v,w_2)$ and $u=\inlr(v,w_1 \plus w_2)$. This case is analogous to the \autoref{caseinlinlr}.

		\item Let $t=\inlr(v_1,w) \plus \inl(v_2)$ and $u=\inlr(v_1\plus v_2,w)$. Thi case is analogous to the \autoref{caseinlinlr}.

		\item Let $t=\inlr(v,w_1) \plus \inr(w_2)$ and $u=\inlr(v,w_1 \plus w_2)$. This case is analogous to the \autoref{caseinlinlr}.

		\item Let $t=\inlr(v_1,w_1) \plus \inlr(v_2,w_2)$ and $u=\inlr(v_1 \plus v_2,w_1 \plus w_2)$. Then $A=B\oplus C$, $v_1$ is a proof of $B$ in the context $\Gamma$, $v_2$ is a proof of $B$ in the context $\Gamma$, $w_1$ is a proof of $C$ in the context $\Gamma$, and $w_2$ is a proof of $C$ in the context $\Gamma$. Hence, $v_1 \plus v_2$ is a proof of $B$ in the context $\Gamma$, $w_1 \plus w_2$ is a proof of $C$ in the context $\Gamma$, and thus, $\inlr(v_1 \plus v_2,w_1 \plus w_2)$ is a proof of $B\oplus C$ in the context $\Gamma$.

		\item Let $t=a \bullet b.\star$ and $u=(a \times b).\star$. Then $\Gamma$ is empty and $A=\one$, and we have that $(a \times b).\star$ is a proof of $\one$ in empty context.

		\item Let $t=a \bullet \lambda \abstr{x} v$ and $u=\lambda \abstr{x} a \bullet v$. Then $A=B\multimap C$, $v$ is a proof of $C$ in the context $\Gamma,x:B$. Hence, $a \bullet v$ is a proof of $C$ in the context $\Gamma,x:B$, and thus, $\lambda \abstr{x} a \bullet v$ is a proof of $B\multimap C$ in the context $\Gamma$.

		\item Let $t=a \bullet \inl(v)$ and $u=\inl(a \bullet v)$. Then $A=B\oplus C$, $v$ is a proof of $B$ in the context $\Gamma$. Hence, $a \bullet v$ is a proof of $B$ in the context $\Gamma$, and thus, $\inl(a \bullet v)$ is a proof of $B\oplus C$ in the context $\Gamma$.

		\item Let $t=a \bullet \inr(v)$ and $u=\inr(a \bullet v)$. This case is analogous to the previous one.

		\item Let $t=a \bullet \inlr(v,w) $ and $u=\inlr(a \bullet v,a \bullet w)$. Then $A=B\oplus C$, $v$ is a proof of $B$ in the context $\Gamma$, and $w$ is a proof of $C$ in the context $\Gamma$. Hence, $a \bullet v$ is a proof of $B$ in the context $\Gamma$, $a \bullet w$ is a proof of $C$ in the context $\Gamma$, and thus, $\inlr(a \bullet v,a \bullet w)$ is a proof of $B\oplus C$ in the context $\Gamma$.
    \qedhere
  \end{enumerate}
\end{proof}

\subsection{Proof of \autoref{introductionslinear}}

\begin{proof}[Proof of \autoref{introductionslinear}]
Let $t$ be a closed irreducible proof of some proposition $A$. We prove,
by induction on the structure of $t$ that $t$ is an introduction.

As the proof $t$ is closed, it is not a variable.

It cannot be a sum $u \plus v$, as if it were $u$ and $v$ would be
closed irreducible proofs of the same proposition, hence, by induction
hypothesis, they would either be both introductions of $\one$, both
introductions of $\multimap$, or both introductions of $\oplus$, and
the proof $t$ would be reducible.

It cannot be an elimination as if it were of the form $\elimone(u,v)$,
$u~v$, or
$\elimplus(u,\abstr{x}v,\abstr{y}w)$, then $u$ would a closed
irreducible proof, hence, by induction hypothesis, it would an
introduction and the proof $t$ would be reducible.  

Hence, it is an introduction. \qedhere
\end{proof}

\section{Omitted proofs in \autoref{seclinearity}}\label{app:linearity}
\subsection{Proof of \autoref{thm:converse}}
To prove this theorem, we start by proving some algebraic properties of proof-terms, in the following lemma.

\begin{lem} \label{vecstructure}
  If $A \in {\mathcal V}$ and $t$, $t_1$, $t_2$, and $t_3$ are closed proofs of
  $A$, then
    \begin{enumerate}
      \item $(t_1 \plus t_2) \plus t_3 \equiv t_1 \plus (t_2 \plus t_3)$
      \item $t_1 \plus t_2 \equiv t_2 \plus t_1$
      \item $a \bullet b \bullet t \equiv (a \times b) \bullet t$
      \item $a \bullet (t_1 \plus t_2) \equiv a \bullet t_1 \plus a \bullet t_2$
      \item $(a + b) \bullet t \equiv a \bullet t \plus b \bullet t$
    \end{enumerate}
\end{lem}
\begin{proof}
  ~
  \begin{enumerate}
    \item By induction on $A$. If $A = \one$, then $t_1$, $t_2$, and $t_3$
      reduce respectively to $a.\star$, $b.\star$, and $c.\star$. We have
      $$
	(t_1 \plus t_2) \plus t_3 \lras ((a + b) + c).\star = (a + (b +
	c)).\star \llas t_1 \plus (t_2 \plus t_3)
      $$
      If $A = A_1 \oplus  A_2$, then 
      $t_1$, $t_2$, and $t_3$ reduce either to a term of the form $\inl$, $\inr$, or $\inlr$, leading to twenty-seven cases.
      As all the cases are similar we only give two examples.
      \begin{itemize}
	\item 
	  If $t_1 \lra^* \inlr(u_1,v_1)$, $t_2 \lra^* \inlr(u_2,v_2)$, and
	  $t_3 \lra^* \inlr(u_3,v_3)$, using the induction hypothesis, we
	  have
	  \begin{align*}
	    (t_1 \plus t_2) \plus t_3 \lras &\ 
	    \inlr((u_1 \plus u_2) \plus u_3,(v_1 \plus v_2) \plus v_3)\\
	    \equiv &\ 
	    \inlr(u_1 \plus (u_2 \plus u_3),v_1 \plus (v_2 \plus v_3))
	    \llas t_1 \plus (t_2 \plus t_3)
	  \end{align*}
	\item
	  If $t_1 \lra^* \inl(u_1)$, $t_2 \lra^* \inr(v_2)$, and $t_ 3
	  \lra^* \inl(u_3)$, we have
	  \begin{align*}
	    (t_1 \plus t_2) \plus t_3 \lras &\ 
	    \inlr(u_1 \plus u_3,v_1)
	    \llas t_1 \plus (t_2 \plus t_3)
	  \end{align*}
      \end{itemize}

    \item By induction on $A$.  If $A = \one$, then $t_1$ and $t_2$ reduce
      respectively to $a.\star$ and $b.\star$. We have
      $$
	t_1 \plus t_2 \lras (a + b).\star = (b + a).\star \llas t_2 \plus t_1
      $$
      If $A = A_1 \oplus  A_2$, then 
      $t_1$ and $t_2$ reduce either to a term of the form $\inl$, $\inr$, or $\inlr$, leading to nine cases.  As all the cases are similar we only give two examples.
      \begin{itemize}
	\item 
	  If $t_1 \lra^* \inlr(u_1,v_1)$ and $t_2 \lra^*
	  \inlr(u_2,v_2)$, using the induction hypothesis, we have
	  $$
	    t_1 \plus t_2 \lras \inlr(u_1 \plus u_2,v_1 \plus v_2) \equiv
	    \inlr(u_2 \plus u_1,v_2 \plus v_1) \llas t_2 \plus t_1
	  $$

	\item
	  If $t_1 \lra^* \inl(u_1)$ and $t_2 \lra^* \inr(v_2)$, we have
	  $$
	    t_1 \plus t_2 \lras \inlr(u_1,v_2)  \llas t_2 \plus t_1
	  $$
      \end{itemize}

    \item By induction on $A$.  If $A = \one$, then $t$ reduces to
      $c.\star$. We have
      $$
	a \bullet b \bullet t \lras (a \times (b \times c)).\star
	= ((a \times b) \times c).\star \llas (a \times b) \bullet t
      $$
      If $A = A_1 \oplus  A_2$, then $t$ reduces to
      a term of the form $\inl$, $\inr$, or $\inlr$, leading to three cases.
      \begin{itemize}

	\item If $t \lra^* \inlr(u,v)$, using
	  the induction hypothesis, we have
	  $$
	    a \bullet b \bullet t \lras \inlr(a \bullet b \bullet u,a \bullet b
	    \bullet v) \equiv \inlr((a \times b) \bullet u,(a \times b) \bullet v)
	    \llas (a \times b) \bullet t
	  $$

	\item If $t \lra^* \inl(u)$, using
	  the induction hypothesis, we have
	  $$
	    a \bullet b \bullet t \lras \inl(a \bullet b \bullet u)
	    \equiv \inl((a \times b) \bullet u)
	    \llas (a \times b) \bullet t
	  $$

	\item If $t \lra^* \inr(v)$, using
	  the induction hypothesis, we have
	  $$
	    a \bullet b \bullet t \lras \inr(a \bullet b
	    \bullet v) \equiv \inr((a \times b) \bullet v)
	    \llas (a \times b) \bullet t
	  $$
      \end{itemize}

    \item By induction on $A$.  If $A = \one$, then $t_1$ and $t_2$ reduce
      respectively to $b.\star$ and $c.\star$. We have
      $$
	a \bullet (t_1 \plus t_2) \lras (a \times (b + c)).\star
	= (a \times b + a \times c).\star \llas  a \bullet t_1 \plus a \bullet t_2
      $$
      If $A = A_1 \oplus  A_2$, then $t_1$ and $t_2$
      reduce either to a term of the form $\inl$, $\inr$, or $\inlr$, leading to nine cases.  As all the cases are similar we only give two examples.
      \begin{itemize}

	\item If $t_1 \lra^*
	  \inlr(u_1,v_1)$ and 
	  $t_2 \lra^* \inlr(u_2,v_2)$,
	  using the induction
	  hypothesis, we have
	  \begin{align*}
	    a \bullet (t_1 \plus t_2) \lras &\ \inlr(a \bullet (u_1 \plus u_2),a \bullet (v_1 \plus v_2))\\
	    \equiv &\ \inlr(a \bullet u_1 \plus a \bullet u_2,a \bullet v_1 \plus a \bullet v_2) \llas a \bullet t_1 \plus a \bullet t_2
	  \end{align*}

	\item If $t_1 \lra^* \inl(u_1)$ and $t_2 \lra^* \inr(v_2)$, we have
	  \begin{align*}
	    a \bullet (t_1 \plus t_2) \lras &\ \inlr(a \bullet u_1,a
	    \bullet v_2) \llas a \bullet
	    t_1 \plus a \bullet t_2
	  \end{align*}

      \end{itemize}

    \item By induction on $A$.  If $A = \one$, then $t$ reduces to
      $c.\star$. We have
      $$
	(a + b) \bullet t \lras ((a + b) \times c).\star =
	(a \times c + b \times c).\star \llas a \bullet t \plus b \bullet t
      $$
      If $A = A_1 \oplus  A_2$, then $t$ reduces
      to
      a term of the form $\inl$, $\inr$, or $\inlr$, leading to three cases.
      \begin{itemize}

	\item If $t \lra^* \inlr(u,v)$.  Using
	  the induction hypothesis, we have
	  $$
	    (a + b) \bullet t  \lras \inlr((a + b) \bullet u,(a + b) \bullet v)
	    \equiv \inlr(a \bullet u \plus b \bullet u,a \bullet v \plus b \bullet v)
	    \llas a \bullet t \plus b \bullet t
	  $$
	\item If $t \lra^* \inl(u)$.  Using
	  the induction hypothesis, we have
	  $$
	    (a + b) \bullet t  \lras \inl((a + b) \bullet u)
	    \equiv \inl(a \bullet u \plus b \bullet u)
	    \llas a \bullet t \plus b \bullet t
	  $$

	\item If $t \lra^* \inr(v)$.  Using
	  the induction hypothesis, we have
	  \[
	    (a + b) \bullet t  \lras \inr((a + b) \bullet v)
	    \equiv \inr(a \bullet v \plus b \bullet v)
	    \llas a \bullet t \plus b \bullet t
	    \tag*{\qedhere}
	  \]
      \end{itemize}
  \end{enumerate}
\end{proof}
\begin{rem}
As the zero vector of, for example $\mathbb{C}^2$, has several representations:
$\inlr(0.\star, 0.\star)$, $\inl(0.\star)$, and $\inr(0.\star)$, the properties
of proofs of vector propositions (\autoref{vecstructure}) are those of a
commutative semi-module, rather than those of a vector space. If we represent
the vectors with a connective, such as conjunction or the ``sup'' connective
\cite{DiazcaroDowekTCS23}, where we have weaker introduction rules
(only one introduction rule, corresponding to $\inlr$, but no $\inl$
and $\inr$), we get a vector space structure~\cite[Lemma
  3.4]{DiazcaroDowekMSCS24}. Yet, the three misssing properties $t + 0
= t$, $t + (-t) = 0$, and $1 . t = t$ are not useful to represent
vectors and matrices.
\end{rem}

We also need to introduce a notion of {\em
  elimination context} that is a standard generalization of the notion
of head variable.  In the $\lambda$-calculus, we can decompose a term
$t$ as a sequence of applications $t = u~v_1~\ldots~v_n$, with terms
$v_1, \ldots, v_n$ and a term $u$, which is not an application.  Then
$u$ may either be a variable, in which case it is the head variable of
the term, or an abstraction.  Similarly, any proof in the quantum
in-left-right-calculus without $\elimplus^{nd}$ can be decomposed into
a sequence of elimination rules, forming an elimination context, and a
proof $u$ that is either a variable, an introduction, a sum, or a
product.

\begin{defi}[Elimination context]
  An elimination context is a proof with a single free variable, written
  $\_$, that is a proof in the language
  $$
    K = \_
    \mid \elimone(K,u)
    \mid K~u
    \mid \elimplus(K,\abstr{x}r,\abstr{y}s)
  $$
  where $u$ is a closed proof,
  $FV(r) \subseteq \{x\}$, and $FV(s) \subseteq \{y\}$.
\end{defi}

To prove \autoref{thm:converse}, we will first prove that for any $A$, and any $B \in {\mathcal V}$, if $t$ a
  closed proof of $A \multimap B$, that does not contain the symbol $\elimplus^{nd}$, and $u_1$ and $u_2$ and two closed proofs of
  $A$. Then
  \(
    t~(u_1 \plus u_2) \equiv t~u_1 \plus t~u_2
  \) and 
  \(
    t~(a\bullet u_1) \equiv a\bullet t~u_1
  \) (\autoref{cor:cor1a}).
But, it is more
convenient to
prove first the following equivalent statement (\autoref{linearity}): For a proof $t$ of $B$
such that $x:A \vdash t:B$
$$
  (u_1 \plus u_2/x)t \equiv (u_1/x)t \plus (u_2/x)t
  \quad\textrm{and}\quad (a \bullet u_1/x)t \equiv a \bullet (u_1/x)t
$$

We proceed by induction on the measure $\mu(t)$ of the proof $t$, but
the case analysis is non-trivial. Indeed, when $t$ is an elimination,
for example when $t = t_1~t_2$, the variable $x$ must occur in $t_1$,
and we would like to apply the induction hypothesis to this proof.
But we cannot because $t_1$ is a proof of an implication, that is not
in $\mathcal{V}$.  This leads us to first decompose the proof $t$ into
a proof of the form $(t'/\_)K$ where $K$ is an elimination context and
$t'$ is either the variable $x$, an introduction, a sum, or a product,
and analyse the different possibilities for $t'$. The cases where $t'$
is an introduction, a sum or a product are easy, but the case where it
is the variable $x$, that is where $t = (x/\_)K$, is more complex.
Indeed, in this case, we need to prove
$$
(u_1 \plus u_2/\_)K \equiv (u_1/\_)K \plus (u_2/\_)K
\qquad\textrm{and}\qquad
(a \bullet u_1/\_)K \equiv a \bullet (u_1/\_)K
$$
and this leads to a second case analysis where we consider the last elimination
rule of $K$ and how it interacts with $u_1$ and $u_2$.

For example, when $K = (\elimplus(\_,\abstr{y}r,\abstr{z}s)/\_)K_1$,
then $u_1$ and $u_2$ are closed proofs of an additive disjunction
$\oplus$, thus they reduce to two introductions, for example
$\inlr(u_{11},u_{12})$ and $\inlr(u_{21},u_{22})$, and $(u_1 \plus
u_2/\_) K$ reduces to $(u_{11} \plus u_{21}/y)(r/\_)K_1 \plus (u_{12}
\plus u_{23}/z)(s/\_)K_1$.  And we conclude by applying the induction
hypothesis to the irreducible form of $(r/\_)K_1$ and $(s/\_)K_1$.

In the $\lambda$-calculus, we can decompose a term $t$ as a sequence of
applications $t = u~v_1~\ldots~v_n$, with terms $v_1, \ldots, v_n$ and a term
$u$, which is a variable or an abstraction. Similarly, any proof in the quantum
in-left-right-calculus without $\elimplus^{nd}$ can be decomposed into a
sequence of elimination rules, forming an elimination context, and a proof $u$
that is either a variable, an introduction, a sum, or a product.

\begin{lem}[Decomposition of a proof]
  \label{elim}
  If $t$ is an irreducible proof such that $x:C \vdash t:A$, then there
  exist an elimination context $K$, a proof $u$, and a proposition $B$,
  such that $\_:B \vdash K:A$, $x:C \vdash u:B$, $u$ is either the
  variable $x$, an introduction, a sum, or a product, and $t = K\{u\}$.
\end{lem}
\begin{proof}
  By induction on the structure of $t$.

\begin{itemize}
    \item If $t$ is the variable $x$, an introduction, a sum, or a
     product, we take $K = \_$, $u = t$, and $B = A$.

    \item If $t = \elimone(t_1,t_2)$, then $t_1$ is not a closed proof as
      otherwise it would be a closed irreducible proof of $\one$, hence,
      by \autoref{introductionslinear}, it would be an introduction and $t$
      would not be irreducible. Thus, by the inversion property, $x:C
      \vdash t_1:\one$ and $\vdash t_2:A$.

      By induction hypothesis, there exist $K_1$, $u_1$ and $B_1$ such
      that $\_:B_1 \vdash K_1:\one$, $x:C \vdash u_1:B_1$, and $t_1 =
      K_1\{u_1\}$.  We take $u = u_1$, $K = \elimone(K_1,t_2)$, and $B =
      B_1$.  We have $\_:B \vdash K:A$, $x:C \vdash u:B$, and $K\{u\} =
      \elimone(K_1\{u_1\},t_2) = t$.
      
    \item If $t = t_1~t_2$,
          we apply the same method as for the case $t = \elimone(t_1,t_2)$.

    \item If $t = \elimplus(t_1,\abstr{y}t_2,\abstr{z}t_3)$, then
      $t_1$ is not a closed proof as otherwise it would be a closed
      irreducible proof of an additive disjunction $\oplus$, hence, by
      the introduction property, \autoref{introductionslinear}, it
      would be an introduction, and $t$ would not be
      irreducible. Thus, by the inversion property, $x:C \vdash
      t_1:D_1 \oplus D_2$, $y:D_1 \vdash t_2:A$, and $z:D_2 \vdash
      t_3:A$.

      By induction hypothesis, there exist $\_:B_1 \vdash K_1:D_1 \oplus
      D_2$, $x:C \vdash u_1:B_1$, and $t_1 = K_1\{u_1\}$.  We take $u =
      u_1$, $K = \elimplus(K_1,\abstr{y}t_2,\abstr{z} t_3)$, and
      $B = B_1$.  We have
      $\_:B \vdash K:A$, $x:C \vdash u:B$, and $K\{u\} =
      \elimplus(K_1\{u_1\},\abstr{y}t_2,\abstr{z}t_3) = t$.
      \qedhere
\end{itemize}
\end{proof}

A second lemma shows that we can always decompose an elimination
context $K$ different from~$\_$ into a smaller elimination context
$K_1$ and a last elimination rule $K_2$.  This is similar to the fact
that we can always decompose a non-empty list into a smaller list and
its last element.

\begin{lem}[Decomposition of an elimination context]
  \label{horrible}
  If $K$ is an elimination context such that $\_:A \vdash K:B$
  and $K \neq \_$, then $K$ has the form $K_1\{K_2\}$
  where $K_1$ is an elimination context and
  $K_2$ is an elimination context formed with a single
  elimination rule, that is the elimination rule of the top symbol of $A$.
\end{lem}

\begin{proof}
  As $K$ is not $\_$, it has the form $K = L_1\{L_2\}$.
  If $L_2 = \_$, we take $K_1 = \_$, $K_2 = L_1$ and, as the proof is
  well-formed, $K_2$ must be an elimination of the top symbol of $A$.
  Otherwise, by induction hypothesis, $L_2$ has the form $L_2 = K'_1\{K'_2\}$,
  and $K'_2$ is an elimination of the top symbol of $A$.
  Hence, $K = L_1\{K'_1\{K'_2\}\}$. We take $K_1 = L_1\{K'_1\}$, $K_2 = K'_2$.
  \qedhere
\end{proof}

We are now ready to prove that
if $t$ is a closed proof of $A \multimap B$, $B \in {\mathcal V}$,
then 
$$
  t~(u_1 \plus u_2) \equiv t~u_1 \plus t~u_2
  \qquad\qquad\textrm{and}\qquad\qquad
  t~(a \bullet u_1) \equiv a \bullet t~u_1
$$
In fact, it is more convenient to prove first the following equivalent
statement: For a proof $t$ of $B$ such that $x:A \vdash t:B$
$$
  t\{u_1 \plus u_2\} \equiv t\{u_1\} \plus t\{u_2\}
  \qquad\qquad\textrm{and}\qquad\qquad
  t\{a \bullet u_1\} \equiv a \bullet t\{u_1\}
$$

For example, when $K = K_1 \{\elimplus(\_,\abstr{y}r,\abstr{z}s)\}$,
then $u_1$ and $u_2$ are closed proofs of an additive disjunction
$\oplus$, thus they reduce to two introductions, for example
$\inlr(u_{11},u_{12})$ and $\inlr(u_{21},u_{22})$, and $K\{u_1 \plus
u_2\}$ reduces to $K_1\{r\}\{u_{11} \plus u_{21}\} \plus
K_1\{s\}\{u_{12} \plus u_{23}\}$. So, we need to apply the induction
hypothesis to the irreducible form of $K_1 \{r\}$ and $K_1\{s\}$. To
prove that these proofs are smaller than $t$, we need
Lemmas~\ref{lem:msubst}, \ref{lem:mured1}, and \ref{lem:mured2}.

In fact, 
the case of the elimination of the
connective $\one$ is simpler because no substitution occurs in $r$
in this case and the case of the elimination of the implication is
simpler because this rule is just the modus ponens and not the
generalized elimination rule of this connective. Thus, the only remaining
case is that of the elimination rule of the additive disjunction
$\oplus$.

\begin{thm}[Linearity]
  \label{linearity}
    Let $A$ and $B \in {\mathcal V}$ be propositions, $t$ a proof such that $x:A \vdash t:B$, and $u_1$ and $u_2$ and two closed proofs of $A$, then 
    $$
      t\{u_1 \plus u_2\} \equiv t\{u_1\} \plus t\{u_2\}
      \qquad\qquad\textrm{and}\qquad\qquad
      t\{a \bullet u_1\} \equiv a \bullet t\{u_1\}
    $$
\end{thm}
\begin{proof}
Without loss of generality, we can assume that $t$ is irreducible.
  We proceed by induction on $\mu(t)$. 

  Using \autoref{elim}, the term $t$ can be decomposed as $K\{t'\}$ where $t'$ is either an introduction, a sum, a product, or the variable $x$.

  \begin{itemize}
    \item 
      If $t'$ is an introduction, as $t$ is irreducible, $K = \_$ and
      $t'$ is a proof of $B \in {\mathcal V}$, $t'$ is either
      $a.\star$, $\inl(t_1)$, $\inr(t_2)$, or $\inlr(t_1,t_2)$. However,
      since $a.\star$ is not a proof in $x:A$, it is
      $\inl(t_1)$, $\inr(t_2)$, or $\inlr(t_1,t_2)$. 
      Using the induction hypothesis with $t_1$ and with $t_2$
      ($\mu(t_1) < \mu(t')$, $\mu(t_2) < \mu(t')$),
      in the first case, we get 
      \begin{align*}
	t\{u_1 \plus u_2\}
	\equiv
	\inl(t_1\{u_1\} \plus t_1\{u_2\})
	\lla 
	t\{u_1\} \plus t\{u_2\}
      \end{align*}
      And
      \begin{align*}
	t\{a \bullet u_1\}
	\equiv \inl(a \bullet t_1\{u_1\})
	&
	\lla
	a \bullet t\{u_1\}
      \end{align*}

The second case is similar. In the third,      
      we get
      \begin{align*}
	t\{u_1 \plus u_2\}
	\equiv
	\inlr(t_1\{u_1\} \plus t_1\{u_2\},t_2\{u_1\} \plus t_2\{u_2\})
	\lla 
	t\{u_1\} \plus t\{u_2\}
      \end{align*}
      And
      \begin{align*}
	t\{a \bullet u_1\}
	\equiv \inlr(a \bullet t_1\{u_1\},a \bullet t_2\{u_1\})
	&
	\lla
	a \bullet t\{u_1\}
      \end{align*}

    \item If $t' = t_1 \plus t_2$, then using the induction hypothesis
      with $t_1$, $t_2$, and $K$ (notice that, by \autoref{lem:msubst}, $\mu(t_1) < \mu(t)$, $\mu(t_2) <
      \mu(t)$, and $\mu(K) < \mu(t)$) and \autoref{vecstructure} (1.,
      2., and 4.), we get
      \begin{align*}
	t\{u_1 \plus u_2\}
&	\equiv K\{(t_1\{u_1\} \plus t_1\{u_2\})
	\plus (t_2\{u_1\} \plus t_2\{u_2\})\}\\
	&\equiv 
	K\{( t_1\{u_1\} \plus t_2\{u_1\})
	\plus (t_1\{u_2\} \plus t_2\{u_2\})\}
	\equiv  
	t\{u_1\} \plus t\{u_2\}
      \end{align*}
      And 
      \begin{align*}
	t\{a \bullet u_1\}
&	\equiv K\{a \bullet t_1\{u_1\} \plus a \bullet t_2\{u_1\}\}
	\equiv K\{a \bullet (t_1\{u_1\} \plus t_2\{u_1\})\}
	\equiv 
	a \bullet t\{u_1\}
      \end{align*}

    \item
      If $t' = b \bullet t_1$, then using the induction hypothesis
      with $t_1$ and $K$ (notice that, by \autoref{lem:msubst}, $\mu(t_1) < \mu(t)$, $\mu(K) < \mu(t)$) and
      $K$ and \autoref{vecstructure} (4.~and 3.), we get
      \begin{align*}
	t\{u_1 \plus u_2\}
	\equiv K\{b \bullet (t_1 \{u_1\} \plus t_1\{u_2\})\}
	& \equiv K\{b \bullet t_1 \{u_1\} \plus b \bullet t_1\{u_2\}\}
	 \equiv 
	t\{u_1\} \plus t\{u_2\}
      \end{align*}
      And 
      \begin{align*}
	t\{a \bullet u_1\}
&	\equiv K\{b \bullet a \bullet t_1 \{u_1\}\}
\equiv K\{a \bullet b \bullet t_1 \{u_1\}\}
	\equiv 
	a \bullet t\{u_1\}
      \end{align*}

    \item If $t'$ is the variable $x$, we need to prove
      $$
	K\{u_1 \plus u_2\} \equiv K\{u_1\} \plus K\{u_2\}
	\qquad\textrm{and}\qquad
	K\{a \bullet u_1\} \equiv a \bullet K\{u_1\}
      $$
      By \autoref{horrible},
      $K$ has the form $K_1\{K_2\}$ and $K_2$ is an elimination of the top symbol of $A$.
      We consider the various cases for $K_2$. 
      \begin{itemize}
	\item If $K = K_1\{\elimone(\_,r)\}$, then $u_1$ and $u_2$ are closed
	  proofs of $\one$, thus $u_1 \lras b.\star$ and $u_2\lras c.\star$.
	  Using the induction hypothesis with the proof $K_1$
	  (notice that, by \autoref{lem:msubst}, $\mu(K_1) < \mu(K) = \mu(t)$) and \autoref{vecstructure} (5.~and 3.)
	  \begin{align*}
	    K\{u_1 \plus u_2\}
	    &\lras K_1 \{\elimone({b.\star} \plus c.\star,r)\}
	    \lras K_1 \{(b + c) \bullet r\}
	    \equiv (b + c) \bullet K_1 \{r\}\\
	    &\equiv b \bullet K_1 \{r\} \plus c \bullet K_1 \{r\}
	    \equiv K_1 \{b \bullet r\} \plus K_1 \{c \bullet r\}\\
	    & \llas K_1 \{\elimone(b.\star,r)\} \plus K_1 \{\elimone(c.\star,r)\}
	    \llas K\{u_1\} \plus K\{u_2\}
	  \end{align*}
	  And
	  \begin{align*}
	    K\{a \bullet u_1\}
	    &\lras K_1 \{\elimone(a \bullet b.\star,r)\}
	    \lras K_1 \{(a \times b) \bullet r\}
	    \equiv (a \times b) \bullet  K_1 \{r\}\\
	   & \equiv a \bullet b \bullet  K_1 \{r\}
	    \equiv a \bullet  K_1 \{b \bullet r\}
	     \llas a \bullet K_1 \{\elimone(b.\star,r)\}
	    \llas a \bullet K\{u_1\}
	  \end{align*}

	\item
	  If $K = K_1 \{\_~s\}$, then $u_1$ and $u_2$ are closed 
	  proofs of an implication, thus 
	  $u_1 \lras \lambda \abstr{y} u'_1$
	  and
	  $u_2 \lras \lambda \abstr{y} u'_2$.
	  Using the induction hypothesis with the proof $K_1 $
	  (notice that, by \autoref{lem:msubst}, $\mu(K_1 ) < \mu(K) = \mu(t)$), we get
	  \begin{align*}
	    K\{u_1 \plus u_2\}
	    &\lras K_1 \{(\lambda \abstr{y} u'_1 \plus \lambda \abstr{y} u'_2)~s\}
	    \lras K_1 \{u'_1\{s\} \plus u'_2\{s\}\}\\
	    &\equiv K_1 \{u'_1\{s\}\} \plus K_1 \{u'_2\{s\}\}
	     \llas K_1 \{(\lambda \abstr{y} u'_1)~s\} \plus K_1 \{(\lambda \abstr{y} u'_2)~s\}\\
	   & \llas K\{u_1\} \plus K\{u_2\}
	  \end{align*}
	  And
	  \begin{align*}
	    K\{a \bullet u_1\}
	    &\lras K_1 \{(a \bullet \lambda \abstr{y} u'_1)~s\}
	    \lras K_1 \{a \bullet u'_1\{s\}\}\\
	   & \equiv a \bullet  K_1 \{u'_1\{s\}\}
	     \lla
	    a \bullet K_1 \{(\lambda \abstr{y} u'_1)~s\}
	    \llas a \bullet K\{u_1\}
	  \end{align*}

	\item
	  If $K = K_1 \{\elimplus(\_,\abstr{y}r,\abstr{y}s)\}$, then
          let $r'$ be the irreducible form of $K_1\{r\}$ and $s'$ be
          the irreducible form of $K_1\{s\}$. We want to use the
          induction hypothesis with the proofs $K_1$, $r'$, and
          $s'$. To do so remark that by \autoref{lem:msubst},
          $\mu(K_1 ) < \mu(K_1) +
          \mu(\elimplus(\_,\abstr{y}r,\abstr{y}s)) = \mu(K) = \mu(t)$.
          Thus, $\mu(K_1) < \mu(t)$.
          By
	  Lemmas~\ref{lem:msubst}, \ref{lem:mured1}, and
          \ref{lem:mured2}, $\mu(r') \leq \mu(K_1 \{r\}) = \mu(K_1 ) +
          \mu(r) < \mu(K_1 ) + 1 + \mu(r) \leq \mu(K_1) + 1 +
          \max(\mu(r),\mu(s))) = \mu(K_1) +
          \mu(\elimplus(\_,\abstr{y}r,\abstr{y}s)) = \mu(K) =
          \mu(t)$. Thus, $\mu(r') < \mu(t)$. In the same way, $\mu(s')
          < \mu(t)$.

          Let us first prove that 
	  $$
	    K\{u_1 \plus u_2\} \equiv K\{u_1\} \plus K\{u_2\}
	  $$
          The proof $u_1$ is a closed proof of an additive disjunction
          $\oplus$, thus $u_1 \lras \inl(u_{11})$, $u_1 \lras
          \inr(u_{12})$, or $u_1 \lras \inlr(u_{11},u_{12})$.  In the
          same way, $u_2 \lras \inl(u_{21})$, $u_2 \lras
          \inr(u_{22})$, or $u_2 \lras \inlr(u_{21},u_{22})$, leading
          to nine cases.

          \begin{enumerate}
          \item\label{caseA} If $u_1 \lras \inl(u_{11})$ and $u_2 \lras \inl(u_{21})$,
          then using the induction hypothesis with the proof $r'$ 
\begin{align*}
     K\{u_1 \plus u_2\}
     &\lras K_1 \{\elimplus(\inl(u_{11}) \plus \inl(u_{21}),
        \abstr{y}r, \abstr{z}s)\}\\
     &\lras K_1 \{\elimplus(\inl(u_{11} \plus u_{21}),
        \abstr{y}r, \abstr{z}s)\}\\
     &\lras K_1 \{r\{u_{11} \plus u_{21}\}\}\\
     &\lras r'\{u_{11} \plus u_{21}\}\\
     &\equiv r'\{u_{11}\} \plus r'\{u_{21}\}\\
     &\llas  K_1\{r\{u_{11}\}\} \plus K_1\{r\{u_{21}\}\}\\
     &\llas K_1 \{\elimplus(\inl(u_{11}),\abstr{y}r,\abstr{z}s)\}
            \plus K_1 \{\elimplus(\inl(u_{21}),\abstr{y}r,\abstr{z}s)\}\\
     &\llas K\{u_1\} \plus K\{u_2\}
\end{align*}

       \item\label{caseB} If $u_1 \lras \inl(u_{11})$ and $u_2 \lras \inr(u_{22})$,
         then using the induction hypothesis with the proof $K_1$
\begin{align*}
     K\{u_1 \plus u_2\}
     &\lras K_1 \{\elimplus(\inl(u_{11}) \plus \inr(u_{22}),
        \abstr{y}r, \abstr{z}s)\}\\
     &\lras K_1 \{\elimplus(\inlr(u_{11},u_{22}),
        \abstr{y}r, \abstr{z}s)\}\\
     &\lras K_1 \{r\{u_{11}\} \plus s\{u_{22}\}\}\\
     &\equiv K_1\{r\{u_{11}\}\} \plus K_1\{s\{u_{22}\}\}\\
     &\llas K_1 \{\elimplus(\inl(u_{11}),\abstr{y}r,\abstr{z}s)\}
            \plus K_1 \{\elimplus(\inr(u_{22}),\abstr{y}r,\abstr{z}s)\}\\
     &\llas K\{u_1\} \plus K\{u_2\}
\end{align*}

       \item\label{caseC} If $u_1 \lras \inl(u_{11})$ and $u_2 \lras
         \inlr(u_{21},u_{22})$, then using the induction hypothesis
         with the proofs $K_1$ and $r'$ and
         \autoref{vecstructure} (1.)
\begin{align*}
     K\{u_1 \plus u_2\}
     &\lras K_1\{\elimplus(\inl(u_{11}) \plus \inlr(u_{21},u_{22}),
        \abstr{y}r, \abstr{z}s)\}\\
     &\lras K_1 \{\elimplus(\inlr(u_{11} \plus u_{21},u_{22}),
        \abstr{y}r, \abstr{z}s)\}\\
     &\lras K_1 \{r\{u_{11} \plus u_{21}\} \plus s\{u_{22}\}\}\\
        &\equiv K_1 \{r\{u_{11} \plus u_{21}\}\} \plus K_1\{s\{u_{22}\}\}\\
     &\lras r'\{u_{11} \plus u_{21}\} \plus K_1\{s\{u_{22}\}\}\\
     &\equiv (r'\{u_{11}\} \plus r'\{u_{21}\}) \plus K_1\{s\{u_{22}\}\}\\
     &\equiv r'\{u_{11}\} \plus (r'\{u_{21}\} \plus K_1\{s\{u_{22}\}\})\\
     &\llas K_1\{r\{u_{11}\}\} \plus
        (K_1 \{r\{u_{21}\}\} \plus K_1\{s\{u_{22}\}\})\\
     &\equiv K_1\{r\{u_{11}\}\} \plus
        (K_1\{r\{u_{21}\} \plus s\{u_{22}\}\})\\
     &\llas K_1 \{\elimplus(\inl(u_{11}),\abstr{y}r,\abstr{z}s)\}
            \plus K_1 \{\elimplus(\inlr(u_{21},u_{22}),\abstr{y}r,\abstr{z}s)\}\\
     &\llas K\{u_1\} \plus K\{u_2\}
\end{align*}

       \item If $u_1 \lras \inr(u_{12})$ and $u_2 \lras \inl(u_{21})$,
         the proof is similar to that of case~\eqref{caseB}, except that we use
         the induction hypothesis with the proof $K_1$ and 
         \autoref{vecstructure} (2.).

       \item If $u_1 \lras \inr(u_{12})$ and $u_2 \lras \inr(u_{22})$,
         the proof is similar to that of case~\eqref{caseA}, except that we use the
         induction hypothesis with the proof
         $s'$.
         
       \item If $u_1 \lras \inr(u_{12})$ and $u_2 \lras
         \inlr(u_{21},u_{22})$, the proof is similar to that of case~\eqref{caseC}, except that we use
         the induction hypothesis with the proofs $K_1$ and $s'$ and
         \autoref{vecstructure} (1.~and 2.).

       \item If $u_1 \lras \inlr(u_{11},u_{12})$ and $u_2 \lras
         \inl(u_{21})$, the proof is similar to that of case~\eqref{caseC},
         except that we use the induction hypothesis with the proofs
         $K_1$ and $r'$ and \autoref{vecstructure} (1.~and 2.)

       \item If $u_1 \lras \inlr(u_{11},u_{12})$ and $u_2 \lras
         \inr(u_{22})$, the proof is similar to that of case~\eqref{caseC},
         except that we use the induction hypothesis with the proofs $K_1$ and
         $s'$.

       \item If $u_1 \lras \inlr(u_{11},u_{12})$ and $u_2 \lras
         \inlr(u_{21},u_{22})$, then using the induction hypothesis
         with the proofs $K_1$, $r'$, and $s'$, and
         \autoref{vecstructure} (1. and 2.)
\begin{align*}
     K\{u_1 \plus u_2\}
     &\lras K_1 \{\elimplus(\inlr(u_{11},u_{12}) \plus \inlr(u_{21},u_{22}),
        \abstr{y}r, \abstr{z}s)\}\\
     &\lras K_1 \{\elimplus(\inlr(u_{11} \plus u_{21},u_{12} \plus u_{22}),
        \abstr{y}r, \abstr{z}s)\}\\
     &\lras K_1 \{r\{u_{11} \plus u_{21}\} \plus s\{u_{12} \plus u_{22}\}\}\\
     &\equiv K_1 \{r\{u_{11} \plus u_{21}\}\} \plus K_1\{s\{u_{12} \plus u_{22}\}\}\\
     &\lras r'\{u_{11} \plus u_{21}\} \plus s'\{u_{12} \plus u_{22}\}\\
     &\equiv (r'\{u_{11}\} \plus r'\{u_{21}\}) \plus (s'\{u_{12}\} \plus s'\{u_{22}\})\\
     &\equiv (r'\{u_{11}\} \plus s'\{u_{12}\}) \plus (r'\{u_{21}\} \plus s'\{u_{22}\})\\
     &\llas (K_1\{r\{u_{11}\}\} \plus K_1\{s\{u_{12}\}\}) \plus
        (K_1 \{r\{u_{21}\}\} \plus K_1\{s\{u_{22}\}\})\\
     &\equiv (K_1\{r\{u_{11}\} \plus s\{u_{12}\}\}) \plus
        (K_1\{r\{u_{21}\} \plus s\{u_{22}\}\})\\
     &\llas K_1 \{\elimplus(\inlr(u_{11},u_{12}),\abstr{y}r,\abstr{z}s)\}
            \plus K_1 \{\elimplus(\inlr(u_{21},u_{22}),\abstr{y}r,\abstr{z}s)\}\\
     &\llas K\{u_1\} \plus K\{u_2\}
\end{align*}
\end{enumerate}

          Let us then prove that 
	  $$
	    K\{a \bullet u_1\} \equiv a \bullet K\{u_1\}
	  $$
	    The proof
          $u_1$ is a closed proof of an additive disjunction $\oplus$,
          thus $u_1 \lras \inl(u_{11})$, $u_1 \lras \inr(u_{12})$, or
          $u_1 \lras \inlr(u_{11},u_{12})$, leading to three cases.

          \begin{enumerate}
          \item\label{caseAbullet} If $u_1 \lras \inl(u_{11})$, then using the
              induction hypothesis with the proof $r'$
	  \begin{align*}
	    K\{a \bullet u_1\}
	    &\lras K_1 \{\elimplus(a \bullet \inl(u_{11}),\abstr{y}r,\abstr{z}s)\}\\
	    &\lras K_1 \{\elimplus(\inl(a \bullet u_{11}), \abstr{y}r,\abstr{z}s)\}\\
	    &\lras K_1 \{r\{a \bullet u_{11}\}\}\\
            &\lras r'\{a \bullet u_{11}\}\\
            & \equiv a \bullet r'\{u_{11}\}\\
	    &\llas a \bullet  K_1 \{r\{u_{11}\}\}\\
            &\llas a \bullet K_1 \{\elimplus(\inl(u_{11}),\abstr{y}r,\abstr{z}s)\}\\
	   & \llas a \bullet K\{u_1\}
	  \end{align*}
          
          \item If $u_1 \lras \inr(u_{12})$, 
 the proof is similar to that of case~\eqref{caseAbullet}, except that we use the
 induction hypothesis with the proof $s'$.

       \item If $u_1 \lras \inlr(u_{11},u_{12})$, then using the
              induction hypothesis with the proofs $K_1$, $r'$, and
              $s'$, and \autoref{vecstructure} (4.)
	  \begin{align*}
	    K\{a \bullet u_1\}
	    &\lras K_1 \{\elimplus(a \bullet \inlr(u_{11},u_{12}),\abstr{y}r,\abstr{z}s)\}\\
	    &\lras K_1 \{\elimplus(\inlr(a \bullet u_{11},a \bullet u_{12}), \abstr{y}r,\abstr{z}s)\}\\
	    &\lras K_1 \{r\{a \bullet u_{11}\} \plus s\{a \bullet u_{12}\}\}\\
      	    &\equiv K_1 \{r\{a \bullet u_{11}\}\} \plus K_1\{s\{a \bullet u_{12}\}\}\\
            &\lras r'\{a \bullet u_{11}\} \plus s'\{a \bullet u_{12}\}\\
            & \equiv a \bullet r'\{u_{11}\} \plus a \bullet s'\{u_{12}\}\\
            & \equiv a \bullet (r'\{u_{11}\} \plus s'\{u_{12}\})\\
	    &\llas a \bullet (K_1 \{r\{u_{11}\}\} \plus K_1 \{s\{u_{12}\}\})\\
       	    &\equiv a \bullet K_1 \{r\{u_{11}\} \plus s\{u_{12}\}\}\\
            &\llas a \bullet K_1 \{\elimplus(\inlr(u_{11},u_{12}),\abstr{y}r,\abstr{z}s)\}\\
	    & \llas a \bullet K\{u_1\}
 \tag*{\qedhere}
\end{align*}
\end{enumerate}
\end{itemize}
\end{itemize}
\end{proof}

\begin{rem} The \autoref{linearity} and its
  corollaries hold for linear proofs, but not for non-linear ones. The
  linearity is used, in an essential way, in two places. First, in the first
  case of the proof of \autoref{linearity} when we
  remark that $a.\star$ is not a proof in the context $x:A$. Indeed, if $t$
  could be $a.\star$ then linearity would be violated as $1.\star\{4.\star
  \plus 5.\star\} = 1.\star$ and $1.\star\{4.\star\} \plus 1.\star\{5.\star\}
  \equiv 2.\star$.  Then, in the proof of
  \autoref{elim}, we remark that when
  $t_1~t_2$ is a proof in the context $x:A$ then $x$ must occur in $t_1$, and
  hence it does not occur in $t_2$, that is therefore closed. This way, the
  proof $t$ eventually has the form $K\{u\}$ and this would not be the case if
  $x$ could occur in $t_2$ as well.
\end{rem}

\begin{cor}
\label{cor:cor1a}
  Let $A$ and $B$ be propositions, with $B \in {\mathcal V}$, $t$ a
  closed proof of $A \multimap B$, that does not contain the symbol
  $\elimplus^{nd}$, and $u_1$ and $u_2$ and two closed proofs of
  $A$. Then
  $$
    t~(u_1 \plus u_2) \equiv t~u_1 \plus t~u_2
    \qquad\qquad\textrm{and}\qquad\qquad
    t~(a\bullet u_1) \equiv a\bullet t~u_1
  $$
\end{cor}
\begin{proof}
  As $t$ is a closed proof of $A \multimap B$, using
  \autoref{introductionslinear}, it reduces to an irreducible proof of the
  form $\lambda \abstr{x} t'$.  Let $u'_1$ be the irreducible form of
  $u_1$, and $u'_2$ that of $u_2$.

If $B \in {\mathcal V}$, using \autoref{linearity}, we have
$$
t~(u_1 \plus u_2) \lras t'\{u'_1 \plus u'_2\} \equiv t'\{u'_1\} \plus t'\{u'_2\} \llas (t~u_1) \plus (t~u_2)
$$
hence 
\(
t~(u_1 \plus u_2)  \equiv  (t~u_1) \plus (t~u_2)
\).
And
$$
 t~(a\bullet u_1) \lras t'\{a \bullet u'_1\} \equiv a \bullet t'\{u'_1\} \llas a\bullet (t~u_1)
 $$
 hence
\(
 t~(a\bullet u_1)  \equiv  a\bullet (t~u_1)
 \).
  \qedhere
\end{proof}

\begin{proof}[Proof of \autoref{thm:converse}]
  Using \autoref{cor:cor1a} and \autoref{parallelsum}, we have
  \begin{align*}
    F({\bf u} + {\bf v}) = \underline{t~\overline{\bf u + \bf v}^A}
    = \underline{t~(\overline{\bf u}^A \plus \overline{\bf v}^A)}
    & =
    \underline{t~\overline{\bf u}^A \plus t~\overline{\bf v}^A}
    =
    \underline{t~\overline{\bf u}^A} +
    \underline{t~\overline{\bf v}^A} =
    F({\bf u}) + F({\bf v})
    \\
    F(a {\bf u}) = \underline{t~\overline{a \bf u}^A}
    = \underline{t~(a \bullet \overline{\bf u}^A)} 
    & =
    \underline{a\bullet t~\overline{\bf u}^A}
    = a \underline{t~\overline{\bf u}^A} = a F({\bf u})
    &\tag*{\qedhere}
  \end{align*}
\end{proof}

\begin{rem}
The theorem cannot extend to the full calculus as the symbol
$\elimplus^{nd}$ is used to define quantum measurement, that is a
non-linear operation, as we will see in \autoref{sec:quantumcomputing}.
\end{rem}

\begin{rem}[No-cloning]
With respect to the in-left-right-+-calculus, the quantum
in-left-right-calculus introduces three differences: 
restricting the deduction rules to linear ones,
adding the rule prod, and
adding the scalars. 
If we had build an extension of the
in-left-right-+-calculus with a rule prod and with scalars, but had
kept the intuitionistic deduction rules, instead of the linear ones,
we could clone Qbits, which is a forbidden operation in quantum
physics~\cite{WoottersZurekNature82}. Indeed, on such a calculus,
consider the proof $c$
\begin{align*}
\lambda \abstr{x}~\elimor(x, 
    &\abstr{y_1} \elimtop(y_1,\elimor(x, \abstr{z_1} \elimtop(z_1,\ket{00}),\abstr{z_2} \elimtop(z_2,\ket{01}))), \\
    &\abstr{y_2} \elimtop(y_2,\elimor(x, \abstr{z_3} \elimtop(z_3,\ket{10}), \abstr{z_4} \elimtop(z_4,\ket{11})))
\end{align*}
where
$\ket{00}$ is a notation for
  $\inlr(\inlr(1.\star,0.\star),\inlr(0.\star,0.\star))$,
$\ket{01}$ for
  $\inlr(\inlr(0.\star,1.\star),\inlr(0.\star,0.\star))$,
$\ket{10}$ for
  $\inlr(\inlr(0.\star,0.\star),\inlr(1.\star,0.\star))$,
and $\ket{11}$ for
  $\inlr(\inlr(0.\star,0.\star),\inlr(0.\star,1.\star))$.

The proof $c~\inlr(a.\star,b.\star)$ would reduce to the proof
$\inlr(\inlr((a \times a).\star,(a\times b).\star),\inlr((b \times
a).\star, (b \times b).\star))$. Thus, the term $c$ would express the
cloning function from ${\mathbb C}^2$ to ${\mathbb C}^4$, mapping
$\left(\begin{smallmatrix} a\\b\end{smallmatrix}\right)$ to
$\left(\begin{smallmatrix} a^2\\ab\\ab\\b^2\end{smallmatrix}\right)$,
which is the tensor product of the vector $\left(\begin{smallmatrix}
  a\\b\end{smallmatrix}\right)$ with itself.

But, by \autoref{corollary2}, no proof of $(\one \oplus \one)
\multimap ((\one \oplus \one) \oplus (\one \oplus \one))$, in the
quantum in-left-right-calculus without $\elimplus^{nd}$, can express
this function, because it is not linear.
\end{rem}

\subsection{Proof of \autoref{corollary2}}
\begin{lem}
  \label{cor:cor1}
    Let $A$ and $B$ be propositions, $t$ a proof, and $u_1$ and $u_2$ and two closed proofs of $A$. Then

  If 
  $x:A \vdash t:B$ 
  then  
  $$
    t\{u_1 \plus u_2\} \sim t\{u_1\} \plus t\{u_2\}
    \qquad\qquad\textrm{and}\qquad\qquad
    t\{a \bullet u_1\} \sim a \bullet t\{u_1\}  
  $$
\end{lem}
\begin{proof}
  Let $C \in {\mathcal V}$ and $c$ be a proof such that
  $\_:B \vdash c:C$. Then applying \autoref{linearity} to
  the proof $c\{t\}$ we get 
  $$
    c\{t\{u_1 \plus u_2\}\} \equiv c\{t\{u_1\}\} \plus c\{t\{u_2\}\}
    \qquad\qquad\textrm{and}\qquad\qquad
    c\{t\{a \bullet u_1\}\} \equiv a \bullet c\{t\{u_1\}\}
  $$
  and applying it again to the proof $c$ we get
  $$
    c\{t\{u_1\} \plus t\{u_2\}\} \equiv c\{t\{u_1\}\} \plus c\{t\{u_2\}\}
    \qquad\qquad\textrm{and}\qquad\qquad
    c\{a \bullet t\{u_1\}\} \equiv a \bullet c\{t\{u_1\}\}
  $$
  Thus
  $$
    c\{t\{u_1 \plus u_2\}\} \equiv c\{t\{u_1\} \plus t\{u_2\}\}
    \qquad\qquad\textrm{and}\qquad\qquad
    c\{t\{a \bullet u_1\}\} \equiv c\{a \bullet t\{u_1\}\}
  $$
  that is 
  $$
    t\{u_1 \plus u_2\} \sim t\{u_1\} \plus t\{u_2\}
    \qquad\qquad\textrm{and}\qquad\qquad
    t\{a \bullet u_1\} \sim a \bullet t\{u_1\}
    \mbox{\qedhere}
    $$ 
\end{proof}

\begin{proof}[Proof of \autoref{corollary2}]
 Using \autoref{cor:cor1}, we have
  $$
    t~(u_1 \plus u_2) \lras t'\{u'_1 \plus u'_2\}
    \sim t'\{u'_1\} \plus t'\{u'_2\} \llas (t~u_1) \plus (t~u_2)
    $$
hence
  \(
    t~(u_1 \plus u_2) \sim (t~u_1) \plus (t~u_2)
    \).
And
  $$
    t~(a\bullet u_1) \lras t'\{a \bullet u'_1\}
    \sim a \bullet t'\{u'_1\} \llas a\bullet (t~u_1)
  $$
hence
  \(
    t~(a\bullet u_1) \sim a \bullet (t~u_1)
  \).
  \qedhere
\end{proof}

\end{document}